\PassOptionsToPackage{dvipsnames,table}{xcolor}
\documentclass[acmsmall,screen, nonacm ]{acmart}

\usepackage[normalem]{ulem}

\usepackage{paralist}
\usepackage{marvosym}
\usepackage{multirow}
\usepackage{float}
\usepackage{booktabs}
\usepackage{algorithm}
\usepackage[noend]{algpseudocode}
\usepackage{nth}
\usepackage{fontawesome}
\usepackage[inline]{enumitem}
\usepackage{xspace}
\usepackage{subcaption}

\usepackage{bm}

\usepackage{tikz}
\usetikzlibrary{shapes.geometric, arrows}

\tikzset{
  to*/.style={
    shorten >=.25em,#1-to,
    to path={-- node[inner sep=0pt,at end,sloped] {${}^*$} (\tikztotarget) \tikztonodes}
  },
  to*/.default=
}

\usepackage{paralist}
\usepackage{booktabs}
\usepackage{combelow}
\usepackage[utf8]{inputenc}
\usepackage[T1]{fontenc}
\usepackage{microtype}
\usepackage{verbatim}
\usepackage{siunitx}
\usepackage{amsmath}
\usepackage{amscd}
\usepackage{array}
\usepackage{mathtools}
\usepackage{stmaryrd}
\usepackage{amsthm}
\usepackage{mathpartir}
\usepackage{newfloat}
\usepackage{adjustbox}
\usepackage{hyperref}
 \usepackage[
  capitalize,
  noabbrev,
]{cleveref}

\usepackage{csquotes}

\makeatletter

\RequirePackage{xcolor}
\RequirePackage{forloop}
\RequirePackage{xifthen}
\RequirePackage{keyval}

\newcommand*\printboolean[1]{\ifthenelse{\boolean{#1}}{true}{false}}

\newcommand*\copyboolean[2]{\ifthenelse{\boolean{#2}}{\setboolean{#1}{true}}{\setboolean{#1}{false}}}

\newcommand*\parseboolean[3]{\ifthenelse{\equal{#2}{true} \OR \equal{#2}{false}}{\setboolean{#1}{#2}
  }{\KV@err{Invalid value ``#1'' for key ``#3''.
      Allowed values are ``true'' and ``false''}}}

\newcommand*\jasminfamily{} \newcommand*\jasminfontsize{} \newcommand*\jasminfont{\jasminfamily\jasminfontsize}

\definecolor{jasmindname}{HTML}{0000FF}
\definecolor{jasminkw}{HTML}{008000}
\definecolor{jasmintype}{HTML}{B00040}
\definecolor{jasminstorageclass}{HTML}{B00040}
\definecolor{jasmincomment}{HTML}{408080}
\definecolor{jasminprimitive}{HTML}{000000}
\definecolor{jasminannotation}{HTML}{7D9029}
\definecolor{jasminconstant}{HTML}{801A66}
\definecolor{jasminhighlight}{HTML}{CCFFFF}

\newcommand*\jasmindname[1]{\textcolor{jasmindname}{#1}}
\newcommand*\jasminkw[1]{\textcolor{jasminkw}{\textbf{#1}}}
\newcommand*\jasmintype[1]{\textcolor{jasmintype}{#1}}

\newcommand*\jasmincomment[1]{\textcolor{jasmincomment}{\textit{#1}}}

\newcommand*\jasminconstant[1]{\textcolor{jasminconstant}{#1}}

\newcommand*\jasminopenbrace{\{}
\newcommand*\jasminclosebrace{\}}

\newlength\jasmintab

\newcounter{jasmin@indent}
\newcommand*\jasminindent[1]{\hspace*{1ex}\forloop{jasmin@indent}{0}{\value{jasmin@indent} < #1}{\hspace*{\jasmintab}}}

\newcounter{jasmin@lineno}
\newcommand*\jasminlineno@box{\makebox[1.62ex][r]{\scriptsize\arabic{jasmin@lineno}}\kern2pt}

\newboolean{jasminlineno@print}
\newboolean{jasminlineno@global}

\newcommand*\parselineno[1]{\parseboolean{jasminlineno@print}{#1}{lineno}}

\newcommand*\jasminlineno{\refstepcounter{jasmin@lineno}\ifthenelse{\boolean{jasminlineno@print}}{\jasminlineno@box }{\ifthenelse{\boolean{jasminlineno@global}}{\phantom{\jasminlineno@box}}{}}}

\newcommand*\jasminnewline@before{}
\define@key{jasminnewline}{before}{\renewcommand*\jasminnewline@before{#1}}

\define@key{jasminnewline}{lineno}{\parselineno{#1}}

\newcommand*\jasminnewline[1][]{\copyboolean{jasminlineno@global}{jasminlineno@print}\setkeys{jasminnewline}{before=,#1}\newline \jasminnewline@before \jasminlineno \copyboolean{jasminlineno@print}{jasminlineno@global}}

\newboolean{jasmincode@frame} \newcommand*\jasmincode@innerwidth{} \newcommand*\jasmincode@outerwidth{} \newcommand*\jasmincode@outerpos{} 

\define@key{jasmincode}{lineno}{\parselineno{#1}}
\define@key{jasmincode}{frame}{\parseboolean{jasmincode@frame}{#1}{frame}}
\define@key{jasmincode}{family}{\renewcommand*{\jasminfamily}{#1}}
\define@key{jasmincode}{fontsize}{\renewcommand*{\jasminfontsize}{#1}}
\define@key{jasmincode}{innerwidth}{\renewcommand*{\jasmincode@innerwidth}{#1}}
\define@key{jasmincode}{outerwidth}{\renewcommand*{\jasmincode@outerwidth}{#1}}
\define@key{jasmincode}{outerpos}{\renewcommand*{\jasmincode@outerpos}{#1}}
\define@key{jasmincode}{startingno}{\setcounter{jasmin@lineno}{#1}}

\newcommand*\jasmincode@rule{\rule{\columnwidth}{0.4pt}}

\newcommand*\jasmincode@drawframe{\ifthenelse{\boolean{jasmincode@frame}}{\jasmincode@rule }{\phantom{\jasmincode@rule}}}

\newenvironment*{jasmincode}[1][]{\setkeys{jasmincode}{frame=true,lineno=true,startingno=1,family=\fontfamily{lmtt},fontsize=\footnotesize,outerwidth=\columnwidth,innerwidth=\columnwidth,outerpos=,#1}\setlength{\parindent}{0pt}\addtocounter{jasmin@lineno}{-1}\begin{minipage}[\jasmincode@outerpos]{\jasmincode@outerwidth}
    \vspace{0pt}\centering \begin{minipage}{\jasmincode@innerwidth}
\jasmincode@drawframe 

      \raggedright \jasminfont \let\\=\jasminnewline \jasminlineno }{

  \raisebox{1ex}{\jasmincode@drawframe}\end{minipage}\end{minipage}}

\newcommand*\jazz[2][]{\setkeys{jasmincode}{family=\fontfamily{lmtt},fontsize=\small,#1}{\jasminfont \mbox{#2}}}

\@ifpackageloaded{hyperref}{
  \hypersetup{hypertexnames=false}
}{}

\@ifpackageloaded{cleveref}{
  \crefname{jasmin@lineno}{line}{Line}
  \crefformat{jasmin@lineno}{line~#1}
  \Crefformat{jasmin@lineno}{Line~#1}
  \crefrangeformat{jasmin@lineno}{lines~#1 to~#2}
  \Crefrangeformat{jasmin@lineno}{Lines~#1 to~#2}
  \crefmultiformat {jasmin@lineno}{lines~#2#1#3}{ and~#2#1#3}{, #2#1#3}{ and~#2#1#3}
  \Crefmultiformat {jasmin@lineno}{Lines~#2#1#3}{ and~#2#1#3}{, #2#1#3}{ and~#2#1#3}
}{}

\makeatother

\newif\ifNotes\Notestrue

\ifNotes
    \newcommand{\colorcomment}[2]{\leavevmode\unskip\space{\color{#1}#2}\xspace\PackageWarning{authorcomments}{Comment:{#2}}}
    \newcommand{\taggedcolorcomment}[3]{\colorcomment{#1}{[\textbf{#2}: #3]}}
\else
\newcommand{\colorcomment}[2]{\leavevmode\unskip\relax\PackageWarning{authorcomments}{Comment:{#2}}}
\newcommand{\taggedcolorcomment}[3]{\leavevmode\unskip\relax}
\fi

\newlist{requirelist}{enumerate}{1} \setlist[requirelist]{label=\normalfont\bfseries(\roman*),ref=(\roman*) from Theorem \thetheorem}
\crefname{requirelisti}{requirement}{requirements}
\Crefname{requirelisti}{Requirement}{Requirements}

\makeatletter
\newcommand*{\step}{\@ifstar{\step@multistep}{\step@onestep}}
\newcommand*{\step@multistep}[2]{\mathrel{\raisebox{-0.15em}{\ensuremath{{\xrightarrow[{#1}]{\smash{\raisebox{-0.08em}{\ensuremath{{\scriptstyle#2}}}}}\!\!\vphantom{\to}^{\ast}}}}}}
\newcommand*{\step@onestep}[2]{\mathrel{\raisebox{-0.15em}{\ensuremath{{\xrightarrow[{#1}]{\smash{\raisebox{-0.08em}{\ensuremath{{\scriptstyle#2}}}}}}}}}}
\makeatother

\makeatletter
\newcommand*{\sem}{\@ifstar{\sem@multistep}{\sem@onestep}}
\newcommand*{\sem@multistep}[4]{\ensuremath{{#1 \step@multistep{#2}{#3} #4}}}
\newcommand*{\sem@onestep}[4]{\ensuremath{{#1 \step@onestep{#2}{#3} #4}}}
\makeatother

\makeatletter
\newcommand*{\aux@twoargs}[2]{(\ifx&#1\else{#1}, \fi{#2})}
\newcommand*{\Tobsname@single}{\ensuremath{{T}}}
\newcommand*{\Tobsname@multi}{\ensuremath{{T}}}
\newcommand*{\Tobsname}{\@ifstar{\Tobsname@multi}{\Tobsname@single}}
\newcommand*{\Tobs@single}[2]{\ensuremath{{\Tobsname@single\aux@twoargs{#1}{#2}}}}
\newcommand*{\Tobs@multi}[2]{\ensuremath{{\Tobsname@multi\aux@twoargs{#1}{#2}}}}
\newcommand*{\Tobs}{\@ifstar{\Tobs@multi}{\Tobs@single}}
\makeatother

\NewCommandCopy{\mathphi}{\phi}
\renewcommand*{\phi}{\ensuremath{{\mathphi}}}

\definecolor{rosewater}{HTML}{dc8a78}
\definecolor{flamingo}{HTML}{dd7878}
\definecolor{pink}{HTML}{ea76cb}
\definecolor{mauve}{HTML}{8839ef}
\definecolor{red}{HTML}{d20f39}
\definecolor{maroon}{HTML}{e64553}
\definecolor{peach}{HTML}{fe640b}
\definecolor{yellow}{HTML}{df8e1d}
\definecolor{green}{HTML}{40a02b}
\definecolor{teal}{HTML}{179299}
\definecolor{sky}{HTML}{04a5e5}
\definecolor{sapphire}{HTML}{209fb5}
\definecolor{blue}{HTML}{1e66f5}
\definecolor{lavender}{HTML}{7287fd}
\definecolor{text}{HTML}{4c4f69}
\definecolor{subtext1}{HTML}{5c5f77}
\definecolor{subtext0}{HTML}{6c6f85}
\definecolor{overlay2}{HTML}{7c7f93}
\definecolor{overlay1}{HTML}{8c8fa1}
\definecolor{overlay0}{HTML}{9ca0b0}
\definecolor{surface2}{HTML}{acb0be}
\definecolor{surface1}{HTML}{bcc0cc}
\definecolor{surface0}{HTML}{ccd0da}
\definecolor{base}{HTML}{eff1f5}
\definecolor{mantle}{HTML}{e6e9ef}
\definecolor{crust}{HTML}{dce0e8}

\definecolor{shadecolor}{named}{base}

\colorlet{lightgreen}{green!70}
\colorlet{lightyellow}{peach!70}

\usepackage{pifont}
\usepackage{makecell}

\definecolor{darkgreen}{HTML}{dce0e8}

\makeatletter
\newcommand{\nonanon}[2]{\if@ACM@anonymous{#2}\else{#1}\fi}
\makeatother

\newif\ifAppendix \Appendixtrue 

\makeatletter
\begin{document}

\title{Decompiling for Constant-Time Analysis}

\author{Santiago Arranz-Olmos}
\orcid{0009-0007-7425-570X}
\affiliation{\institution{MPI-SP}
  \city{Bochum}
  \country{Germany}
}
\email{santiago.arranz-olmos@mpi-sp.org}

\author{Gilles Barthe}
\orcid{0000-0002-3853-1777}
\affiliation{\institution{MPI-SP}
  \city{Bochum}
  \country{Germany}
}
\affiliation{\institution{IMDEA Software Institute}
  \city{Madrid}
  \country{Spain}
}
\email{gilles.barthe@mpi-sp.org}

\author{Lionel Blatter}
\orcid{0000-0001-9058-2005}
\affiliation{\institution{MPI-SP}
  \city{Bochum}
  \country{Germany}
}
\email{lionel.blatter@mpi-sp.org}

\author{Youcef Bouzid}
\orcid{0009-0002-1275-1613}
\affiliation{\institution{ENS Paris-Saclay}
  \city{Gif-sur-Yvette}
  \country{France}
}
\email{youcef.bouzid@ens-paris-saclay.fr}

\author{Sören van~der~Wall}
\orcid{0009-0009-4781-8583}
\affiliation{\institution{TU Braunschweig}
  \city{Braunschweig}
  \country{Germany}
}
\email{s.van-der-wall@tu-bs.de}

\author{Zhiyuan Zhang}
\orcid{0009-0000-2669-5654}
\affiliation{\institution{MPI-SP}
  \city{Bochum}
  \country{Germany}
}
\email{zhiyuan.zhang@mpi-sp.org}

\begin{abstract}
The constant-time programming discipline is commonly used to protect
cryptographic libraries against side-channel attacks. However, it is
hard to write constant-time code; moreover, compilers can introduce
constant-time violations. Therefore, it is important to ensure that
assembly code is constant-time. One approach is to show that source
programs are constant-time, and that constant-timeness is preserved by
compilation. In this paper, we explore the methodological soundness
and scalability of the Decompile-then-Analyze approach, a less
conventional alternative that has been suggested in the broader
setting of static analysis. Informally, the Decompile-then-Analyze
approach uses decompilers a front-end for static analysis tools. As a
motivation for our study, we show that current decompilers eliminate
CT vulnerabilities before CT analysis, leading to non-CT programs
being accepted as constant-time. Independently, we provide
\emph{constructed} examples of non-CT, exploitable, programs that are
accepted by two popular CT analysis tools; in both cases the culprit
are program transformations that are used internally prior to CT
analysis and eliminate CT violations.  While our examples do not
invalidate the general approach of these tools, they emphasize the
need for studying the Decompile-then-Analyze approach.

On the methodological side, we define the notion of \emph{CT
transparency}. Informally, a program transformation is CT transparent
if does not eliminate nor introduce CT violations. We also provide
general methods for proving that a transformation is CT transparent,
and show that several transformations of interest are transparent. We
also sketch an extension of CT transparency to speculative
constant-time, which is used by cryptographic software as a
protection against Spectre attacks.

On the practical side, we build a CT-transparent version of the
popular LLVM-based decompiler \textsc{RetDec}{}, and combine it with CT-LLVM,
an existing CT verification tool for LLVM. We evaluate the resulting
tool, called \textsc{CT-RetDec}{} on a benchmark set of real-world
vulnerabilities in binaries, and show that the modifications had
significant impact on how well \textsc{CT-RetDec}{} performs.

\end{abstract}

\maketitle

\section{Introduction}
Decompilers\footnote{See \url{https://decompilation.wiki/} for an overview and pointers to the literature.}
are routinely used in vulnerability analysis to transform binary
programs into source or IR programs on which (manual) analysis can be
carried out. To maximize their benefits, decompilers aim to produce source
or IR programs that are both readable and correctly capture the
behavior of their corresponding binary programs. Of course, achieving
correctness and readability simultaneously is intrinsically hard. Yet,
in spite of the challenge, there has been significant progress towards
this goal, through a combination of technical
developments~\cite{BrumleyLSW13,YakdanEGS15,schulte2018evolving,DBLP:conf/ccs/GussoniFFA20,BurkPKV22,BasqueBGOMBDS024,DBLP:conf/uss/ZouKWGBT24},
and extensive
evaluations~\cite{DBLP:conf/uss/AndriesseCVSB16,yakdan2016helping,dewolf,DBLP:conf/issta/LiuW20,DBLP:conf/acsac/MatteiMKV22,DBLP:conf/issta/CaoZL024,DramkoLSVG24}.

More recently, researchers have started to explore the possibility to
use decompilation to conduct static analysis for properties such as memory
layout and memory safety
violations~\cite{DBLP:conf/sp/LiuYWB22,DBLP:conf/asiaccs/MantovaniCSB22,DBLP:conf/asplos/ZhouYHCZ24}. The
idea is simple: take a binary program, use a decompiler to produce an
IR or source program, and finally run a static analysis on the source
or IR\@. While less direct, this approach offers complementary benefits
over binary-level analysis, in situations where source-level analysis
tools are more precise, more scalable, or offer features not supported
by their binary-level counterparts.

\paragraph*{Problem Statement}
In this paper, we consider the question of using decompilation for
constant-time
analysis---see~\cite{BarbosaBBBCLP21,DBLP:conf/sp/JancarFBSSBFA22,GeimerVRDBM23}
for an overview of the field.  Constant-time (CT) analysis aims to
ensure that programs are protected against timing side-channel
attacks, in which
the attacker recovers secrets by observing the execution of a victim program.
The prominent timing side-channel attacks observe the accessed memory locations \cite{bernstein2005cache} and the control flow of the victim program~\cite{DBLP:conf/icisc/MolnarPSW05,Koc96}.
The constant-time policy mitigates these side-channel attacks:
It mandates that the program does not perform secret dependent memory accesses and does not branch on secrets.

The main reasons to study decompilation for constant-time analysis are:
\begin{enumerate}
\item The class of attacks mitigated by the constant-time policy, namely
  timing side-channel attacks, is devastating.
  In particular, they allow (possibly remote) attackers to recover cryptographic keys very
  efficiently~\cite{Koc96,bernstein2005cache}.
\item The constant-time policy, despite its relative simplicity, is very hard to achieve, even for
  expert cryptographic developers.
\item Many popular CT analysis tools\footnote{See
\url{https://crocs-muni.github.io/ct-tools/} for a list of tools.}
  operate at source or IR level.
These tools can not be used to analyze binaries without a decompiler.
\item Mainstream compilers introduce violations of the CT policy (CT violations) to binary
  programs~\cite{DBLP:conf/eurosp/SimonCA18,DanielBR20,schneider2024breaking,geimer2025fun}, and there is little prospect that the issue will
  be fixed in the near term.
  These CT violations are undetectable in the source code.
\end{enumerate}
Therefore, using decompilation for constant-time analysis has the potential
to improve a pressing and as yet unresolved security problem. In order
to evaluate this potential, our work addresses the following research
questions:
\begin{description}
\item[RQ1] Are existing decompilers suitable for being used in
  combination with constant-time analysis?

\item[RQ2] Are there guiding principles to make decompilers suitable
  to constant-time analysis?

\item[RQ3] Is there a practical decompiler that adheres to these
  principles and can be used for constant-time analysis of real-world
  code bases?
\end{description}
We answer \textbf{RQ1} negatively by exhibiting a real-world binary
program that contains an exploitable CT violation. The program is 
decompiled by RetDec, a state-of-the-art LLVM decompiler, into a
constant-time LLVM IR program that is (rightly) proved secure by LLVM
CT analysis tools. Our binary program stems from the reference
implementation of ML-KEM, a recently standardized key encapsulation
mechanism that is provably secure against quantum adversaries. However, it is
vulnerable to the Clangover (CVE-2024-37880)
compiler-induced side-channel vulnerability.

We answer \textbf{RQ2} positively by introducing the notion of
CT transparent transformations. Informally, a transformation is
CT transparent if it neither introduces nor eliminates CT violations,
generating no false positives or false negatives in the analysis
results.
We develop rigorous techniques for proving that program
transformations are CT transparent by extending CT simulations,
originally developed in~\citet{DBLP:conf/csfw/BartheGL18} to prove
constant-time preservation, which is a weaker property than CT transparency.
Specifically, we identify a condition on simulations,
called \emph{\mbox{\ensuremath{{\mathcal{PC}}}-injectivity}}, which allows us
to reuse existing simulations from~\citet{DBLP:conf/csfw/BartheGL18}
and related works, to prove CT transparency with little additional 
cost. Using our techniques, we show that many common program
transformations are CT transparent; for others, we provide simple
examples to explain where and why they fail to be transparent.

We answer \textbf{RQ3} positively by developing a toolchain, called
\textsc{CT-RetDec}{}, that combines a modified version of the decompiler \textsc{RetDec}{}
and CT-LLVM, a CT analysis tool for LLVM
IR. We then use \textsc{CT-RetDec}{} to analyze binary code to find 
CT violations.

As a contribution of independent interest, we observe that \textsc{CT-Verif}{}~\cite{DBLP:conf/uss/AlmeidaBBDE16}
and \textsc{BinSec}{}~\cite{DanielBR20}, two popular CT analysis tools, use non-CT transparent
converters to transform input programs into another representation on
which the CT analysis is performed. In both cases, we can leverage
non-transparency of converters to craft \emph{constructed} non-CT
examples that are accepted by the tools. While our examples do not
invalidate the general approach of these tools for real-world code, we
recommend that CT analysis tools rely on transparent converters, and
that such converters are clearly exposed in the implementation.

\paragraph*{Summary of Contributions}
In summary, our main contributions are:
\begin{itemize}[topsep=0pt]
\item A demonstration that state-of-the-art decompilers are not transparent and cannot be used
  for constant-time analysis;
\item A formalization of transparency and techniques to prove
  that a transformation is transparent;
\item An analysis of common program transformations, providing for
  each of them a proof of transparency, or a counterexample to
  transparency;
\item A toolchain, called \textsc{CT-RetDec}{}, for analyzing binary code against the CT policy, 
    built with transparent decompilation in mind; and
\item A study of the transparency of converters employed in CT analysis tools.
\end{itemize}

\paragraph*{Concurrent Approaches to CT Analysis}
We refer to the approach of applying CT analysis after decompilation as
the Decompile-then-Analyze approach. 
This approach has been adopted implicitly as an internal component in 
prior CT analysis tools, such as CacheS~\cite{0011BL0ZW19}, and has also 
been used explicitly as an external front-end, for example by decompiling 
binaries with Ghidra when analyzing cryptographic protocol implementations~\cite{abs-2511-11385}.

There are two concurrent approaches to deliver CT guarantees for binary code
next to the Decompile-then-Analyze approach:
Analyze-then-Compile and Binary Analysis.
Analyze-then-Compile performs source- or IR-level analysis in order to prove that the program does not violate the CT policy,
and then compiles the program into the binary.
However, it is well-known that secure source programs can be compiled into binary programs that contain CT violations,
leading to a dangerous compiler security gap~\cite{DBLP:conf/sp/DSilvaPS15,DBLP:conf/eurosp/SimonCA18,DBLP:conf/uss/XuLDDLWPM23}.
Secure compilation aims to address this gap by integrating security
considerations into compiler
developments~\cite{DBLP:conf/csfw/AbateB0HPT19,DBLP:journals/toplas/AbateBCDGHPTT21}. A
specific line of work within secure compilation focuses on preserving
CT by compilation~\cite{DBLP:conf/csfw/BartheGL18,DBLP:journals/pacmpl/BartheBGHLPT20,DBLP:conf/ccs/BartheGLP21}. However, as stated above, there is little
prospect of mainstream compilers to adopt these CT preservation guarantees.

The Binary Analysis approach performs CT analysis on the binary-level, using
e.g., \textsc{BinSec}{}~\cite{DanielBR20}. However, many approaches to CT analysis
depend on general-purpose analyses, e.g., alias analysis, that are
easier to perform at IR level. Furthermore, CT analysis may
involve complex reasoning that goes beyond automated
methods, e.g., to justify declassifying some information, or to
reason about leakage in refined models. For such cases, binary-level
analysis is not an option.\footnote{Admittedly, our paper focuses on the
baseline CT policy and does not deal with declassification nor leakage
in refined models, and leaves these extensions for future work.}

\paragraph*{Paper Structure}
\Cref{sec:motivating-example} presents a motivating example and
empirical test results, demonstrating that transparency failures in
decompilers are a systematic issue.
\Cref{sec:setting,sec:proof-techniques} give a formal definition of
\ensuremath{{\text{CT}}}{} transparency and a technique for proving transparency.
\Cref{sec:passes} investigates the transparency of common
transformations.
Then, \cref{sec:ct-retdec} introduces \textsc{CT-RetDec}{} and evaluates its
performance in finding \ensuremath{{\text{CT}}}{} violations.
\Cref{sec:cttools} presents two CT analysis tools that fail
transparency on constructed input binaries.
Finally, \cref{sec:speculation} extends the formalism to speculative
constant-time.

\paragraph*{Artifact}
We present an artifact containing a mechanization of the general version
of \cref{thm:ct:soundness} and the experiments described in the paper in
\cite{secdec-artifact}.

{}

\section{Motivating Example and Impact}\label{sec:motivating-example}

In this section, we investigate the Decompile-then-Analyze approach on a
concrete real-world side-channel vulnerability to motivate our study of CT
transparency. We discover that decompilation removes the vulnerability from
the program, so that subsequent CT analysis (correctly) reports the absence of any
vulnerability. We then empirically test five state-of-the-art decompilers on
artificial, minimal counterexamples to demonstrate that this issue is not
limited to a single decompiler, but happens systematically throughout practical
tools.

\subsection{Clangover: A Vulnerability Hidden by Decompilation}

Clangover is a real-world side-channel vulnerability (CVE-2024-37880),
that exists in the binary program of the reference implementation of ML-KEM,
the standardized post-quantum key encapsulation mechanism.
The vulnerability is in the \jazz{poly\_frommsg(r, msg)}
function that turns a secret message (\jazz{msg})
into the coefficients of a polynomial bit by bit.
To do so, it loops over \jazz{msg},
and for each bit, it sets the corresponding coefficient of the polynomial
pointed by \jazz{r} to either zero or a constant.
The left-hand side of~\Cref{lst:clangover}
displays the vulnerable excerpt of the binary code of \jazz{poly\_frommsg} in
\texttt{x86-64} syntax.
It is the \jazz{for}-loop that iterates over each byte's
bits, extracts the bit, and sets the coefficient for that bit. Line~4 loads
the current byte of \jazz{msg} from memory and Line~6 extracts the
current bit from it. Line~7 branches on the
extracted bit in order to write the correct coefficient to memory in Line~10.
The conditional branch on the secret bits of \jazz{msg}
constitutes the CT violation: it is a branching condition that depends on a secret.

\begin{figure}
  \begin{jasmincode}[fontsize=\scriptsize,outerwidth=36ex,outerpos=t]
    \jasminindent{0}\jasmincomment{; rsi:msg - rax:msg\_offset - r8d:byte}\\
    \jasminindent{0}\jasmincomment{; rcx:j - rdx:coef - rdi:poly\_offset}\\
    \jasminindent{0}.LBB0\_2: \jasmincomment{; for(j = 0; j < 8; j++)}\\
    \jasminindent{1}movzx   r8d, byte ptr [rsi + rax]\\
    \jasminindent{1}xor     edx, edx\\
    \jasminindent{1}bt      r8d, ecx\\
    \jasminindent{1}jae     .LBB0\_4 \jasmincomment{; leaks bit of msg}\\
    \jasminindent{1}mov     edx, \jasminconstant{1665}\\
    \jasminindent{0}.LBB0\_4:\\
    \jasminindent{1}mov     word ptr [rdi + \jasminconstant{2}*rcx], dx\\
    \jasminindent{1}inc     rcx\\
    \jasminindent{1}cmp     rcx, \jasminconstant{8}\\
    \jasminindent{1}jne     .LBB0\_2
  \end{jasmincode}\hfill \begin{jasmincode}[fontsize=\scriptsize,outerwidth=50ex,outerpos=t]
    \jasminindent{0}\\
    \jasminindent{0}\\
    \jasminindent{0}\jasminkw{for} (\jasmintype{int64\_t} j = \jasminconstant{0}; j < \jasminconstant{8}; j++) \jasminopenbrace{}\\
    \jasminindent{1}byte = *(\jasmintype{char*})(msg\_offset + msg);\\
    \jasminindent{0}\\
    \jasminindent{0}\\
    \jasminindent{1}coef = (\jasminconstant{1} <{}< (\jasmintype{int32\_t})j \% \jasminconstant{32} \& (\jasmintype{int32\_t})byte) == \jasminconstant{0} ?\\
    \jasminindent{3}\jasminconstant{0} : \jasminconstant{1665};       \jasmincomment{// no leak}\\
    \jasminindent{0}\\
    \jasminindent{1}*(\jasmintype{int16\_t*})(\jasminconstant{2} * j + poly\_offset) = coef;\\
    \jasminindent{0}\\
    \jasminindent{0}\\
    \jasminindent{0}\jasminclosebrace{}
  \end{jasmincode}
  \caption{Decompiling Clangover removes the CT vulnerability.}\label{lst:clangover}\end{figure}

To employ the Decompile-the-Analyze approach,
we feed the binary program into the off-the-shelf decompiler \textsc{RetDec}{},
a state-of-the-art decompiler that emits LLVM IR\@.
We then analyze the decompiled program using \textsc{CT-Verif}{},
a state-of-the-art CT analysis tool for LLVM IR~\cite{DBLP:conf/uss/AlmeidaBBDE16}.
The corresponding decompiled \jazz{for}-loop is shown on the right-hand side of
\Cref{lst:clangover}, where we manually recovered sensible variable names
for readability. Critically, the decompiler has replaced the conditional
branch on the secret bits by a conditional move instruction using the \jazz{?} operator (Line~7).
This removes the CT violation, as conditional move instructions do not
leak their conditional~\cite{DOIT}.
Indeed, \textsc{CT-Verif}{} correctly identifies that the decompiled program is CT\@.
Hence, using \textsc{RetDec}{} for the Decompile-then-Analyze approach misses CT
vulnerabilities.

\subsection{Transparency Failures in Decompilers}\label{sec:decompiler-eval}

The previous section highlights the risk of the Decompile-then-Analyze approach
on a singular binary program and decompiler.
This naturally raises the question of whether it is an isolated incident, or
whether this is a systematical issue in decompilers.
We address this question by empirically testing state-of-the-art decompilers
on carefully crafted non-CT binary code snippets. They are minimal examples
that contain CT violations.
The CT violations leak a secret through either a dead load/store instruction, or a spurious or simple branch.
We defer the presentation of the examples and the offending transformations to \Cref{sec:nontranspasse}.
Our test spans five decompilers: \textsc{Angr} 9.2.127~\cite{angr}, \textsc{BinaryNinja}
4.1.5747~\cite{binaryninja}, \textsc{Ghidra} 11.2~\cite{ghidra}, \textsc{Hex-Rays}
8.4.0.240320~\cite{hexrays} and \textsc{RetDec}{} 5.0~\cite{retdec} using the
Decompiler Explorer~\cite{dogbolt}.
In \Cref{tab:passes}, we report if a decompiler removes or keeps the CT violation. 
The result is that each of the five decompilers removes CT violations.
Therefore, the issue is indeed systematic throughout practical decompilers.

\begin{table}
  \caption{Analysis of five decompilers. A \textcolor{red}{\Lightning}{} (resp.\ -{})
    means the decompiler removes (resp.\ keeps) the CT violation.
}\label{tab:case_study}
  \centering\small \footnotesize
  \resizebox{\linewidth}{!}{
  \begin{tabular}{ccccccl}
    \toprule
    Example
    & \multicolumn{1}{l}{\textsc{Angr}}
    & \multicolumn{1}{l}{\textsc{BinaryNinja}}
    & \multicolumn{1}{l}{\textsc{Ghidra}}
    & \multicolumn{1}{l}{\textsc{Hex-Rays}}
    & \multicolumn{1}{l}{\textsc{RetDec}}
    & Root cause \\
    \midrule
    \cref{lst:clangover}
    & \textcolor{red}{\Lightning}{}
    & -
    & -
    & -
    & \textcolor{red}{\Lightning}{}
   & If Conversion\\
   \cref{lst:branch_coalescing}
    & \textcolor{red}{\Lightning}{}
    & \textcolor{red}{\Lightning}{}
    & -
    & \textcolor{red}{\Lightning}{}
    & \textcolor{red}{\Lightning}{}
   & Branch Coalescing \\
  \cref{lst:empty_branch_coalescing}
   & -
    & \textcolor{red}{\Lightning}{}
    & \textcolor{red}{\Lightning}{}
    & \textcolor{red}{\Lightning}{}
    & \textcolor{red}{\Lightning}{}
   & Empty Branch Coalescing \\
 \cref{lst:dead_load}
    & \textcolor{red}{\Lightning}{}
    & \textcolor{red}{\Lightning}{}
    & \textcolor{red}{\Lightning}{}
    & \textcolor{red}{\Lightning}{}
    & \textcolor{red}{\Lightning}{}
 &  Dead Load Elimination\\
 \cref{lst:self_store}
    & -
    & -
    & -
    & -
    & \textcolor{red}{\Lightning}{}
  & Dead Store Elimination\\
    \bottomrule
  \end{tabular}
  }
\end{table}

\section{Transparency}\label{sec:setting}
This section reviews the constant-time property using an
abstract model of computation and introduces the notion of transparent
transformations.

We keep the programming language abstract and say that 
\ensuremath{{\mathcal{L}}}{} is a set of \emph{programs}.  The semantics of a program is given
by transitions on \emph{states} \ensuremath{{\mathcal{S}}}{} that capture, for instance,
the values of registers and memory and the current program
point, at a point in time.  These transitions constitute a relation
\ensuremath{{{\xrightarrow{}} \subseteq \ensuremath{{\mathcal{S}}} \times \ensuremath{{\mathcal{S}}}}}, where we
write \ensuremath{{s \xrightarrow{} s'}} to indicate that a state~\(s\)
transitions to state~\(s'\). The language also specifies which states
are \emph{final}, i.e., states where the program has terminated. We
assume that the semantics is deterministic, i.e., if \ensuremath{{s
  \xrightarrow{} s'}} and \ensuremath{{s \xrightarrow{} s''}} then \ensuremath{{s' =
  s''}}, and safe, i.e., for every $s$, either $s$ is final, or
there exists $s'$ such that $\ensuremath{{s \xrightarrow{} s'}}$.

\paragraph*{Leakage}
We extend this model to capture side-channel leakage by instrumenting
transitions with \emph{observations} (denoted \ensuremath{{\mathcal{O}}}{}): the
transition \sem{s}{}{o}{s'} indicates that going from \(s\) to \(s'\)
emits the observation~\(o \in \ensuremath{{\mathcal{O}}}\).
Observations represent what is leaked by the transition, i.e., what an
attacker can learn.
For example, in the constant-time leakage model, a transition that
performs a memory load emits an observation that contains the address of
the load, and a conditional branch emits an observation containing its
condition.
These values are leaked because an attacker who observes the data and instruction caches can compute them, respectively.

In this work, we assume that observations reveal the control flow of a program~\cite{DBLP:conf/icisc/MolnarPSW05}.
To do so, we assume that every state \(s\) has a \emph{program point}
(we write \ensuremath{{\mathcal{PC}}}{} for the set of program points), denoted \ensuremath{{\ensuremath{{\mathsf{pc}}}({s})}}.
For example, the program point in an assembly language is the program
counter, and in a structured language it is the remaining code to be
executed.
Formally, the assumption that observations reveal the program point means:
for every two states at the same program point
\ensuremath{{\ensuremath{{\ensuremath{{\mathsf{pc}}}({s_1})}} = \ensuremath{{\ensuremath{{\mathsf{pc}}}({s_2})}}}}, that step with the
same observation \sem{s_1}{}{o}{s_1'} and
\sem{s_2}{}{o}{s_2'}, we have that the resulting states
are also at the same program point, i.e.,
\ensuremath{{\ensuremath{{\ensuremath{{\mathsf{pc}}}({s_1'})}} = \ensuremath{{\ensuremath{{\mathsf{pc}}}({s_2'})}}}}.

\newcommand*{\lblnil}{refl}
\newcommand*{\lblcons}{trans}
\newcommand*{\refnil}{\hyperref[fig:sem:multi-step]{\textsc{\lblnil}}}
\newcommand*{\refcons}{\hyperref[fig:sem:multi-step]{\textsc{\lblcons}}}
\begin{figure}
  \begin{gather*}
    \inferrule[\lblnil]{
    }{
      \sem*{s}{}{\ensuremath{{\epsilon}}}{s}
    }
    \qquad
    \inferrule[\lblcons]{
      \sem{s}{}{o}{s'}\\
      \sem*{s'}{}{\ensuremath{{\boldsymbol{o}}}}{s''}
    }{
      \sem*{s}{}{o \mathop{\cdot} \ensuremath{{\boldsymbol{o}}}}{s''}
    }
  \end{gather*}
  \caption{Multi-step execution semantics.}\label{fig:sem:multi-step}
\end{figure}

\paragraph*{Program Behavior}
We define the behavior of programs in terms of executions, which
comprise multiple steps.
\Cref{fig:sem:multi-step} defines executions as the reflexive,
transitive closure \sem*{s}{}{\ensuremath{{\boldsymbol{o}}}}{s'}
of~\({}\xrightarrow{}\).
We accumulate the observations leaked during the execution in a list, denoted with a bold~\ensuremath{{\boldsymbol{o}}}{}.
The length of \ensuremath{{\boldsymbol{o}}}{}, denoted \ensuremath{{\left|\ensuremath{{\boldsymbol{o}}}\right|}}, matches the number of steps in the execution.

A program $P\in \ensuremath{{\mathcal{L}}}$ starts its
execution with a list of \emph{input values}, taken from an input
domain~\ensuremath{{\mathcal{I}}}{}.
Inputs capture the interaction of the program with
its environment, e.g., the arguments given to the entry point of a
program. The inputs fully determine the initial state of the program,
thus we write \ensuremath{{P(i) \in \ensuremath{{\mathcal{S}}}}} for
the initial state of~\(P\) on inputs~\(i \in \ensuremath{{\mathcal{I}}}\).
The behavior of a program are the observations along all its executions,
as follows.

\begin{definition}[Program Behavior]\label{def:beh}
  The behavior \ensuremath{{\mathit{Beh}({P}, {i})}} of program~\(P\) on input~\(i\) is the set of
  leakage traces starting from~\(P(i)\), i.e.,
  \ensuremath{{
    \ensuremath{{\mathit{Beh}({P}, {i})}} \triangleq
    \{\ensuremath{{\boldsymbol{o}}}\mid\sem*{P(i)}{}{\ensuremath{{\boldsymbol{o}}}}{s}\}
  }}.
\end{definition}

\paragraph*{Constant-Time}
As usual with noninterference definitions, our security property is
based on an \emph{indistinguishability} relation
\ensuremath{{\phi \subseteq \ensuremath{{\mathcal{I}}} \times \ensuremath{{\mathcal{I}}}}} on inputs.
Indistinguishability expresses what part of the input is public.
For example, if the inputs to a program are a secret message \(m\) and
its public length \(n\), it is appropriate to define
\ensuremath{{(m, n) \mathrel{\mathphi} (m', n') \triangleq n = n'}}, as the secret
message is not known to an attacker, while the length of the message
is known.
The intuition is that an adversary cannot distinguish the two
inputs \ensuremath{{(m, n)}} and \ensuremath{{(m', n')}} if the public part coincides,
i.e., the lengths \(n\)~and~\(n'\) are the same.

A program \(P\) is \emph{constant-time} w.r.t.\ an indistinguishability
relation \phi{} (denoted \(P\) is \ensuremath{{\phi\text{-CT}}}) if the observations
generated by executions starting from indistinguishable inputs produce
the same observations.
We use the term \emph{\ensuremath{{\phi\text{-CT}}}{} violation} to refer to a difference in
the observations generated by a program.

\begin{definition}[\ensuremath{{\phi\text{-CT}}}]\label{def:ct}
  \(P\) is \ensuremath{{\phi\text{-CT}}}{} if for all pairs of inputs \ensuremath{{i_1 \mathrel{\mathphi} i_2}}
  we have \ensuremath{{\ensuremath{{\mathit{Beh}({P}, {i_1})}} = \ensuremath{{\mathit{Beh}({P}, {i_2})}}}}.
\end{definition}

\paragraph*{Transparency}
We now turn to the interactions between program transformation and
constant-time.
Consider a transformation \ensuremath{{\ensuremath{{\llparenthesis\,{\cdot}\,\rrparenthesis}} :
  \ensuremath{{\mathcal{L}}}_s \to \ensuremath{{\mathcal{L}}}_t}} that maps programs in an \emph{input}
language, denoted \(\ensuremath{{\mathcal{L}}}_s\), to programs in an \emph{output}
language, denoted \(\ensuremath{{\mathcal{L}}}_t\) (throughout the paper, we use these
subscripts to refer to input and output languages). We want to reason
about whether \ensuremath{{\llparenthesis\,{\cdot}\,\rrparenthesis}} introduces or removes \ensuremath{{\phi\text{-CT}}}{} violations.
We consider three relevant properties, as follows.

\begin{definition}[Reflection, Preservation, and Transparency]\label{def:ct:transparency}
  We say that a program transformation \ensuremath{{\ensuremath{{\llparenthesis\,{\cdot}\,\rrparenthesis}} : \ensuremath{{\mathcal{L}}}_s \to
    \ensuremath{{\mathcal{L}}}_t}} between an input language $\ensuremath{{\mathcal{L}}}_s$ and an output language $\ensuremath{{\mathcal{L}}}_t$
  \begin{itemize}
      \item Reflects \ensuremath{{\text{CT}}}{} if for each \(P\) and $\phi$,
    \ensuremath{{\llparenthesis\,{P}\,\rrparenthesis}} is \ensuremath{{\phi\text{-CT}}}{} implies \(P\) is \ensuremath{{\phi\text{-CT}}}{};
\item Preserves \ensuremath{{\text{CT}}}{} if for each $P$ and $\phi$,
    \(P\) is \ensuremath{{\phi\text{-CT}}}{} implies \ensuremath{{\llparenthesis\,{P}\,\rrparenthesis}} is \ensuremath{{\phi\text{-CT}}}{}; and
\item Is \ensuremath{{\text{CT}}}{} transparent if it both reflects and preserves \ensuremath{{\text{CT}}}{}.
  \end{itemize}
\end{definition}

Indeed, if $\ensuremath{{\llparenthesis\,{\cdot}\,\rrparenthesis}}$ reflects $\ensuremath{{\text{CT}}}{}$, it does not remove a $\ensuremath{{\phi\text{-CT}}}{}$ violation for any indistinguishability relation~$\phi$ and input program $P$.
Similarly, if $\ensuremath{{\llparenthesis\,{\cdot}\,\rrparenthesis}}$ preserves $\ensuremath{{\text{CT}}}{}$, it does not introduce any $\ensuremath{{\phi\text{-CT}}}{}$ violations.

\section{Proof Techniques}\label{sec:proof-techniques}
In this section, we present a method to prove the CT transparency of a
program transformation, drawing on existing work in the area of secure
compilation that targets CT preservation, for instance~\cite{DBLP:conf/csfw/BartheGL18,DBLP:journals/pacmpl/BartheBGHLPT20,DBLP:conf/ccs/BartheGLP21}.

We aim to establish a \emph{simulation} between the input program and
output program of the transformation.  A simulation links states of
the input and the output programs so that the transitions of the
output program can be mimicked by the input program.  In order to
preserve CT, this approach additionally relates input and output
observations: output observations must be expressible as a function of
input observations, which is called the observation transformer.  To
extend this approach to CT transparency, our simulations pose an
additional requirement on the observation transformers.

We start our presentation with a simplified version of our simulations
to build some intuition, and present the general version afterward.

\paragraph{Lock-Step Simulations for CT Transparency}
The basic form of simulations are \emph{lock-step} simulations,
where every step of the output program is mimicked by precisely one step
of the input program.

\begin{definition}[Lock-Step Simulation Diagram]\label{def:ct:lock-step-diagram}
  A relation \ensuremath{{{\sim} \subseteq \ensuremath{{\mathcal{S}}}_s \times \ensuremath{{\mathcal{S}}}_t}}
  satisfies a lock-step simulation diagram w.r.t.\ an observation
  transformer \ensuremath{{\Tobsname : \ensuremath{{\mathcal{O}}}_s \to \ensuremath{{\mathcal{O}}}_t}} if
  for every pair of related states \ensuremath{{s \sim t}} in which the output
  program steps \sem{t}{}{o_t}{t'}, then the input program
  steps \sem{s}{}{o_s}{s'} such that \ensuremath{{s' \sim t'}} and
  \ensuremath{{o_t = \Tobs{}{o_s}}}.

  \noindent \begin{minipage}{0.70\linewidth}
    In the diagram, the relation on the left and the step at the bottom
    (in \textcolor{RoyalBlue}{thin blue}) are premises, and the relation
    on the right and the step at the top (in
    \textcolor{RedOrange}{\textbf{bold red}}) are conclusions.
  \end{minipage}\hfill \begin{minipage}{0.28\linewidth}
    \centering \begin{tikzpicture}
      \node[RoyalBlue] (s)  at (0   , 1.3) {\(s\)};
      \node[RedOrange] (s') at (2.3 , 1.3) {\(\bm{s'}\)};
      \node[RoyalBlue] (t)  at (0   , 0)   {\(t\)};
      \node[RoyalBlue] (t') at (2.3 , 0)   {\(t'\)};

      \draw[dotted,RoyalBlue] (s) -- (t) node[midway,left] {\(\sim\)};
      \draw[dotted,thick,RedOrange] (s') -- (t') node[midway,left] {\(\bm{\sim}\)};

      \draw[->,thick,RedOrange] (s) -- (s') node[midway,above] {\(o_s\)};
      \draw[->,RoyalBlue] (t) -- (t') node[midway,above]
        {\(o_t \textcolor{RedOrange}{{} = \Tobs{}{o_s}}\)};
    \end{tikzpicture}
  \end{minipage}
\end{definition}

Intuitively, the observation transformer~\Tobsname{} explains the
observations of the output program as a function of the
observations of the input program.
This guarantees preservation because if the input program is \ensuremath{{\phi\text{-CT}}}{},
indistinguishable inputs give equal observations \ensuremath{{o_s = o_s'}},
and, thus, also the output program produces the same observations
\ensuremath{{\Tobs{}{o_s} = \Tobs{}{o_s'}}}.

CT reflection, however, aims for the dual of preservation: we want that
if the output program is \ensuremath{{\phi\text{-CT}}}{}, then the input program should be
\ensuremath{{\phi\text{-CT}}}{} as well (recall \cref{def:ct:transparency}).
Rather than crafting a new kind of simulation, we identify an additional
requirement on \Tobsname{} that guarantees CT reflection: injectivity.
If the observation transformer is injective, we can flip the argument
around: when the output program's observations are equal,
\ensuremath{{o_t = \Tobs{}{o_s} = \Tobs{}{o_s'} = o_t'}},
injectivity makes the input program's observations equal as well,
\ensuremath{{o_s = o_s'}}.
The following theorem identifies the sufficient conditions for using a
lock-step simulation to prove that a transformation also reflects CT, i.e., that it is CT transparent.

\begin{theorem}[Soundness of Lock-Step Simulations]\label{thm:ct:lock-step-soundness}
  A program transformation
  \ensuremath{{\ensuremath{{\llparenthesis\,{\cdot}\,\rrparenthesis}}{} : \ensuremath{{\mathcal{L}}}_s \to \ensuremath{{\mathcal{L}}}_t}} is CT transparent if
  for every input program~\(P\) there exist \(\sim\) and \Tobsname{}
  such that
  \begin{requirelist}
  \item $\sim$ satisfies a lock-step simulation diagram
    w.r.t.\ $\Tobsname$;\label{thm:ct:lock-step-soundness:simulation}
  \item Initial states are related: \ensuremath{{P(i) \sim \ensuremath{{\llparenthesis\,{P}\,\rrparenthesis}}(i)}} for
    every \(i\); and\label{thm:ct:lock-step-soundness:initial}
  \item \Tobsname{} is injective.\label{thm:ct:lock-step-soundness:injectivity}
  \end{requirelist}\end{theorem}

Note that this theorem ensures transparency with respect to \emph{every}
indistinguishability relation~\phi{} even though the simulation relation
\({\sim}\) and the observation transformer \Tobsname{} do not depend on
\phi{}.

The critical requirement for CT reflection in
\cref{thm:ct:lock-step-soundness} is that the observation transformer
\Tobsname{} is injective, that is, that equal output observations imply
equal input observations.
Since prior work has shown that lock-step simulations guarantee
preservation, and our condition guarantees reflection, we obtain transparency.

Lock-step simulations are too constraining for many standard transformations.
We introduce a relaxed version of simulations that allows us to
prove reflection for all transformations considered in the remaining paper.

\paragraph{Relaxed Simulations for CT Transparency}
In contrast to lock-step simulations,
relaxed simulations allow the input and output programs to perform
different number of steps.
To this end, when the output program steps, we allow the input program
to make any positive number of steps.
We introduce a function \ensuremath{{\ensuremath{{\mathsf{ns}}} : \ensuremath{{\mathcal{PC}}} \to \mathbb{N}_{> 0}}} that
determines the \emph{number of steps} the input program performs to
``catch up'' with the output program, i.e., such that the states are
related by $\sim$ again.
Requiring that the number of steps is nonzero simplifies our
definitions, and is easily satisfied in all our use cases, since
decompiler transformations rarely introduce instructions in the output
program---their aim is to simplify the program.
Consequently, some transformations remove instructions from the input
program, which means that the output program reaches a final state
before the input program does.
In these cases, we must justify the leakage of the input program without
an output program step: we introduce a \emph{suffix} function
\ensuremath{{\ensuremath{{\mathsf{sf}}}: \ensuremath{{\mathcal{PC}}} \to \ensuremath{{\mathcal{O}}}_s^*}}, which determines the
remaining observations of the input program when it reaches a program
point where the output program has terminated.

In addition to these relaxations, the observation transformer $\Tobsname{}$ becomes
a partial function that may additionally inspect program point of the source state to
transform observations, i.e., it takes an extra argument.
These extensions allow relaxed simulations to accommodate all
transformations of this work.
We now recast \Cref{def:ct:lock-step-diagram} and \Cref{thm:ct:lock-step-soundness}.

\begin{definition}[Simulation Diagram]\label{def:ct:diagram}
    A relation \ensuremath{{{\sim} \subseteq \ensuremath{{\mathcal{S}}}_s \times \ensuremath{{\mathcal{S}}}_t}}
    satisfies a simulation diagram w.r.t.\
    a partial observation transformer
    \ensuremath{{\Tobsname : \ensuremath{{\mathcal{PC}}} \times \ensuremath{{\mathcal{O}}}_s^* \rightharpoonup \ensuremath{{\mathcal{O}}}_t}},
    a number-of-steps function \ensuremath{{\mathsf{ns}}}{},
    and a suffix function~\ensuremath{{\mathsf{sf}}}{},
    if for all related states \ensuremath{{s \sim t}},
    the input program can execute \sem*{s}{}{\ensuremath{{\boldsymbol{o}}}_{s}}{s'}
    such that
    \begin{requirelist}
    \item \begin{minipage}[t]{0.68\linewidth}
            If the output program steps \sem{t}{}{o_t}{t'},
            then the input program takes
            \ensuremath{{\ensuremath{{\left|\ensuremath{{\boldsymbol{o}}}_s\right|}} = \ensuremath{{\ensuremath{{\mathsf{ns}}}(\ensuremath{{\ensuremath{{\mathsf{pc}}}({s})}})}}}} steps from
            \sem*{s}{}{\ensuremath{{\boldsymbol{o}}}_s}{s'},
            so that \ensuremath{{o_t = \Tobs{\ensuremath{{\ensuremath{{\mathsf{pc}}}({s})}}}{\ensuremath{{\boldsymbol{o}}}_s}}}
            and $s' \sim t'$.
        \end{minipage}\hfill \begin{minipage}[t]{0.3\linewidth}
            \makebox[\textwidth][c]{
              \small
              \begin{tikzpicture}[baseline=5ex]
                \node[RoyalBlue] (s)  at (0   , 1.1) {\(s\)};
                \node[RedOrange] (s') at (2.8 , 1.1) {\(\bm{s'}\)};
                \node[RoyalBlue] (t)  at (0   , 0)   {\(t\)};
                \node[RoyalBlue] (t') at (2.8 , 0)   {\(t'\)};

                \draw[dotted,RoyalBlue] (s) -- (t) node[midway,left] {\(\sim\)};
                \draw[dotted,thick,RedOrange] (s') -- (t') node[midway,left] {\(\bm{\sim}\)};
                \draw[->,thick,RedOrange] (s) -- (s')
                node[midway,above] {\(\ensuremath{{\boldsymbol{o}}}_s\)}
                node[above right, xshift=-3.0ex,yshift=-.5ex]
                {\scriptsize $*$};

                \draw[->,RoyalBlue] (t) -- (t')
                node[midway,below]{\(o_t \textcolor{RedOrange}{{} = \Tobs{\ensuremath{{\ensuremath{{\mathsf{pc}}}({s})}}}{\ensuremath{{\boldsymbol{o}}}_s}}\)};
              \end{tikzpicture}
            }
        \end{minipage}
        \item\label{def:ct:diagram:injectivity}
            \begin{minipage}[t]{0.68\linewidth}
            If~$t$ is final instead, then
            the source observations \(\ensuremath{{\boldsymbol{o}}}_s\) are
            \ensuremath{{\ensuremath{{\mathsf{sf}}}({\ensuremath{{\ensuremath{{\mathsf{pc}}}({s})}}})}},
            \(s'\) is final, and
            $s' \sim t$.
            \end{minipage}\hfill \begin{minipage}[t]{0.3\linewidth}
              \makebox[\textwidth][c]{
                \small
                \begin{tikzpicture}[baseline=5ex]
                  \node[RoyalBlue] (s)  at (0   , 1.1) {\(s\)};
                  \node[RedOrange] (s') at (2.8 , 1.1) {\(\bm{s'}\)};
                  \node[RoyalBlue] (t)  at (0   , 0)   {\(t\)};

                \draw[dotted,RoyalBlue] (s) -- (t) node[midway,left] {\(\sim\)};
                \draw[dotted,thick,RedOrange] (s'.south west) -- (t)
                  node[midway,left] {\(\sim\)};
                \draw[->,thick,RedOrange] (s) -- (s')
                    node[midway,above] {\(\ensuremath{{\boldsymbol{o}}}_s = \ensuremath{{\ensuremath{{\mathsf{sf}}}({\ensuremath{{\ensuremath{{\mathsf{pc}}}({s})}}})}}\)}
                    node[above right, xshift=-3.0ex,yshift=-.5ex]  {\scriptsize $*$};
            \end{tikzpicture}
            }
            \end{minipage}
    \end{requirelist}
\end{definition}

We must overcome a final challenge to achieve transparency in this
relaxed setting: we redefine the injectivity requirement
\labelcref{thm:ct:lock-step-soundness:injectivity},
since now the observation transformer \Tobsname{}
inspects program points and is partial.
First we require injectivity on~\Tobsname{} \emph{per program point}.
Second, we require injectivity only where~\Tobsname{} is defined.
We call this new condition \emph{\ensuremath{{\mathcal{PC}}}{}-injectivity}.

\begin{definition}[\ensuremath{{\mathcal{PC}}}{}-Injectivity]\label{def:ct:inj}
  A partial observation transformer $\Tobsname{} : \ensuremath{{\mathcal{PC}}} \times \ensuremath{{\mathcal{O}}}_s^* \to \ensuremath{{\mathcal{O}}}_t$ is \ensuremath{{\ensuremath{{\mathcal{PC}}}\text{-injective}}}{},
  when for each $\ensuremath{{\mathsf{pc}}} \in \ensuremath{{\mathcal{PC}}}$,
  \begin{equation*}
    \Tobs{\ensuremath{{\mathsf{pc}}}}{\ensuremath{{\boldsymbol{o}}}_{1}} = \Tobs{\ensuremath{{\mathsf{pc}}}}{\ensuremath{{\boldsymbol{o}}}_{2}} \neq \bot \implies
    \ensuremath{{\boldsymbol{o}}}_{1} = \ensuremath{{\boldsymbol{o}}}_{2}\text.
  \end{equation*}
\end{definition}

\begin{theorem}[Soundness of Simulations]\label{thm:ct:soundness}
  A program transformation
  \ensuremath{{{\ensuremath{{\llparenthesis\,{\cdot}\,\rrparenthesis}}} : \ensuremath{{\mathcal{L}}}_s \to \ensuremath{{\mathcal{L}}}_t}}
  is \ensuremath{{\text{CT}}}{} transparent if for every input program \(P\) there exist
  a relation \(\sim\),
  a number-of-steps function~\ensuremath{{\mathsf{ns}}}{},
  an observation transformer~\Tobsname{}, and
  a suffix function~\ensuremath{{\mathsf{sf}}}{},
  such that
  \begin{requirelist}
  \item\label{thm:ct:soundness:simulation} $\sim$ satisfies a simulation
    diagram w.r.t.\ \Tobsname{}, \ensuremath{{\mathsf{ns}}}{}, and \ensuremath{{\mathsf{sf}}}{};
  \item\label{thm:ct:soundness:initial} Initial states are related:
    \ensuremath{{P(i) \sim \ensuremath{{\llparenthesis\,{P}\,\rrparenthesis}}(i)}} for every \(i\); and
  \item\label{thm:ct:soundness:injectivity} \Tobsname{} is \ensuremath{{\ensuremath{{\mathcal{PC}}}\text{-injective}}}.
  \end{requirelist}
\end{theorem}

For a detailed proof of this theorem, we refer the reader to our \textsc{Rocq}{}
development in the artifact.
In this mechanized proof, we split \cref{thm:ct:soundness} into
two parts: Theorem \texttt{th2\_preserves} states that simulations
ensure preservation, and Theorem \texttt{th2\_reflect} states that
simulations guarantee reflection.
In our formal development, we further generalized some of the
definitions, which we simplified in the text for better presentation.

\section{Transparent and Nontransparent Transformations}\label{sec:passes}
This section considers the CT transparency of ten common
transformations.
We prove that seven of them are transparent, provide
counterexamples for the rest, and summarize the results in
\cref{tab:passes}.

\begin{table}
  \caption{Summary of decompilation passes in this section.
    A~\textcolor{green}{\ding{51}}{} means that we prove the pass transparent in
    \Cref{sec:prog-optimizations}, a~\textcolor{red}{\Lightning}{} that we present a
    counterexample to \ensuremath{{\text{CT}}}{} reflection in \Cref{sec:nontranspasse},
    and~\(\bullet\){} means the pass is \ensuremath{{\text{CT}}}{} preserving, but we omit
    the proof.}\label{tab:passes}
  \small \begin{tabular}{l c c}
    \toprule
    Pass & \multicolumn{1}{l}{Reflection} & \multicolumn{1}{l}{Preservation}\\
    \midrule
    Branch Coalescing & \textcolor{red}{\Lightning}{} & \(\bullet\){}\\
If Conversion & \textcolor{red}{\Lightning}{} & \(\bullet\){}\\
    Memory Access Elimination & \textcolor{red}{\Lightning}{} & \(\bullet\){}\\
Constant Folding & \textcolor{green}{\ding{51}}{} & \textcolor{green}{\ding{51}}{}\\
    Untiling & \textcolor{green}{\ding{51}}{} & \textcolor{green}{\ding{51}}{}\\
    Dead Branch Elimination & \textcolor{green}{\ding{51}}{} & \textcolor{green}{\ding{51}}{}\\
    Dead Assignment Elimination & \textcolor{green}{\ding{51}}{} & \textcolor{green}{\ding{51}}{}\\
    Unspilling & \textcolor{green}{\ding{51}}{} & \textcolor{green}{\ding{51}}{}\\
    Structural Analysis & \textcolor{green}{\ding{51}}{} & \textcolor{green}{\ding{51}}{}\\
    Loop Rotation & \textcolor{green}{\ding{51}}{} & \textcolor{green}{\ding{51}}{}\\
    \bottomrule
  \end{tabular}
\end{table}

\subsection{Language and Leakage Model}\label{sec:passes:language}
To prove CT transparency of the program transformations,
we instantiate our framework from \Cref{sec:setting} with a
core language of register assignments, memory operations, and loops.
The syntax of the language is as follows:
\begin{gather*}
    \begin{align*}
        \ensuremath{{\mathit{e}}} & \Coloneqq z \mid b \mid x \mid \oplus(\ensuremath{{\mathit{e}}}, \cdots, \ensuremath{{\mathit{e}}}), &
        \ensuremath{{\mathit{a}}} & \Coloneqq
        \ensuremath{{{x}\mathcolor{cyan}{=}{e}}} \mid \ensuremath{{{x}\mathcolor{cyan}{=}{\left[e\right]}}} \mid \ensuremath{{{\left[e\right]}\mathcolor{cyan}{=}{x}}},
                &
    \end{align*}
    \\
    \begin{align*}
        \ensuremath{{\mathit{c}}} & \Coloneqq
        \ensuremath{{\mathcolor{mauve}{\mathtt{skip}}}}
        \mid \ensuremath{{\mathit{c}}} \mathop{;} \ensuremath{{\mathit{c}}}
        \mid \ensuremath{{\mathit{a}}}
        \mid \ensuremath{{\mathcolor{mauve}{\mathtt{if}}~{e}~\mathcolor{mauve}{\mathtt{then}}~\ensuremath{{\mathit{c}}}~\mathcolor{mauve}{\mathtt{else}}~\ensuremath{{\mathit{c}}}}}
        \mid \ensuremath{{\mathcolor{mauve}{\mathtt{while}}~{e}~\mathcolor{mauve}{\mathtt{do}}~{\ensuremath{{\mathit{c}}}}}},
                &
    \end{align*}
\end{gather*}
where
\(z\)~is an integer,
\(b\)~is a boolean,
\(x\)~is a register variable, and
\(\oplus\)~an operation (such as negation \(\neg\) or addition \(+\)).
Expressions~\(e\) are built from constants and register variables
but may not invoke memory accesses,
this means that they have no side effects and no leakage.
Atomic commands~\ensuremath{{\mathit{a}}}{} are assignments, loads, and stores.
Commands~\ensuremath{{\mathit{c}}}{} are the standard while language constructs.
We make the usual assumptions that the empty program \ensuremath{{\mathcolor{mauve}{\mathtt{skip}}}}{} is neutral for sequentialization
\ensuremath{{\ensuremath{{\mathcolor{mauve}{\mathtt{skip}}}} \mathop{;} \ensuremath{{\mathit{c}}} = \ensuremath{{\mathit{c}}} \mathop{;} \ensuremath{{\mathcolor{mauve}{\mathtt{skip}}}} = \ensuremath{{\mathit{c}}}}}, and that
sequencing is associative \ensuremath{{
  (\ensuremath{{\mathit{c}}} \mathop{;} \ensuremath{{\mathit{c}}}') \mathop{;} \ensuremath{{\mathit{c}}}''
  =
  \ensuremath{{\mathit{c}}} \mathop{;} (\ensuremath{{\mathit{c}}}' \mathop{;} \ensuremath{{\mathit{c}}}'')
}}.

\newcommand*{\lblassign}{Assign}
\newcommand*{\lblload}{Load}
\newcommand*{\lblstore}{Store}
\newcommand*{\lblseqcond}{Cond}
\newcommand*{\lblwhile}{While}
\newcommand*{\lblseq}{Seq}
\newcommand*{\refassign}{\hyperref[fig:sem:single-step]{\textsc{\normalfont\textsc{\lblassign}}}}
\newcommand*{\refload}{\hyperref[fig:sem:single-step]{\textsc{\normalfont\textsc{\lblload}}}}
\newcommand*{\refstore}{\hyperref[fig:sem:single-step]{\textsc{\normalfont\textsc{\lblstore}}}}
\newcommand*{\refif}{\hyperref[fig:sem:single-step]{\textsc{\normalfont\textsc{\lblseqcond}}}}
\newcommand*{\refwhile}{\hyperref[fig:sem:single-step]{\textsc{\normalfont\textsc{\lblwhile}}}}
\newcommand*{\refseq}{\hyperref[fig:sem:single-step]{\textsc{\normalfont\textsc{\lblseq}}}}

\begin{figure}
  \small
  \begin{mathpar}
    \inferrule[\lblassign]{
      \rho' = \rho[x \leftarrow \ensuremath{{\llbracket\, e \,\rrbracket_{\rho}}}]
    }{
        \ensuremath{{\langle \ensuremath{{{x}\mathcolor{cyan}{=}{e}}}, \rho, \mu \rangle}}
        \step{}{\ensuremath{{\bullet}}}
        \ensuremath{{\langle \ensuremath{{\mathcolor{mauve}{\mathtt{skip}}}}, \rho', \mu \rangle}}
    }

    \inferrule[\lblload]{
      n = \ensuremath{{\llbracket\, e \,\rrbracket_{\rho}}}\\
      \rho' = \rho[x \leftarrow \mu(n)]
    }{
        \ensuremath{{\langle \ensuremath{{{x}\mathcolor{cyan}{=}{\left[e\right]}}}, \rho, \mu \rangle}}
        \step{}{\ensuremath{{\mathsf{adr}\; n}}}
        \ensuremath{{\langle \ensuremath{{\mathcolor{mauve}{\mathtt{skip}}}}, \rho', \mu \rangle}}
    }

    \inferrule[\lblstore]{
      n = \ensuremath{{\llbracket\, e \,\rrbracket_{\rho}}}\\
      \mu' = \mu[n \leftarrow \rho(x)]
    }{
        \ensuremath{{\langle \ensuremath{{{\left[e\right]}\mathcolor{cyan}{=}{x}}}, \rho, \mu \rangle}}
        \step{}{\ensuremath{{\mathsf{adr}\; n}}}
        \ensuremath{{\langle \ensuremath{{\mathcolor{mauve}{\mathtt{skip}}}}, \rho, \mu' \rangle}}
    }

    \inferrule[\lblseqcond]{
        \ensuremath{{\llbracket\, e \,\rrbracket_{\rho}}} = b
    }{
        \ensuremath{{\langle \ensuremath{{\mathcolor{mauve}{\mathtt{if}}~{e}~\mathcolor{mauve}{\mathtt{then}}~\ensuremath{{\mathit{c}}}_{\mathsf{t\kern-0.9ptt}}~\mathcolor{mauve}{\mathtt{else}}~\ensuremath{{\mathit{c}}}_{\mathsf{f\kern-1.1ptf}}}}, \rho, \mu \rangle}}
        \step{}{\ensuremath{{\mathsf{br}\; b}}}
        \ensuremath{{\langle \ensuremath{{\mathit{c}}}_{b}, \rho, \mu \rangle}}
    }

    \inferrule[\lblwhile]{
        \ensuremath{{\llbracket\, e \,\rrbracket_{\rho}}} = b
        \\
        \ensuremath{{\mathit{c}}}_{\mathsf{t\kern-0.9ptt}} = \ensuremath{{\mathit{c}}};\ensuremath{{\mathcolor{mauve}{\mathtt{while}}~{e}~\mathcolor{mauve}{\mathtt{do}}~{\ensuremath{{\mathit{c}}}}}}
        \\
        \ensuremath{{\mathit{c}}}_{\mathsf{f\kern-1.1ptf}} = \ensuremath{{\mathcolor{mauve}{\mathtt{skip}}}}
    }{
        \ensuremath{{\langle \ensuremath{{\mathcolor{mauve}{\mathtt{while}}~{e}~\mathcolor{mauve}{\mathtt{do}}~{\ensuremath{{\mathit{c}}}}}}, \rho, \mu \rangle}}
        \step{}{\ensuremath{{\mathsf{br}\; b}}}
        \ensuremath{{\langle \ensuremath{{\mathit{c}}}_b, \rho, \mu \rangle}}
    }

    \inferrule[\lblseq]{
        \ensuremath{{\langle \ensuremath{{\mathit{c}}}, \rho, \mu \rangle}}
        \step{}{o}
        \ensuremath{{\langle \ensuremath{{\mathit{c}}}', \rho', \mu' \rangle}}
    }{
        \ensuremath{{\langle \ensuremath{{\mathit{c}}} \mathop{;} \ensuremath{{\mathit{c}}}'', \rho, \mu \rangle}}
        \step{}{o}
        \ensuremath{{\langle \ensuremath{{\mathit{c}}}'\mathop{;} \ensuremath{{\mathit{c}}}'', \rho', \mu' \rangle}}
    }
  \end{mathpar}
  \caption{Single-step execution semantics.}\label{fig:sem:single-step}
\end{figure}

The semantics operates on states \ensuremath{{\langle \ensuremath{{\mathit{c}}}, \rho, \mu \rangle}}
consisting of the command~\ensuremath{{\mathit{c}}}{} remaining to be executed,
a register map~\ensuremath{{\rho}}{} that associates each register variable with
a value, and
a memory~\ensuremath{{\mu}}{} that associates each address with a value.
The program point of a state is its code, i.e.,
\ensuremath{{\ensuremath{{\ensuremath{{\mathsf{pc}}}({\ensuremath{{\langle c, \ensuremath{{\rho}}, \ensuremath{{\mu}} \rangle}}})}} \triangleq c}},
and a state is final if its code is $\ensuremath{{\mathcolor{mauve}{\mathtt{skip}}}}$.
\Cref{fig:sem:single-step} presents the small-step
semantics \sem{s}{}{o}{s'} of commands.
We choose the standard constant-time leakage model~\cite{DBLP:conf/csfw/BartheGL18} for our semantics, i.e.,
memory operations leak their address (\ensuremath{{\mathsf{adr}\; n}} in \textsc{\lblload}, \textsc{\lblstore}) and
branch instructions leak their condition (\ensuremath{{\mathsf{br}\; b}} in \textsc{\lblseqcond}, \textsc{\lblwhile}).
The semantics and leakage model satisfy the requirements of \Cref{sec:setting} (it is deterministic, and programs are safe).

\subsection{Constant-Time Transparent Transformations}\label{sec:prog-optimizations}
We now examine the seven transparent transformations in
\Cref{tab:passes} and provide proof sketches for their transparency with
the framework from \Cref{sec:proof-techniques}.
We will keep the transformations as abstract as possible and focus on
their transparency rather than the details of their functionality.
When convenient, we split passes into an analysis, which annotates the program,
followed by a program transformation, which expects these annotations in the input program.
We present detailed definitions of the transformations and
proofs of transparency in \ifAppendix{\cref{appendix:proofs-prog-optimizations}}\fi{}.

\paragraph*{Expression Substitution}
This transformation replaces expressions with other expressions that
produce the same values (it generalizes Constant Folding and
Untiling).
For example, it may replace \ensuremath{{3 * x - x}} with \ensuremath{{2 * x}}.
It may, as well, rely on the program's structure and replace
the instruction sequence \ensuremath{{\ensuremath{{{x}\mathcolor{cyan}{=}{y*8}}} \mathop{;} \ensuremath{{{x}\mathcolor{cyan}{=}{\left[z + x\right]}}}}}
with \ensuremath{{\ensuremath{{{x}\mathcolor{cyan}{=}{y*8}}} \mathop{;} \ensuremath{{{x}\mathcolor{cyan}{=}{\left[z + (y*8)\right]}}}}}.
In general, Expression Substitution is preceded by an analysis that identifies
and annotates all expressions that is to be replaced.
We write an annotated expression as $\ensuremath{{\{{e_s}\mathbin{\text{\normalfont\guilsinglright}}{e_t}\}}}e$, which means that
every expression $e_s$ occuring in $e$ are to be replaced by $e_t$.
Whichever analysis is used, as long as it guarantees the following criterion,
that states that the expressions $e_s$ and $e_t$ indeed yield the same values,
we can prove the actual transformation transparent.
\begin{proposition}\label{prop:expr-subst-guarantee}
  For all inputs $i$ and executions
  $\sem*{P(i)}{}{\ensuremath{{\boldsymbol{o}}}}{\ensuremath{{\langle c, \rho, \mu \rangle}}}$,
  if the next instruction of $c$ contains an annotated expression
  $\ensuremath{{\{{\ensuremath{{\mathit{e}}}_s}\mathbin{\text{\normalfont\guilsinglright}}{\ensuremath{{\mathit{e}}}_t}\}}}\ensuremath{{\mathit{e}}}$, then
  $\ensuremath{{\llbracket\, \ensuremath{{\mathit{e}}}_s \,\rrbracket_{\rho}}} = \ensuremath{{\llbracket\, \ensuremath{{\mathit{e}}}_t \,\rrbracket_{\rho}}}$.
\end{proposition}
A few typical passes in static analysis tools and decompilers are subsumed
by this definition. Among them are Constant Folding and Untiling.
Constant Folding replaces expressions without variables by an appropriate
constant (our first example). Untiling is a decompilation pass that typically
sequences of instructions into a single, more complex instruction.
This is desirable when compilers split complex expressions into multiple
instructions in order to compute the expression on hardware,
but the complex expression is more readable.
Our second example from above is a case of untiling:
the first instruction $\ensuremath{{{x}\mathcolor{cyan}{=}{y*8}}}$ computes the address offset for the
second instruction $\ensuremath{{{x}\mathcolor{cyan}{=}{\left[z+x\right]}}}$,
which are merged into a single instruction $\ensuremath{{{x}\mathcolor{cyan}{=}{\left[z+(y*8)\right]}}}$.
Notice we define Expression Substitution to not remove the first instruction \ensuremath{{{x}\mathcolor{cyan}{=}{y*8}}}.
This step is left to Dead Assignment Elminination, for
which we prove transparency in a seperate transformation.

Let us turn to proving the transparency of Expression Substitution.
The transformation maintains the structure of the program and the
number of statements.
Therefore, we define a lock-step simulation and apply
\Cref{thm:ct:lock-step-soundness}.
This involves defining a relation~\({\sim}\) and a
function~\Tobsname{} such that
\begin{requirelist}
\item \({\sim}\) satisfies a lock-step simulation diagram
  w.r.t.\ \Tobsname{};
\item \ensuremath{{P(i) \sim \ensuremath{{\llparenthesis\,{P}\,\rrparenthesis}}(i)}} for every \(i\); and
\item \Tobsname{} is injective.
\end{requirelist}

\begin{theorem}
  Expression Substitution is CT transparent.
\end{theorem}
\begin{proof}[Proof Sketch]
  We define the simulation relation $\sim$ such that \ensuremath{{\langle c, \ensuremath{{\rho}}, \ensuremath{{\mu}} \rangle}} is
  related only to \ensuremath{{\langle \ensuremath{{\llparenthesis\,{c}\,\rrparenthesis}}, \ensuremath{{\rho}}, \ensuremath{{\mu}} \rangle}}.
  That is, the code in the output program's state is the transformed input program's code at that state,
  and the register map as well as the memory coincides.
  The leakage transformer is the identity: \ensuremath{{\Tobs{}{o_s} = o_s}}.
  This satisfies the requirements of \Cref{thm:ct:lock-step-soundness}
  as follows:
  \begin{requirelist}
  \item Since the evaluated expressions produce the same value
    (\Cref{prop:expr-subst-guarantee}),
    every step in the output program matches exactly one step in the input
    program, and they produce the same register map and memory.
    Since the register map and memory are identical, the input and
    output observations are equal.
    For instance, the instructions \ensuremath{{{x}\mathcolor{cyan}{=}{\left[z + y * 8\right]}}} and \ensuremath{{{x}\mathcolor{cyan}{=}{\left[z + x\right]}}}
    from the above example produce the same
    observation, namely \ensuremath{{\mathsf{adr}\; \ensuremath{{\llbracket\, z+y*8 \,\rrbracket_{}}}}}.
  \item Initial states are trivially related.
  \item \Tobsname{} is injective because it is the identity.\qedhere
  \end{requirelist}
\end{proof}

\paragraph*{Dead Branch Elimination}\label{subsec:dead-branch-elimination}
This transformation eliminates conditional branches that are never taken.
That is, it transforms conditionals \ensuremath{{\mathcolor{mauve}{\mathtt{if}}~{b}~\mathcolor{mauve}{\mathtt{then}}~c_\mathsf{t\kern-0.9ptt}~\mathcolor{mauve}{\mathtt{else}}~c_\mathsf{f\kern-1.1ptf}}}, where
\(b \in \ensuremath{{\{\mathsf{t\kern-0.9ptt},\mathsf{f\kern-1.1ptf}\}}}\) is a
constant, into the corresponding branch \(c_b\).
The transformation modifies no other instructions. Intuitively, this pass is
transparent because a dead branch is never executed, and, therefore,
removing it does not affect leakage.
Also, removing the leak that stems from the branching instruction itself is not an issue,
because the constant $b$ is secret-independent.

\begin{figure}
  \begin{align*}
    \Tobs{c}{\ensuremath{{\boldsymbol{o}}}} &\triangleq
      \begin{cases}
        \Tobs{c_b \mathop{;} c'}{\ensuremath{{\boldsymbol{o}}}'}
        & \ensuremath{{\text{if } c \text{ is } \ensuremath{{\mathcolor{mauve}{\mathtt{if}}~{b}~\mathcolor{mauve}{\mathtt{then}}~c_\mathsf{t\kern-0.9ptt}~\mathcolor{mauve}{\mathtt{else}}~c_\mathsf{f\kern-1.1ptf}}} \mathop{;} c'\text{and } \ensuremath{{\boldsymbol{o}}} = \ensuremath{{\mathsf{br}\; b}} \mathop{\cdot} \ensuremath{{\boldsymbol{o}}}',}}
        \\
        \ensuremath{{\boldsymbol{o}}}
        & \text{if $c$ is not $\ensuremath{{\mathcolor{mauve}{\mathtt{if}}~{b}~\mathcolor{mauve}{\mathtt{then}}~c_\mathsf{t\kern-0.9ptt}~\mathcolor{mauve}{\mathtt{else}}~c_\mathsf{f\kern-1.1ptf}}} \mathop{;} c'$ and  $\ensuremath{{\left|\ensuremath{{\boldsymbol{o}}}\right|}} = 1$,}
        \\
        \bot & \ensuremath{{\text{otherwise}}}\text,
      \end{cases}
    \\
    \ensuremath{{\ensuremath{{\mathsf{ns}}}(c)}} &\triangleq
      \begin{cases}
        \ensuremath{{\ensuremath{{\mathsf{ns}}}(c_b; c')}} + 1 &
          \ensuremath{{\text{if } c \text{ is } \ensuremath{{\mathcolor{mauve}{\mathtt{if}}~{b}~\mathcolor{mauve}{\mathtt{then}}~c_\mathsf{t\kern-0.9ptt}~\mathcolor{mauve}{\mathtt{else}}~c_\mathsf{f\kern-1.1ptf}}} \mathop{;} c'}}\text,\\
        1 & \ensuremath{{\text{otherwise}}}\text,
      \end{cases}
    \\
    \ensuremath{{\ensuremath{{\mathsf{sf}}}({c})}} &\triangleq
      \begin{cases}
        \ensuremath{{\mathsf{br}\; b}} \mathop{\cdot} \ensuremath{{\ensuremath{{\mathsf{sf}}}({c_b; c'})}} &
          \ensuremath{{\text{if } c \text{ is } \ensuremath{{\mathcolor{mauve}{\mathtt{if}}~{b}~\mathcolor{mauve}{\mathtt{then}}~c_\mathsf{t\kern-0.9ptt}~\mathcolor{mauve}{\mathtt{else}}~c_\mathsf{f\kern-1.1ptf}}} \mathop{;} c'}}\text,\\
        \ensuremath{{\epsilon}} &\ensuremath{{\text{otherwise}}}\text.
      \end{cases}
  \end{align*}
  \caption{Relaxed simulation instantiations of $\protect\Tobsname$,
    $\ensuremath{{\mathsf{ns}}}$, and $\ensuremath{{\mathsf{sf}}}$ for Dead Branch Elimination.}\label{fig:sim-dbe}
\end{figure}

\begin{theorem}
  Dead Branch Elimination is CT transparent.
\end{theorem}
\begin{proof}[Proof Sketch]
  Since Dead Branch Elimination removes code, the input and output
  programs perform a different number of steps.
  This is the setting to use our relaxed version of simulations and
  apply \Cref{thm:ct:soundness}.
  It requires us to provide $\sim$ and $\Tobsname$ as before,
  but additionally requires the number of steps function $\ensuremath{{\mathsf{ns}}}$
  and the suffix transformer $\ensuremath{{\mathsf{sf}}}$.
  We define $\sim$ to only hold on $\ensuremath{{\langle c, \ensuremath{{\rho}}, \ensuremath{{\mu}} \rangle}} \sim \ensuremath{{\langle \ensuremath{{\llparenthesis\,{c}\,\rrparenthesis}}, \ensuremath{{\rho}}, \ensuremath{{\mu}} \rangle}}$ as before,
  and we provide the remaining instantiations in \Cref{fig:sim-dbe}.

  Our definitions satisfy the requirements of \Cref{thm:ct:soundness}:
  \begin{requirelist}
  \item We prove that $\sim$ satisfies \Cref{def:ct:diagram} w.r.t. $\Tobsname$,
    $\ensuremath{{\mathsf{ns}}}$, and $\ensuremath{{\mathsf{sf}}}$,
      so let $s \sim t$ be given.
    If \(t\) is not a final state, \(s\)~needs to perform zero
    or more steps corresponding to the constant branches
    before performing the same step as~\(t\).
    This number of steps is provided by \ensuremath{{\ensuremath{{\mathsf{ns}}}(\ensuremath{{\ensuremath{{\mathsf{pc}}}({s})}})}}.
    Consequently, the input program's observations are of the form
    \ensuremath{{\ensuremath{{\mathsf{br}\; b_1}} \mathop{\cdot} \dots \mathop{\cdot} \ensuremath{{\mathsf{br}\; b_n}} \mathop{\cdot} o_t}},
    where \(b_i\) is the \(i\)-th dead branch of \(s\).
    The observation transformer \Tobsname{} deletes the branching observations.
    If $t$ is final, there might remain conditional branches to be executed
    in~$s$. Therefore, $\ensuremath{{\mathsf{sf}}}$ provides \ensuremath{{\mathsf{br}\; b_i}} observations
    corresponding to the branches at $s$.
  \item As with Expression Substitution, initial states are trivially related
    by the definition of $\sim$.
  \item To show \ensuremath{{
      \Tobs{c}{\ensuremath{{\boldsymbol{o}}}_1} = \Tobs{c}{\ensuremath{{\boldsymbol{o}}}_2} \neq \bot \implies
      \ensuremath{{\boldsymbol{o}}}_1 = \ensuremath{{\boldsymbol{o}}}_2
    }},
    we proceed by induction on the length of \(\ensuremath{{\boldsymbol{o}}}_1\).
    There are only two significant cases, corresponding to whether the
    code $c$ starts with a dead branch.
    If it does not, the transformer is the identity which makes
    $\ensuremath{{\boldsymbol{o}}}_1 = \ensuremath{{\boldsymbol{o}}}_2$ hold trivially. Otherwise, the first observations of
    both \(\ensuremath{{\boldsymbol{o}}}_1\) and \(\ensuremath{{\boldsymbol{o}}}_2\) must be \ensuremath{{\mathsf{br}\; b}},
    and the remaining observations match by inductive hypothesis. \qedhere
  \end{requirelist}
\end{proof}

\paragraph*{Dead Assignment Elimination}
This transformation removes assignments to dead variables,
i.e., when the assignments have no impact on the execution of the
program.
For example, it replaces \ensuremath{{\ensuremath{{{x}\mathcolor{cyan}{=}{e_1}}} \mathop{;} \ensuremath{{{x}\mathcolor{cyan}{=}{\left[e_2\right]}}}}}
with \ensuremath{{{x}\mathcolor{cyan}{=}{\left[e_2\right]}}} because the first assignment is never used.
This transformation removes only assignments, \emph{but not loads or stores}.
Indeed, removing memory accesses would break CT reflection, as their
addresses may leak (see Dead Load Elimination in
\Cref{sec:nontranspasse}).
We assume that a previous analysis annotated assignment instructions with a set
\ensuremath{{\mathsf{D}}}{} of the dead variables at that point (we omit the correctness guarantee).
The transformation then removes all assignments \ensuremath{{\ensuremath{{{x}\mathcolor{cyan}{=}{e}}}\ensuremath{{\{\ensuremath{{\mathsf{D}}}\}}}}}
where \ensuremath{{x \in \ensuremath{{\mathsf{D}}}}}, i.e., the assigned variable is marked dead.

\begin{figure}
  \begin{align*}
    \Tobs{c}{\ensuremath{{\boldsymbol{o}}}} &\triangleq
    \begin{cases}
      \Tobs{c'}{\ensuremath{{\boldsymbol{o}}}'}
      & \ensuremath{{\text{if } c \text{ is } \text{$\ensuremath{{{x}\mathcolor{cyan}{=}{e}}}\ensuremath{{\{\ensuremath{{\mathsf{D}}}\}}} \mathop{;} c'$ with
        $x \in \ensuremath{{\mathsf{D}}}$ and $\ensuremath{{\boldsymbol{o}}} = \ensuremath{{\bullet}} \mathop{\cdot} \ensuremath{{\boldsymbol{o}}}'$,}}}
      \\
      \ensuremath{{\boldsymbol{o}}}
      & \text{if $c$  is not $\ensuremath{{{x}\mathcolor{cyan}{=}{e}}}\ensuremath{{\{\ensuremath{{\mathsf{D}}}\}}} \mathop{;} c'$  with
        $x \in \ensuremath{{\mathsf{D}}}$ and $\ensuremath{{\left|\ensuremath{{\boldsymbol{o}}}\right|}} = 1$,}
      \\
      \bot & \ensuremath{{\text{otherwise}}}\text,
    \end{cases}
    \\
    \ensuremath{{\ensuremath{{\mathsf{ns}}}(c)}} &\triangleq
    \begin{cases}
      \ensuremath{{\ensuremath{{\mathsf{ns}}}(c')}} + 1
      & \ensuremath{{\text{if } c \text{ is } \ensuremath{{{x}\mathcolor{cyan}{=}{e}}}\ensuremath{{\{\ensuremath{{\mathsf{D}}}\}}} \mathop{;} c'}}
        \text{ and } x \in \ensuremath{{\mathsf{D}}}\text,\\
      1 & \ensuremath{{\text{otherwise}}}\text,
    \end{cases}\\
    \ensuremath{{\ensuremath{{\mathsf{sf}}}({c})}} &\triangleq
    \begin{cases}
      \ensuremath{{\bullet}} \mathop{\cdot} \ensuremath{{\ensuremath{{\mathsf{sf}}}({c'})}}
        & \ensuremath{{\text{if } c \text{ is } \ensuremath{{{x}\mathcolor{cyan}{=}{e}}} \mathop{;} c'}}\text,\\
      \ensuremath{{\epsilon}} & \ensuremath{{\text{otherwise}}}\text.
    \end{cases}
  \end{align*}
  \caption{Relaxed simulation instantiations of $\protect\Tobsname{}$,
    $\ensuremath{{\mathsf{ns}}}$, and $\ensuremath{{\mathsf{sf}}}$ for Dead Assignment Elimination.}\label{fig:sim-dae}
\end{figure}

\begin{theorem}
  Dead Assignment Elimination is CT transparent.
\end{theorem}
\begin{proof}[Proof Sketch]
  The proof of reflection for this pass is similar to Dead Branch
  Elimination, i.e., we apply \Cref{thm:ct:soundness}.
  We define the simulation relation
  \ensuremath{{\ensuremath{{\langle c, \ensuremath{{\rho}}, \ensuremath{{\mu}} \rangle}} \sim \ensuremath{{\langle c', \ensuremath{{\rho}}', \ensuremath{{\mu}}' \rangle}}}}
  when \ensuremath{{c' = \ensuremath{{\llparenthesis\,{c}\,\rrparenthesis}}}}, \ensuremath{{\rho}}{} and \(\ensuremath{{\rho}}'\) coincide everywhere
  except on dead variables, and \ensuremath{{\ensuremath{{\mu}} = \ensuremath{{\mu}}'}}.
  We define the remaining instantiations in \Cref{fig:sim-dae}.
  The proof of the requirements of \Cref{thm:ct:soundness} is analogous
  to Dead Branch Elimination.
\end{proof}

\paragraph*{Unspilling}
This  transformation removes spill-unspill pairs.
Compilers emit spill instructions to temporarily store a value to the stack
before reloading it later with an unspill.
This transformation avoids using the memory by moving such values into fresh
variables instead.
It allows us to translate an assembly-like language with
a limited set of registers to a high-level language with unlimited
abstract variables,
potentially improving the accuracy of static analyses.
A spill-unspill pair consists of a store and a load instructions with
the same constant offset from the stack pointer.
We denote the stack pointer as~\(\mathsf{sp}\).
For instance, unspilling transforms
\ensuremath{{
  \ensuremath{{{\left[\mathsf{sp} + 4\right]}\mathcolor{cyan}{=}{x}}} \mathop{;} c \mathop{;} \ensuremath{{{x}\mathcolor{cyan}{=}{\left[\mathsf{sp} + 4\right]}}}
}}
into \ensuremath{{\ensuremath{{{y}\mathcolor{cyan}{=}{x}}} \mathop{;} c \mathop{;} \ensuremath{{{x}\mathcolor{cyan}{=}{y}}}}}, where
\(c\) does not touch the memory location \ensuremath{{\mathsf{sp} + 4}} and
\(y\) is a fresh temporary variable.
Previously, we mentioned that removing memory accesses does not reflect CT\@.
Fortunately, this pass is CT transparent because the addresses leaked by spills and unspills are
constant offsets from the stack pointer, which in turn is constant
throughout execution---recall that we do not consider function calls in
our language.

\begin{theorem}
  Unspilling is CT transparent.
\end{theorem}
\begin{proof}[Proof Sketch]
    Even though this pass maintains the control flow structure of the
    program, we cannot use \Cref{thm:ct:lock-step-soundness}
    since we need the relaxed version of the observation transformer: we
    need to inspect the program point in order to transform the observation.

    We define
    $\ensuremath{{\langle \ensuremath{{\mathit{c}}}, \rho, \mu \rangle}} \sim \ensuremath{{\langle \ensuremath{{\mathit{c}}}', \rho', \mu' \rangle}}$
    to hold when
    the output program's code is the transformed input program's code,
    \ensuremath{{\ensuremath{{\mathit{c}}}' = \ensuremath{{\llparenthesis\,{\ensuremath{{\mathit{c}}}}\,\rrparenthesis}}}};
    the variable maps \ensuremath{{\rho}}{} and \(\ensuremath{{\rho}}'\) coincide except on the fresh temporary
    variables introduced by the transformation;
    the memories \ensuremath{{\mu}}{} and \(\ensuremath{{\mu}}'\) coincide except on spilled stack offsets;
    and each spilled offset $n$ holds the value of its fresh variable
    $y$, i.e., \ensuremath{{\ensuremath{{\mu}}(\mathsf{sp} + n) = \ensuremath{{\rho}}'(y)}}.

    Since each output program step is simulated by exactly one input
    program step, we define \ensuremath{{\ensuremath{{\ensuremath{{\mathsf{ns}}}(s)}} \triangleq 1}} and
    \ensuremath{{\ensuremath{{\ensuremath{{\mathsf{sf}}}({c})}} \triangleq \ensuremath{{\epsilon}}}}.
    Finally, the observation transformer produces \ensuremath{{\bullet}}{} instead of
    \ensuremath{{\mathsf{adr}\; (\mathsf{sp} + n)}} observations.
    \begin{gather*}
        \Tobs{\ensuremath{{\mathit{c}}}}{\ensuremath{{\boldsymbol{o}}}} \triangleq
        \begin{cases}
            \ensuremath{{\bullet}}
          & \ensuremath{{\text{if } \ensuremath{{\mathit{c}}} \text{ is } \ensuremath{{{x}\mathcolor{cyan}{=}{\left[\mathsf{sp} + n\right]}}} \mathop{;} \ensuremath{{\mathit{c}}}'
            \text{ or } \ensuremath{{{\left[\mathsf{sp} + n\right]}\mathcolor{cyan}{=}{x}}} \mathop{;} \ensuremath{{\mathit{c}}}'}}
            \text{ and } \ensuremath{{\boldsymbol{o}}} = \ensuremath{{\mathsf{adr}\; (\mathsf{sp} + n)}}\text,
        \\
        \ensuremath{{\boldsymbol{o}}}
          & \text{if $c$ is not $\ensuremath{{{x}\mathcolor{cyan}{=}{\left[\mathsf{sp} + n\right]}}} \mathop{;} \ensuremath{{\mathit{c}}}'$ or
            $\ensuremath{{{\left[\mathsf{sp} + n\right]}\mathcolor{cyan}{=}{x}}} \mathop{;} \ensuremath{{\mathit{c}}}'$ and $\ensuremath{{\left|\ensuremath{{\boldsymbol{o}}}\right|}} = 1$,}
        \\
            \bot & \ensuremath{{\text{otherwise}}}\text.
        \end{cases}
    \end{gather*}

    The first two requirements hold straightforwardly.
    Injectivity on instruction spills and unspills instructions holds
    since the program point is the same, and otherwise trivially since the
    transformer returns its argument.
\end{proof}

\paragraph*{Structural Analysis}
This transformation is the core transformation in binary lifters,
which have the goal to take a binary input program and output a structured
program.
It takes an input program in control flow graph syntax
and identifies patterns that correspond to structures such as conditionals
and while loops.
Then, it replaces the structures found and outputs a program in structured syntax.
Structural Analysis sometimes fails when the input CFG program
contains patterns that originated from complex control flow structures such
as break statements.
In this work, we focus on the basic patterns of loops and conditional as
branches displayed in \Cref{fig:example-patterns}.
More involved analyses, like single-exit-single-successor
analysis \cite{DBLP:conf/scopes/EngelLAFB11}, can
accommodate more complex control flow, based on ideas such as tail
regions and iterative analysis \cite{DBLP:conf/uss/SchwartzL13}.

Structural Analysis is an example, where $\ensuremath{{\llparenthesis\,{\cdot}\,\rrparenthesis}} :
\ensuremath{{\mathcal{L}}}_s \to \ensuremath{{\mathcal{L}}}_t$ transforms the syntax of the programs, i.e.,
$\ensuremath{{\mathcal{L}}}_s \neq \ensuremath{{\mathcal{L}}}_t$.
For the output program's syntax $\ensuremath{{\mathcal{L}}}_t$, we use the structured programs
from \Cref{sec:setting}.
The input language $\ensuremath{{\mathcal{L}}}_s$ are control flow graphs that we sketch briefly.

\newcommand*{\refcfgassign}{\hyperref[fig:sem-cfg:single-step]{\textsc{\lblassign}}}
\newcommand*{\reflcfgoad}{\hyperref[fig:sem-cfg:single-step]{\textsc{\lblload}}}
\newcommand*{\refscfgtore}{\hyperref[fig:sem-cfg:single-step]{\textsc{\lblstore}}}
\newcommand*{\refscfgeqcond}{\hyperref[fig:sem-cfg:single-step]{\textsc{\lblseqcond}}}
\newcommand*{\refscfgpeccond}{\hyperref[fig:sem-cfg:single-step]{\textsc{\lblspeccond}}}

\begin{figure}
  \begin{subfigure}{0.3\linewidth}
    \centering
    \begin{tikzpicture}[>=stealth,thick,font=\small]
      \node at ( 0.0,  0.8) (start) {};
      \node at ( 0.0,  0.0) (cond)  {\(i < N\)};
      \node at ( 2.0,  0.5) (inc)   {\ensuremath{{{i}\mathcolor{cyan}{=}{i + 1}}}};
      \node at ( 2.0, -0.5) (dots1) {\dots};
      \node at ( 0.0, -0.8) (dots2) {\dots};

      \draw[->] (start) -- (cond);
      \draw[->] (cond) to[bend right] (dots1.west);
      \draw[->] (dots1) -- (inc);
      \draw[->] (inc.west) to[bend right] (cond);
      \draw[->] (cond) -- (dots2);
    \end{tikzpicture}

    \vspace{1em}

    \begin{tikzpicture}[>=stealth,thick]
      \node at ( 0.0,  0.8) (start) {};
      \node at ( 0.0,  0.0) (cond)  {\(x == 0\)};
      \node at ( 1.0, -0.8) (dots1) {\dots};
      \node at (-1.0, -0.8) (dots2) {\dots};

      \draw[->] (start) -- (cond);
      \draw[->] (cond) -- (dots1);
      \draw[->] (cond) -- (dots2);
    \end{tikzpicture}
  \caption{}\label{fig:example-patterns}
  \end{subfigure}\hfill \begin{subfigure}{0.65\linewidth}
  \small
  \begin{mathpar}
    \inferrule[\lblassign]{
      \ensuremath{{\ell: \ensuremath{{{x}\mathcolor{cyan}{=}{e}}} \mathrel{\text{\normalfont\guilsinglright}}\ell'}}
      \\\\
      \rho' = \rho[x \leftarrow \ensuremath{{\llbracket\, e \,\rrbracket_{\rho}}}]
    }{
      \ensuremath{{\langle \ell, \rho, \mu \rangle}}
      \step{}{\ensuremath{{\bullet}}}
      \ensuremath{{\langle \ell', \rho', \mu \rangle}}
    }

    \inferrule[\lblload]{
      \ensuremath{{\ell: \ensuremath{{{x}\mathcolor{cyan}{=}{\left[e\right]}}} \mathrel{\text{\normalfont\guilsinglright}}\ell'}}
      \\
      a = \ensuremath{{\llbracket\, e \,\rrbracket_{\rho}}}
      \\\\
      \rho' = \rho[x \leftarrow \mu(a)]
    }{
      \ensuremath{{\langle \ell, \rho, \mu \rangle}}
      \step{}{\ensuremath{{\mathsf{adr}\; a}}}
      \ensuremath{{\langle \ell', \rho', \mu \rangle}}
    }

    \inferrule[\lblstore]{
      \ensuremath{{\ell: \ensuremath{{{\left[e\right]}\mathcolor{cyan}{=}{x}}} \mathrel{\text{\normalfont\guilsinglright}}\ell'}}
      \\
      a = \ensuremath{{\llbracket\, e \,\rrbracket_{\rho}}}
      \\\\
      \mu' = \mu[a \leftarrow \rho(x)]
    }{
      \ensuremath{{\langle \ell, \rho, \mu \rangle}}
      \step{}{\ensuremath{{\mathsf{adr}\; a}}}
      \ensuremath{{\langle \ell', \rho, \mu' \rangle}}
    }

    \inferrule[\lblseqcond]{
      \ensuremath{{\ell: e \mathrel{\text{\normalfont\guilsinglright}}\ell_{\mathsf{t\kern-0.9ptt}}, \ell_{\mathsf{f\kern-1.1ptf}}}}
      \\
      b = \ensuremath{{\llbracket\, e \,\rrbracket_{\rho}}}
    }{
      \ensuremath{{\langle \ell, \rho, \mu \rangle}}
      \step{}{\ensuremath{{\mathsf{br}\; b}}}
      \ensuremath{{\langle \ell_{b}, \rho, \mu \rangle}}
    }
  \end{mathpar}
  \caption{}\label{fig:sem-cfg:single-step}
  \end{subfigure}
  \caption{
    (a) Example CFG patterns.
    The top pattern corresponds to a while loop and the bottom one to a
    conditional.
    (b) Semantics of the CFG language.}
\end{figure}

\paragraph*{CFG Language}
Our input language consists of assembly-like CFG programs.
A CFG program is a set $G$ of labelled nodes.
There are two types of nodes:
instruction nodes are triples $(\ell, a, \ell')$ where $\ell$ is the label of
the node, $a$ is an atomic command (i.e., it is an assignment
\ensuremath{{{x}\mathcolor{cyan}{=}{e}}}, a load \ensuremath{{{x}\mathcolor{cyan}{=}{\left[e\right]}}}, or a store \ensuremath{{{\left[e\right]}\mathcolor{cyan}{=}{x}}}),
and $\ell'$ is the label of the successor node.
Branching nodes $(\ell, e, \ell_\mathsf{t\kern-0.9ptt}, \ell_\mathsf{f\kern-1.1ptf})$ instead contain a
branching condition $e$, and two successors.
Each program has a distinguished initial and final label.
When $G$ is clear from context, we write \ensuremath{{\ell: i \mathrel{\text{\normalfont\guilsinglright}}\ell'}} when the
program contains the instruction node $(\ell, i, \ell') \in G$,
and \ensuremath{{\ell: e \mathrel{\text{\normalfont\guilsinglright}}\ell_{\mathsf{t\kern-0.9ptt}}, \ell_{\mathsf{f\kern-1.1ptf}}}} when it contains the
branching node $(\ell, e, \ell_{\mathsf{t\kern-0.9ptt}}, \ell_\mathsf{f\kern-1.1ptf}) \in G$.

The semantics of this language operates on states
\ensuremath{{\langle \ell, \ensuremath{{\rho}}, \ensuremath{{\mu}} \rangle}} consisting of
the label~$\ell$ which is the current program point,
$\ensuremath{{\ensuremath{{\mathsf{pc}}}({\ensuremath{{\langle \ell, \ensuremath{{\rho}}, \ensuremath{{\mu}} \rangle}}})}} = \ell$,
a register map~$\rho$ that associates register variables with values,
and a memory~$\mu$ that associates addresses to values.
We provide its leakage semantics \sem{s}{}{o}{s'} in
\Cref{fig:sem-cfg:single-step}. It is straightforward to see that the
CFG semantics satisfies the
constraints from our framework in \Cref{sec:setting}.

\begin{theorem}
    Structural Analysis is CT transparent.
\end{theorem}
\begin{proof}[Proof Sketch]
We write \ensuremath{{\ensuremath{{\mathit{struct}_{G}}}(\ell)}} for the structured code
that Structural Analysis identified to be the matching structured program
for \(\ell\) in the CFG program \(G\).
As mentioned above, the precise definition of $\ensuremath{{\ensuremath{{\mathit{struct}_{G}}}(\ell)}}$ can
be found in \ifAppendix{\cref{appendix:structural-analysis}}\fi{}.

In order to prove transparency we define a lock-step simulation
and apply \Cref{thm:ct:lock-step-soundness}.
For a input CFG $G$,
the simulation relation~\(\sim\) is equality on register and memory contents,
while the label of the input program is replaced by its structured program,
$\ensuremath{{\langle \ell, \rho, \mu \rangle}} \sim \ensuremath{{\langle \ensuremath{{\ensuremath{{\mathit{struct}_{G}}}(\ell)}}, \rho, \mu \rangle}}$.
We choose the identity as the observation transformer.

\Cref{thm:ct:lock-step-soundness:simulation} follows from the fact that 
$\ensuremath{{\mathit{struct}_{G}}}$ ensures the next instruction to be executed in $\ell$ always coincides with $\ensuremath{{\ensuremath{{\mathit{struct}_{G}}}(\ell)}}$.
That fact gives rise to the following proposition that immediately implies \Cref{thm:ct:lock-step-soundness:simulation}.

\begin{proposition}
  For any pair of related states \ensuremath{{s \sim t}},
  the input state steps \sem{s}{}{o}{s'} if and only if
  the output state steps \sem{t}{}{o}{t'}, and,
  furthermore, \ensuremath{{s' \sim t'}}.
\end{proposition}

The other two requirements are satisfied as well:
initial states that share the same inputs are related trivially,
because the initial label $\ell_\ensuremath{{\mathit{init}}}$ of $G$ maps to the initial program code of $\ensuremath{{\llparenthesis\,{G}\,\rrparenthesis}}$
and the contents of registers and memory coincide;
and \Tobsname{} is injective because it is the identity.
\end{proof}

\paragraph*{Loop Rotation}
This transformation is used to move the entry point of a loop to a
desired location in the loop body.
\Cref{fig:loop-rotate} illustrates the transformation:
the left-hand side presents the input program, where the boxed
instruction is the loop entry, and the right-hand side is the output
program.
The transformation duplicates the loop entry, $\textit{cond}$, and makes
its successor, $\textit{inst}_1$, the new loop entry.
In practice, loops are rotated multiple times in succession in
order to move the entry node to a desired position.

This transformation operates on CFG programs, i.e., its input and output
are CFG programs.
Formally, a loop of an input CFG program $G$ is a triple $(\ell_{\textit{pred}}, \ell_{\textit{entry}}, L)$,
where $L$ is the set of node labels that form the loop,
and $\ell_{\textit{pred}}$ is the label of the only node outside of $L$ that has a
successor in~$L$, namely the entry point $\ell_{\textit{entry}} \in L$.
For simplicity, we require that the loop predecessor $\ell_{\textit{pred}}$ has only $\ell_{\textit{entry}}$
as a successor, i.e., it is an instruction node $\ensuremath{{\ell_{\textit{pred}}: i \mathrel{\text{\normalfont\guilsinglright}}\ell_{\textit{entry}}}}$.
Given a loop $(\ell_{\textit{pred}}, \ell_{\textit{entry}}, L)$ to be rotated, the transformation
$\ensuremath{{\llparenthesis\,{G}\,\rrparenthesis}}$ extends the CFG program $G$ with a node labeled $\ell_{\textit{copy}}$
which is a copy of the $\ell_{\textit{entry}}$ node apart from its label.
It also changes the \ensuremath{{\ell_{\textit{pred}}: i \mathrel{\text{\normalfont\guilsinglright}}\ell_{\textit{entry}}}} node to
$\ensuremath{{\ell_{\textit{pred}}: i \mathrel{\text{\normalfont\guilsinglright}}\ell_{\textit{copy}}}}$, so that $\ell_{\textit{copy}}$ is executed before the loop.

\begin{theorem}
    Loop Rotation is CT transparent.
\end{theorem}
\begin{proof}[Proof Sketch]
  We prove the transparency of Loop Rotation with a lock-step
  simulation and \Cref{thm:ct:lock-step-soundness}.
  We instantiate the simulation relation $\sim$ so that it holds
  on $\ensuremath{{\langle \ell_s, \ensuremath{{\rho}}, \ensuremath{{\mu}} \rangle}} \sim \ensuremath{{\langle \ell_t, \ensuremath{{\rho}}, \ensuremath{{\mu}} \rangle}}$,
  when either $\ell_s = \ell_t$, or $\ell_s = \ell_{\textit{entry}}$ and $\ell_t = \ell_{\textit{copy}}$.
  This means that the register map $\ensuremath{{\rho}}$ and memory $\ensuremath{{\mu}}$ coincide
  exactly.
  The observation transformer is the identity.
  \Cref{thm:ct:lock-step-soundness:simulation} holds because program
  locations also coincide except for when the output program is at label
  $\ell_{\textit{copy}}$: in that case, the input program is at $\ell_{\textit{entry}}$, which has
  the same behavior by definition of \(\ell_{\textit{copy}}\).
  \Cref{thm:ct:lock-step-soundness:initial,thm:ct:lock-step-soundness:injectivity} hold trivially.
\end{proof}

\begin{figure}
    \begin{subfigure}{0.45\linewidth}
        \begin{center}
            \begin{tikzpicture}[
                scale=0.8,every node/.style={transform shape},
                >=stealth,thick
              ]
                \node at ( -3, 1.5) (entry) {};
                \node[draw,minimum width=1.5cm] at ( 0, 0) (header)  {\textit{cond}};
                \node at ( 2, 1.5) (i1)   {$\textit{inst}_1$};
                \node at ( 2, 0.5) (i2) {$\textit{inst}_2$};
                \node at ( 2, -0.5) (idots) {\dots};
                \node at ( 2, -1.5) (in) {$\textit{inst}_n$};
                \node at ( -2, -1.5) (exit) {\textit{exit}};

                \draw[->] (entry) -| ([shift={(-0.3,0)}]cond.north);
                \draw[->] (i1) -- (i2);
                \draw[->] (i2) -- (idots);
                \draw[->] (idots) -- (in);
                \draw[->] (in.west) -| ([shift={(0.3,0)}]header.south);
                \draw[->] (header.north) ++ (0.3,0) |- (i1);
                \draw[->] (header.south) ++ (-0.3,0) |- (exit.east);
            \end{tikzpicture}
        \end{center}
    \end{subfigure}
    \begin{subfigure}{0.45\linewidth}
        \begin{center}
            \begin{tikzpicture}[
                scale=0.8,every node/.style={transform shape},
                >=stealth,thick
              ]
                \node at ( -3, 1.5) (entry) {};
                \node at ( -2, 1.5) (cond-h) {\textit{cond}};
                \node[draw,minimum width=1.5cm] at ( 0, 0) (header)  {$\textit{inst}_1$};
                \node at ( 2, 1.5) (i1)  {$\textit{inst}_2$} ;
                \node at ( 2, 0.5) (i2) {\dots};
                \node at ( 2, -0.5) (idots) {$\textit{inst}_n$};
                \node at ( 2, -1.5) (in) {\textit{cond}};
                \node at ( -2, -1.5) (exit) {\textit{exit}};

                \draw[->] (entry) -- (cond-h);
                \draw[->] (cond-h) -| ([shift={(-0.3,0)}]cond.north);
                \draw[->] (cond-h) -- (exit);
                \draw[->] (i1) -- (i2);
                \draw[->] (i2) -- (idots);
                \draw[->] (idots) -- (in);
                \draw[->] (in.west) ++ (0,0.05) -| ([shift={(0.3,0)}]header.south);
                \draw[->] (header.north) ++ (0.3,0) |- (i1);
                \draw[->] (in.west) ++ (0,-0.08) -- ([shift={(0,-0.08)}]exit.east);
            \end{tikzpicture}
        \end{center}
    \end{subfigure}
    \caption{Loop Rotation duplicates the loop entry and then moves the entry point of the loop by one instruction.}
    \label{fig:loop-rotate}
\end{figure}

\subsection{Nontransparent Transformations}\label{sec:nontranspasse}
This section presents the program transformations from
\cref{tab:passes} that do not reflect CT, providing minimal
counterexamples to illustrate the kind of removed CT violations
removed.
In all examples, the input program to the transformation leaks the value of a
secret variable \jazz{sec} and the transformation removes this CT violation.
As discussed in \cref{sec:decompiler-eval}, we used the counterexamples to test
real-world decompilers for reflection failures.

\paragraph*{If Conversion}
If Conversion simplifies the control flow of a program by converting a
conditional branch into an equivalent branchless conditional move statement,
which are supported by many ISAs.
For example, \texttt{x86-64} provides the \jazz{CMOVcc} instruction to
conditionally update a register when a condition holds.
Contrary to conditional branches, conditional moves do not leak their condition.
\Cref{lst:clangover} presented the counterexample to CT reflection,
as explained in \cref{sec:motivating-example}.

\paragraph*{Branch Coalescing}
This transformation removes conditional branches when their branches are
identical.
\Cref{lst:branch_coalescing} presents an example of Branch Coalescing:
the input program leaks the value of \texttt{sec} via the conditional branch,
but the in output program the branch was coalesced.
This means that the output program does not leak \jazz{sec}.
A special case of this transformation present in many decompilers is
Empty Branch Coalescing (\cref{lst:empty_branch_coalescing}), which
coalesces only on empty branches.

\begin{figure}
  \begin{subfigure}[t]{0.49\linewidth}
    \centering
    \begin{jasmincode}[outerpos=t,outerwidth=26ex]
      \jasminindent{0}\jasminkw{if} (sec) \jasminopenbrace{} x = \jasminconstant{42}; \jasminclosebrace{}\\
      \jasminindent{0}\jasminkw{else} \jasminopenbrace{} x = \jasminconstant{42}; \jasminclosebrace{}
    \end{jasmincode}
    \begin{jasmincode}[outerpos=t,outerwidth=12ex]
      \jasminindent{0}\\
      \jasminindent{0}x = \jasminconstant{42};
    \end{jasmincode}
    \setcounter{subfigure}{0}
    \subcaption{Branch Coalescing.}\label{lst:branch_coalescing}
  \end{subfigure}\hfill \begin{subfigure}[t]{0.48\linewidth}
    \centering
    \begin{jasmincode}[outerpos=t,outerwidth=16ex]
      \jasminindent{0}x = [sec];\\
      \jasminindent{0}x = \jasminconstant{42};
    \end{jasmincode}
    \begin{jasmincode}[outerpos=t,outerwidth=12ex]
      \jasminindent{0}\\
      \jasminindent{0}x = \jasminconstant{42};
    \end{jasmincode}
    \setcounter{subfigure}{2}
    \subcaption{Dead Load Elimination.}\label{lst:dead_load}
  \end{subfigure}

  \bigskip

  \begin{subfigure}[t]{0.49\linewidth}
    \centering
    \begin{jasmincode}[outerpos=t,outerwidth=26ex]
      \jasminindent{0}\jasminkw{if} (sec) \jasminopenbrace{} \jasminclosebrace{}
      \jasminkw{else} \jasminopenbrace{} \jasminclosebrace{}
    \end{jasmincode}
    \begin{jasmincode}[outerpos=t,outerwidth=12ex]
      \jasminindent{0}
    \end{jasmincode}
    \setcounter{subfigure}{1}
    \subcaption{Empty Branch Coalescing.}\label{lst:empty_branch_coalescing}
  \end{subfigure}\hfill \begin{subfigure}[t]{0.48\linewidth}
    \centering
    \begin{jasmincode}[outerpos=t,outerwidth=16ex]
      \jasminindent{0}z = [sec];\\
      \jasminindent{0}y = \jasminconstant{42};\\
      \jasminindent{0}[sec] = z;
    \end{jasmincode}
    \begin{jasmincode}[outerpos=t,outerwidth=12ex]
      \jasminindent{0}\\
      \jasminindent{0}y = \jasminconstant{42};\\
      \jasminindent{0}
    \end{jasmincode}
    \setcounter{subfigure}{3}
    \subcaption{Dead Store Elimination.}\label{lst:self_store}
  \end{subfigure}
  \caption{Counterexamples for nontransparent passes.}\label{fig:counterexamples}
\end{figure}

\paragraph*{Memory Access Elimination}
This is a family of transformations that remove unnecessary memory
accesses.
One example from this family is Dead Load Elimination
(\Cref{lst:dead_load}), where a value loaded from memory is never used.
Another example is Dead-Store Elimination (\Cref{lst:self_store}),
where a store is guaranteed to not change the value of its location.
In both listings, \jazz{sec} leaks due to a memory access, but the
output programs have no such access.
Thus, the value is no longer leaked.

\section{Improving the Transparency of Practical Decompilers}
\label{sec:ct-retdec}

In this section, we aim to improve a practical decompiler, \textsc{RetDec}{},
in order to enable better detection of CT violations.
We modify \textsc{RetDec}{} by disabling passes that violate CT transparency.
To guide our search for the passes to keep and to discard we leverage our theoretical findings from \Cref{sec:passes}.
Because \textsc{RetDec}{} employs a total of 62 passes, we fall back to empirical methods to decide which passes to disable for the passes not covered in \Cref{sec:passes}.
This means that there are two reasons that our modified version of \textsc{RetDec}{} may still not be transparent:
the empirical tests do not guarantee that passes are transparent, and even the passes shown to be transparent in \Cref{sec:passes} might contain implementation bugs.
However, as we will see in our evaluation of the modified version of \textsc{RetDec}{}, this fact does not diminish the usefulness of improving the transparency of \textsc{RetDec}{}.

We first give an overview of \textsc{RetDec}{}, describe the transformations it
performs and discuss their coverage by passes from~\Cref{sec:passes}.
We then turn to our empirical methods for other passes.
Finally, we combine our modified version of \textsc{RetDec}{} with
CT-LLVM~\cite{cryptoeprint:2025/338},
a CT analysis tool for LLVM-IR, to detect CT violations in binaries.

\subsection{RetDec Transformations}\label{sec:retdec-overview}
\textsc{RetDec}{} is a state-of-the-art machine code decompiler that outputs LLVM IR.
It generates well-readable code due to its reliance on LLVM's analysis
and optimization passes. We chose to modify \textsc{RetDec}{} over other
decompilers, because it has a well-structured decompilation pass management.
All transformation passes are organized in a configuration JSON file.
This makes it a good fit to enable or disable decompiler passes at will.

\textsc{RetDec}{}'s decompilation process takes three steps: convert binary code to
LLVM IR, optimize LLVM IR, and convert LLVM IR to C code.
Overall, \textsc{RetDec}{} uses 62 distinct passes to decompile and optimize the code:
27 passes are implemented by \textsc{RetDec}{} itself and the remaining 35 passes are
borrowed from LLVM\@.
Most \textsc{RetDec}{} builtin passes do not perform extensive program transformation.
The big exception is the lifting/structural analysis from binary to LLVM IR (\textit{retdec-decoder}),
and some minor exceptions being refinement of the decompilation output (e.g., \textit{retdec-simple-types})
or user-configurable passes (e.g., \textit{retdec-select-funcs}).
We categorize the remaining LLVM passes into six groups,
and discuss, which of them and how they are covered by our theoretical findings from \cref{sec:passes}.
A complete list of all passes, their inclusion in our modified version of \textsc{RetDec}{},
and which passes are covered by \cref{sec:passes} is available in \ifAppendix{\cref{appendix:retdec-passes}}\fi{}.

\paragraph*{Code Simplification/Elimination}
This category contains ten passes that simplify or eliminate code, 
such as \textit{adce}, \textit{dse}, \textit{early-cse} and \textit{bdce}, which remove dead code.
This includes the removal of dead load/store operations and dead branches.
We demonstrated that removing dead memory operations makes these passes nontransparent (\Cref{lst:dead_load,lst:self_store}).
One exception is the \textit{strip-dead-prototypes} pass, which removes
unused function prototypes. This transformation is similar to Dead Branch
Elimination from~\Cref{sec:passes}, which we proved transparent.

\paragraph*{Loop Optimizations}
This category contains ten passes that optimize loops.
Some passes are used to canonicalize loops or replace them with non-loop forms.
For example, \textit{loop-rotate} is used to convert do-while loops to while
loops, and \textit{loop-idiom} is used to transform simple loops into a
non-loop form (e.g., replace a loop with a \texttt{memcpy} call).
We have proved in \Cref{sec:passes} that Loop Rotation is transparent.
However, some of the other passes are not transparent. For example, \textsc{RetDec}{}
invokes \textit{loop-deletion} to remove loops that have
no side-effects to memory and do not contribute to the program output.
The removal of such loops can remove CT violations similar to Memory Access
Elimination from \Cref{sec:passes}.

\paragraph*{Expression Substitution}
This category contains four passes that perform various expression substitution
optimizations. Specifically, \textsc{RetDec}{} invokes \textit{constprop},
\textit{correlated-propagation} and \textit{scp} to propagate constants,
and \textit{constmerge} to merge duplicate constants.
These passes are special cases of Expression Substitution from \Cref{sec:passes}, which is transparent.

\paragraph*{Control Flow Optimization}
This category contains four passes that simplify the control flow graph of
the program. Two passes, \textit{simplifycfg} and \textit{loop-simplifycfg},
are used to remove unreachable blocks and empty branches,
which we have shown to not be transparent (\Cref{lst:branch_coalescing}).
Another pass, \textit{jump-threading}, is used to remove redundant branches,
whose conditions are known to be constant at compile time.
Removing branches with constant conditions is precisely Dead Branch Elimination
from~\Cref{sec:passes}, which we have proven transparent.
Finally, \textit{sink} is used to move instructions to the blocks where they
are really used. This pass is not transparent as it could potentially move
memory accesses into a dead branch,
so the memory access is effectively eliminated (\Cref{lst:dead_load,lst:self_store}).

\paragraph*{Stack Optimization}
This category contains one pass, \textit{mem2reg}, that promotes stack
variables to register variables. This is similar to our Unspilling pass,
which we have proved to be transparent.

\paragraph*{LLVM Analysis}
This category contains six passes that do not transform the program but
provide analysis results to other passes.
Since they neither remove nor introduce code, they are transparent.

\subsection{Empirical Transparency}

We now focus on our empirical methodology to test the transparency of the
passes that are not covered by the transformations in~\Cref{sec:passes},
and present how we modified \textsc{RetDec}{}.

\paragraph*{Methodology}
In a nutshell, we empirically test the transparency of \textsc{RetDec}{} passes on
carefully crafted binaries that contain CT violations.
Technically, we select which passes are enabled in \textsc{RetDec}{},
feed it with the binaries,
and run CT-LLVM~\cite{cryptoeprint:2025/338} to check for CT violations
in the decompiled code.
CT-LLVM is a recently published CT analysis tool for LLVM-IR, which we chose
due to its precision and usability.
In principle, however, any source-level or LLVM-level CT analysis could be used.
We call a set of transformation passes \emph{empirically transparent} if the
CT analysis tool finds the same CT violations that were contained in the
original binaries.

\paragraph*{Empirical Test Cases}
Our set of empirical test cases checks for both reflection and preservation
of the \textsc{RetDec}{} passes not covered by \Cref{sec:passes}.
For reflection, we use our test cases from
\Cref{lst:clangover,lst:branch_coalescing,lst:empty_branch_coalescing,lst:dead_load,lst:self_store}.
These cover three common failures to reflection: Branch Coalescing,
If Conversion and Dead Load/Store Elimination.
For preservation, we construct one test case according to a non-preservation 
case reported by a recent work~\cite{ZhouDZ24}, which we present in~\cref{lst:invif_conv}.
Specifically, this test case checks whether \textsc{RetDec}{} reverts If Conversion 
by replacing a conditional move by a conditional branch, which 
introduces a false positive in CT analysis.

\begin{figure}
  \centering
  \begin{jasmincode}[outerpos=t,outerwidth=24ex]
    \jasminindent{0}mov    rax, \jasminconstant{0x3}\\
    \jasminindent{0}mov    rcx, \jasminconstant{0x5}\\
    \jasminindent{0}cmp    rdi, \jasminconstant{0x1}\\
    \jasminindent{0}cmovne rax, rcx\\
    \jasminindent{0}ret
  \end{jasmincode}\hspace{4em}\begin{jasmincode}[outerpos=t,outerwidth=24ex]
    \jasminindent{0}\phantom{T1: }cmp rdi, \jasminconstant{0x1}\\
    \jasminindent{0}\phantom{T1: }jne T1\\
    \jasminindent{0}\phantom{T1: }mov rax, \jasminconstant{0x3}\\
    \jasminindent{0}\phantom{T1: }ret\\
    \jasminindent{0}T1: mov rax, \jasminconstant{0x5}\\
    \jasminindent{0}\phantom{T1: }ret
  \end{jasmincode}
  \caption{Inverse If Conversion.}\label{lst:invif_conv}
\end{figure}

\paragraph*{Modifying \textsc{RetDec}{}}
We modify \textsc{RetDec}{} in two steps.
First, we build a minimal version of \textsc{RetDec}{} containing only
version by removing all passes that are not required for \textsc{RetDec}{} to function.
Seven necessary passes remain:
\textit{retdec-provider-init}, which is used to initialize the decompiler,
\textit{retdec-decoder}, an actual transforming pass, which translates assembly
ISA instructions to LLVM-IR instructions,
and four passes \textit{retdec-write-ll}, \textit{retdec-write-bc},
\textit{retdec-write-dsm}, and \textit{retdec-llvmir2hll}
that are essential for generating outputs.
Besides, we keep the \textit{retdec-param-return} pass in the minimal \textsc{RetDec}{}
because it reconstructs function arguments
so that CT analyses can be applied.
Minimal \textsc{RetDec}{} is empirically transparent with our test cases.

In order to improve minimal \textsc{RetDec}{}, we employ an incremental testing strategy.
We iteratively enable passes to test their empirical transparency.
When we identify a pass that violates CT transparency, we discard it.
We also test the passes that we have proven transparent,
because a ``buggy'' implementation could also invalidate our theoretical findings.
We do not find any pass that we have shown to be transparent fail any test case,
though. Meanwhile, we identify 10 passes that are not transparent.
Half of them, \textit{adce}, \textit{dse}, \textit{simplifycfg},
\textit{early-cse} and \textit{bdce},
have already been discussed in~\Cref{sec:retdec-overview}.
Among the other half, we find that using \textit{globalopt}, \textit{gvn} or
\textit{instcombine} removes
the empty branch in~\Cref{lst:empty_branch_coalescing}, while using
\textit{retdec-inst-opt-rda} or
\textit{reassociate} removes the dead load operation in~\Cref{lst:dead_load}.
As a result, we discard these ten nontransparent passes.

\subsection{Evaluation: Finding CT Violations with CT-RetDec}
We use our modified version of \textsc{RetDec}{} for the Decompile-then-Analyze approach:
We build a new binary-level CT analysis tool, called \textsc{CT-RetDec}{},
which combines the modified \textsc{RetDec}{} version with the analysis tool CT-LLVM.
In this section, we evaluate the performance of \textsc{CT-RetDec}{} on a benchmark of
binaries.

\paragraph*{Benchmark Set}
We construct the benchmark set with timing side-channel leakages reported by 
previous works~\cite{clangover,schneider2024breaking,DanielBR20,DBLP:conf/eurosp/SimonCA18}.
Specifically, we include the Clangover vulnerability~\cite{clangover},
five vulnerabilities in selection algorithms~\cite{DBLP:conf/eurosp/SimonCA18}, 
two vulnerabilities in sorting algorithms~\cite{DanielBR20}, 
the check scalar vulnerability in BearSSL~\cite{schneider2024breaking} and the 
CMOVNZ vulnerability in HACL*~\cite{schneider2024breaking}.
Our benchmark set takes advantage of the fact that we are employing the
Decompile-then-Analyze approach:
eight of the listed vulnerabilities are compiler-induced,
i.e., the original source code does not exhibit the vulnerability, but only
the compiled binary.

We create our benchmark by compiling the source code of the vulnerabilities
with Clang in different versions, under different optimization parameters,
and different target architectures.
Different from previous works, we compile the source code with newer compiler versions, namely Clang-10, Clang-14, Clang-18 and Clang-21,
two optimization levels (no optimization \texttt{-O0} and optimization for space \texttt{-Os}),
and two different platforms (\texttt{x86-32} and \texttt{x86-64}).
This results in 160 distinct binaries to analyze.
We then use \textsc{CT-RetDec}{} to detect which configurations of version,
optimization parameters, and target architecture induce the vulnerabilities.

\paragraph*{Experiment Results}
We evaluated the 160 binaries using \textsc{CT-RetDec}{}.
First, to obtain ground truth, we manually inspected the assembly code of all
binaries to establish whether they contain a CT violation.
We present the analysis results in~\Cref{tab:ctretdec}.
The prefix \texttt{ct\_} of the test cases indicates that the source code is CT,
while those prefixed \texttt{nct\_} have a CT violation in the source code.
We mark the result of \textsc{CT-RetDec}{}
with \textcolor{red}{\ding{55}}{} if it reports a CT violation,
and with \textcolor{green}{\ding{51}}{} if it reports no CT violation (the background shades can
be ignored for now).
We confirm that \textsc{CT-RetDec}{} correctly find all CT violations while not 
reporting any false positives.
Hence, the results of \textsc{CT-RetDec}{} match the established ground truth.

The results show that \textsc{CT-RetDec}{} correctly captures the difference among
different compiler versions. For example, it shows that the Clangover
vulnerability is not present in binaries compiled with Clang-10 and Clang-14,
while it is present in those compiled with Clang-18 and Clang-20.
Second, \textsc{CT-RetDec}{} correctly captures the difference between 32-bit and 64-bit
binaries. For example, it shows the constant-time selection algorithms
(e.g., \texttt{ct\_select\_v1}) are always constant-time for 64-bit binaries,
but not for 32-bit binaries when using \texttt{-Os} optimization.
This is likely due to architecture-specific optimizations performed by the compiler.
Third, \textsc{CT-RetDec}{} is precise: for all the CT binaries, \textsc{CT-RetDec}{} preserves
CT to the decompiled program. For example, it shows the vulnerabilities in 
\texttt{nct\_select} and \texttt{nct\_sort} disappear in 64-bit binaries obtained with 
\texttt{-Os} optimizations. These vulnerabilities disappear because compiler replaces 
the conditional branches with conditional moves, which are constant-time.
To summarize, we show that \textsc{CT-RetDec}{} is effective in analyzing binaries and
finding compiler-induced CT violations.

\paragraph*{Comparison with Unmodified RetDec+CT-LLVM}
We now show that \textsc{CT-RetDec}{} outperforms the combination of unmodified \textsc{RetDec}{} and CT-LLVM.
Specifically, we find that the nontransparent \textsc{RetDec}{} does miss most CT violations in the benchmark set.
We make the comparison by running the unmodified \textsc{RetDec}{} + CT-LLVM on the benchmark set.
We highlighted all CT violations missed by nontransparent \textsc{RetDec}{} in \Cref{tab:ctretdec} 
with a shaded green cell \colorbox{green!15}{\textcolor{red}{\ding{55}}}.
As we can see, it can only find leakages in the sorting algorithms, \texttt{ct\_check\_scalar} and \texttt{ct\_hacl\_cmovzn4} under some configurations.
Our comparison results highlight the importance of transparent decompilation for finding CT violations.

\begin{table}
\caption{\textsc{CT-RetDec}{} analysis results. 
CT-RetDEc correctly finds all CT violations (marked with \textcolor{red}{\ding{55}}) and does not 
report false positives for CT binaries (marked with \textcolor{green}{\ding{51}}).}
\label{tab:ctretdec}
\centering
\small
\footnotesize
\begin{tabular}{lcccccccccccccccc}
\toprule

& \multicolumn{4}{c}{\textbf{Clang-10.0.0}} &
\multicolumn{4}{c}{\textbf{Clang-14.0.6}} &
\multicolumn{4}{c}{\textbf{Clang-18.1.8}} &
\multicolumn{4}{c}{\textbf{Clang-21.1.1}} \\

\cmidrule(lr){2-5}\cmidrule(lr){6-9}\cmidrule(lr){10-13}\cmidrule(lr){14-17}

& \multicolumn{2}{c}{\textbf{64-bit}} & \multicolumn{2}{c}{\textbf{32-bit}} &
\multicolumn{2}{c}{\textbf{64-bit}} & \multicolumn{2}{c}{\textbf{32-bit}} &
\multicolumn{2}{c}{\textbf{64-bit}} & \multicolumn{2}{c}{\textbf{32-bit}} &
\multicolumn{2}{c}{\textbf{64-bit}} & \multicolumn{2}{c}{\textbf{32-bit}} \\
\cmidrule(lr){2-3}\cmidrule(lr){4-5}\cmidrule(lr){6-7}\cmidrule(lr){8-9}\cmidrule(lr){10-11}\cmidrule(lr){12-13}\cmidrule(lr){14-15}\cmidrule(lr){16-17}

Test case
& \texttt{O0}{} & \texttt{Os}{} & \texttt{O0}{} & \texttt{Os}{} &
\texttt{O0}{} & \texttt{Os}{} & \texttt{O0}{} & \texttt{Os}{} &
\texttt{O0}{} & \texttt{Os}{} & \texttt{O0}{} & \texttt{Os}{} &
\texttt{O0}{} & \texttt{Os}{} & \texttt{O0}{} & \texttt{Os}{} \\
\midrule
\texttt{ct\_clangover}  & \textcolor{green}{\ding{51}} & \textcolor{green}{\ding{51}} & \textcolor{green}{\ding{51}} & \textcolor{green}{\ding{51}} & \textcolor{green}{\ding{51}} & \textcolor{green}{\ding{51}} & \textcolor{green}{\ding{51}} & \textcolor{green}{\ding{51}} & \textcolor{green}{\ding{51}} & \cellcolor{green!15}\textcolor{red}{\ding{55}} & \textcolor{green}{\ding{51}} & \cellcolor{green!15}\textcolor{red}{\ding{55}} & \textcolor{green}{\ding{51}} & \cellcolor{green!15}\textcolor{red}{\ding{55}} & \textcolor{green}{\ding{51}} & \cellcolor{green!15}\textcolor{red}{\ding{55}} \\
\texttt{ct\_select\_v1}   & \textcolor{green}{\ding{51}} & \textcolor{green}{\ding{51}} & \textcolor{green}{\ding{51}} & \cellcolor{green!15}\cellcolor{green!15}\textcolor{red}{\ding{55}} & \textcolor{green}{\ding{51}} & \textcolor{green}{\ding{51}} & \textcolor{green}{\ding{51}} & \cellcolor{green!15}\textcolor{red}{\ding{55}} & \textcolor{green}{\ding{51}} & \textcolor{green}{\ding{51}} & \textcolor{green}{\ding{51}} & \cellcolor{green!15}\textcolor{red}{\ding{55}} & \textcolor{green}{\ding{51}} & \textcolor{green}{\ding{51}} & \textcolor{green}{\ding{51}} & \cellcolor{green!15}\textcolor{red}{\ding{55}} \\
\texttt{ct\_select\_v2}   & \textcolor{green}{\ding{51}} & \textcolor{green}{\ding{51}} & \textcolor{green}{\ding{51}} & \cellcolor{green!15}\textcolor{red}{\ding{55}} & \textcolor{green}{\ding{51}} & \textcolor{green}{\ding{51}} & \textcolor{green}{\ding{51}} & \cellcolor{green!15}\textcolor{red}{\ding{55}} & \textcolor{green}{\ding{51}} & \textcolor{green}{\ding{51}} & \textcolor{green}{\ding{51}} & \cellcolor{green!15}\textcolor{red}{\ding{55}} & \textcolor{green}{\ding{51}} & \textcolor{green}{\ding{51}} & \textcolor{green}{\ding{51}} & \cellcolor{green!15}\textcolor{red}{\ding{55}} \\
\texttt{ct\_select\_v3}   & \textcolor{green}{\ding{51}} & \textcolor{green}{\ding{51}} & \textcolor{green}{\ding{51}} & \cellcolor{green!15}\textcolor{red}{\ding{55}} & \textcolor{green}{\ding{51}} & \textcolor{green}{\ding{51}} & \textcolor{green}{\ding{51}} & \cellcolor{green!15}\textcolor{red}{\ding{55}} & \textcolor{green}{\ding{51}} & \textcolor{green}{\ding{51}} & \textcolor{green}{\ding{51}} & \cellcolor{green!15}\textcolor{red}{\ding{55}} & \textcolor{green}{\ding{51}} & \textcolor{green}{\ding{51}} & \textcolor{green}{\ding{51}} & \cellcolor{green!15}\textcolor{red}{\ding{55}} \\
\texttt{ct\_select\_v4}   & \textcolor{green}{\ding{51}} & \textcolor{green}{\ding{51}} & \textcolor{green}{\ding{51}} & \cellcolor{green!15}\textcolor{red}{\ding{55}} & \textcolor{green}{\ding{51}} & \textcolor{green}{\ding{51}} & \textcolor{green}{\ding{51}} & \cellcolor{green!15}\textcolor{red}{\ding{55}} & \textcolor{green}{\ding{51}} & \textcolor{green}{\ding{51}} & \textcolor{green}{\ding{51}} & \cellcolor{green!15}\textcolor{red}{\ding{55}} & \textcolor{green}{\ding{51}} & \textcolor{green}{\ding{51}} & \textcolor{green}{\ding{51}} & \cellcolor{green!15}\textcolor{red}{\ding{55}} \\
\texttt{ct\_sort}	 & \textcolor{green}{\ding{51}} & \textcolor{green}{\ding{51}} & \textcolor{green}{\ding{51}} & \cellcolor{green!15}\textcolor{red}{\ding{55}} & \textcolor{green}{\ding{51}} & \textcolor{green}{\ding{51}} & \textcolor{green}{\ding{51}} & \cellcolor{green!15}\textcolor{red}{\ding{55}} & \textcolor{green}{\ding{51}} & \textcolor{red}{\ding{55}} & \textcolor{green}{\ding{51}} & \textcolor{red}{\ding{55}} & \textcolor{green}{\ding{51}} & \textcolor{green}{\ding{51}} & \textcolor{green}{\ding{51}} & \cellcolor{green!15}\textcolor{red}{\ding{55}} \\
\texttt{ct\_check\_scalar} & \textcolor{green}{\ding{51}} & \textcolor{green}{\ding{51}} & \textcolor{green}{\ding{51}} & \textcolor{green}{\ding{51}} & \textcolor{green}{\ding{51}} & \cellcolor{green!15}\textcolor{red}{\ding{55}} & \textcolor{green}{\ding{51}} &  \cellcolor{green!15}\textcolor{red}{\ding{55}} & \textcolor{green}{\ding{51}} & \cellcolor{green!15}\textcolor{red}{\ding{55}} & \textcolor{green}{\ding{51}} & \cellcolor{green!15}\textcolor{red}{\ding{55}} & \textcolor{green}{\ding{51}} & \textcolor{red}{\ding{55}} & \textcolor{green}{\ding{51}} & \textcolor{red}{\ding{55}} \\
\texttt{ct\_hacl\_cmovzn4} & \textcolor{green}{\ding{51}} & \textcolor{green}{\ding{51}} & \textcolor{green}{\ding{51}} & \textcolor{green}{\ding{51}} & \textcolor{green}{\ding{51}}  & \textcolor{green}{\ding{51}} & \textcolor{green}{\ding{51}} & \cellcolor{green!15}\textcolor{red}{\ding{55}} & \textcolor{green}{\ding{51}} & \textcolor{green}{\ding{51}} & \textcolor{green}{\ding{51}} & \textcolor{red}{\ding{55}} & \textcolor{green}{\ding{51}} & \textcolor{green}{\ding{51}} & \textcolor{green}{\ding{51}} & \cellcolor{green!15}\textcolor{red}{\ding{55}} \\
\texttt{nct\_select} & \cellcolor{green!15}\textcolor{red}{\ding{55}} & \textcolor{green}{\ding{51}} & \cellcolor{green!15}\textcolor{red}{\ding{55}} & \cellcolor{green!15}\textcolor{red}{\ding{55}} & \cellcolor{green!15}\textcolor{red}{\ding{55}} & \textcolor{green}{\ding{51}} & \cellcolor{green!15}\textcolor{red}{\ding{55}} & \cellcolor{green!15}\textcolor{red}{\ding{55}} & \cellcolor{green!15}\textcolor{red}{\ding{55}} & \textcolor{green}{\ding{51}} & \cellcolor{green!15}\textcolor{red}{\ding{55}} & \cellcolor{green!15}\textcolor{red}{\ding{55}} & \cellcolor{green!15}\textcolor{red}{\ding{55}} & \textcolor{green}{\ding{51}} & \cellcolor{green!15}\textcolor{red}{\ding{55}} & \cellcolor{green!15}\textcolor{red}{\ding{55}} \\
\texttt{nct\_sort} & \textcolor{red}{\ding{55}} & \textcolor{green}{\ding{51}} & \textcolor{red}{\ding{55}} & \cellcolor{green!15}\textcolor{red}{\ding{55}} & \textcolor{red}{\ding{55}} & \textcolor{green}{\ding{51}} & \textcolor{red}{\ding{55}} & \cellcolor{green!15}\textcolor{red}{\ding{55}} & \textcolor{red}{\ding{55}} & \textcolor{green}{\ding{51}} & \textcolor{red}{\ding{55}} & \cellcolor{green!15}\textcolor{red}{\ding{55}} & \textcolor{red}{\ding{55}} & \textcolor{green}{\ding{51}} & \textcolor{red}{\ding{55}} & \cellcolor{green!15}\textcolor{red}{\ding{55}} \\

\bottomrule
\end{tabular}
\end{table}

\paragraph*{Performance Overhead of CT-RetDec}
By selectively disabling optimization passes for \textsc{CT-RetDec}{}, 
we can expect increases in both the size of the decompiled code and 
the cost of subsequent analysis.
On average, using \textsc{CT-RetDec}{} incurs a 73\% increase in analysis time 
and a 316\% increase in decompiled code size. \Cref{fig:overhead} 
summarizes these overheads across different binary types and compiler 
optimization levels.
Binaries compiled with \texttt{-O0} have the largest overhead, 
as they contain more unoptimized code and therefore offer greater 
opportunities for decompilation optimizations that are disabled in \textsc{CT-RetDec}{}.
In contrast, binaries compiled with \texttt{-Os} are already compact and 
leave less room for further optimization, resulting in lower performance 
overhead. Given the benefit of using a more transparent \textsc{RetDec}{}, 
this performance overhead is generally acceptable.

\begin{figure}
    \centering
    \includegraphics[width=0.9\linewidth]{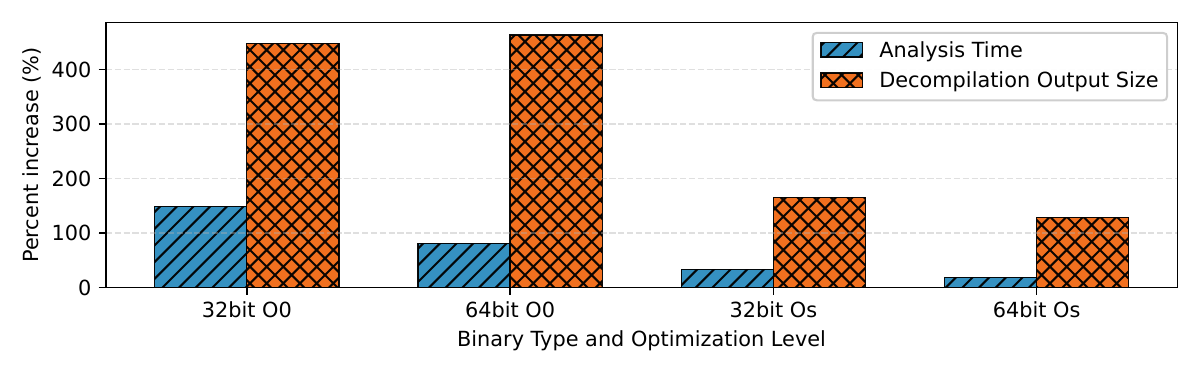}
    \caption{Performance Overhead of CT-RetDec compared to the combination of unmodified RetDec+CT-LLVM.}
    \label{fig:overhead}
\end{figure}

\section{Nontransparent Transformations in CT Analysis Tools}\label{sec:cttools}
So far, we have focused on the Decompile-then-Analyze approach,
ensuring that all program transformations applied by the decompiler
previous to calling the CT analysis tool are transparent.
However, there is a caveat:
Many CT analysis tools start their own analysis by also transforming programs
into an alternative representation previous to the actual analysis.
\Cref{tab:tool_comparison} provides a partial
list of CT tools, indicating for each of them their input language,
the language on which the analysis is carried out,
and which tool is used for the conversion.
CacheS~\cite{0011BL0ZW19} implicitly uses the Decompile-then-Analyze 
approach as it integrates the BINNAVI decompiler to decompile binaries to 
REIL IR before applying the CT analysis. Similarly, 
\textsc{BinSec}{}~\cite{DanielBR20} and 
\textsc{BinSec/Haunted}{}~\cite{DBLP:conf/ndss/DanielBR21} internally lift 
binary programs to DBA IR, which is an IR close to but richer than 
assembly language, before performing CT analysis.
To ensure the soundness of the listed CT analysis
tools, they must guarantee that their employed converters are transparent.
In this section, we examine the converters used by \textsc{CT-Verif}{} and
\textsc{BinSec}{}. 

\begin{table}
  \caption{Program transformations in CT analysis tools.}\label{tab:tool_comparison}
  \centering\small \begin{tabular}{llll}
  \toprule
    \textbf{\makecell{CT Analysis Tool}}
    &\textbf{\makecell{Input \\ Language}}
    & \textbf{\makecell{Output \\ Language}}
    & \textbf{\makecell{Converter}} \\
  \midrule
  FlowTracker~\cite{RodriguesPA16} & C & LLVM IR & LLVM \\
  CANAL~\cite{sung2018canal} & C & LLVM IR & LLVM \\
  \citet{disselkoen2020finding} & C & LLVM IR & LLVM \\
  \citet{barthe2014system} & C & MachIR & CompCert \\
  \citet{DBLP:conf/esorics/BlazyPT17} & C & C\#minor & CompCert \\
  ctverif~\cite{DBLP:conf/uss/AlmeidaBBDE16} & LLVM IR & Boogie IR & SMACK \\
  SideTrail~\cite{athanasiou2018sidetrail} & LLVM IR & Boogie IR & SMACK \\
  \citet{cai2024towards} & LLVM IR & Boogie IR & SMACK \\
  CacheS~\cite{0011BL0ZW19} & Binary & REIL & BINNAVI \\
  \textsc{BinSec}{}/Rel~\cite{DanielBR20} & Binary & DBA & \textsc{BinSec}{} \\
  \textsc{BinSec}{}/Haunted~\cite{DBLP:conf/ndss/DanielBR21} & Binary & DBA & \textsc{BinSec}{} \\
  \bottomrule
  \end{tabular}
\end{table}

\subsection{\textsc{CT-Verif}{}}
\textsc{CT-Verif}{}~\cite{DBLP:conf/uss/AlmeidaBBDE16} is a state-of-the-art CT
analysis tool for LLVM-IR programs. \textsc{CT-Verif}{} is based on a product
program construction, which is sound (non-CT programs are never deemed CT) and relatively complete (under certain conditions, CT programs never deemed non-CT).

\paragraph*{Conversion} \textsc{CT-Verif}{}~\cite{DBLP:conf/uss/AlmeidaBBDE16} internally uses
SMACK~\cite{DBLP:conf/cav/RakamaricE14} to transform an LLVM IR
program $P$ into a Boogie program $\ensuremath{{\llparenthesis\,{P}\,\rrparenthesis}}$ that emulates two
lock-step executions of $P$. The emulation additionally contains
assertions that check for differing leakage in the two executions to detect CT violations.

The SMACK-based transformation from LLVM IR to Boogie performs
standard optimizations at IR level, including Dead Load/Store
Elimination. Unfortunately, Dead Load/Store Elimination is not
transparent (\Cref{lst:dead_load,lst:self_store}).
This makes the results of making it possible to craft
an example of a non-CT \texttt{C} program that is accepted by \textsc{CT-Verif}{}, a soundness bug.

\paragraph*{Experiment}
To confirm our hypothesis, we construct such a program and apply \textsc{CT-Verif}{} 
to analyze it.
The program is shown on the left side of~\Cref{lst:ct-verif}.
It violates CT because the secret input \jazz{secret\_bit} is used as a memory address (Line~3).
However, the CT violation is removed by Dead Code Elimination in SMACK and no longer
exists in the code shown on the right side of~\Cref{lst:ct-verif} 
on which the safety analysis is performed.
Therefore, the transformation does not reflect CT\@, and \textsc{CT-Verif}{}
erroneously reports that the program is constant-time.

\begin{figure}
  \centering
  \begin{jasmincode}[outerpos=t,outerwidth=48ex]
    \jasminindent{0}\jasmincomment{// C Code}\\
    \jasminindent{0}\jasmintype{int} public\_array[\jasminconstant{2}] = \jasminopenbrace{}\jasminconstant{100}, \jasminconstant{101}\jasminclosebrace{};\\
    \jasminindent{0}\jasmintype{int} \jasmindname{func}(\jasmintype{int} secret\_bit) \jasminopenbrace{}\\
    \jasminindent{1}\jasmintype{int} dead\_load = array[secret\_bit];\\
    \jasminindent{1}\jasminkw{return} \jasminconstant{1};\\
    \jasminindent{0}\jasminclosebrace{}
  \end{jasmincode}\hspace{2em}\begin{jasmincode}[outerpos=t,outerwidth=30ex]
    \jasminindent{0}\jasmincomment{; LLVM IR}\\
    \jasminindent{0}\\
    \jasminindent{0}\jasminkw{define} i32 @func(i32) \jasminopenbrace{}\\
    \jasminindent{0}\\
    \jasminindent{1}ret i32 \jasminconstant{1};\\
    \jasminindent{0}\jasminclosebrace{}
  \end{jasmincode}
  \caption{Non-CT program accepted by \textsc{CT-Verif}{}.}\label{lst:ct-verif}
\end{figure}

\subsection{\textsc{BinSec}{}}\label{sec:unsoundness:binsec}
\textsc{BinSec}{}~\cite{DanielBR20} is a state-of-the-art binary-level CT analysis tool based on bounded verification with relational symbolic execution,
which is also sound and relatively complete.

\paragraph*{Conversion} \textsc{BinSec}{} lifts a binary program to an IR called
DBA (for Dynamic Bitvector Automata) before it performs relational
symbolic execution on the DBA program. When lifting a binary program
to DBA IR, \textsc{BinSec}{} converts conditional branches into two
\jazz{goto} statements---one points to the target address and the
other to the fall-through address. However, the lifting also performs
Empty Branch Coalescing (\Cref{lst:empty_branch_coalescing}),
which we have shown to be nontransparent,
a soundness bug in \textsc{BinSec}{}.

\paragraph*{Experiment} We confirm our hypothesis that \textsc{BinSec}{} can
accept non-CT programs by constructing the program shown on the left
side of~\Cref{lst:binsec_empty_branch}. The program compares a secret
input argument \jazz{rdi} to one.  If the two are equal, it
branches to address \jazz{T1}.  Otherwise, it executes the
fall-through code, which is also at \jazz{T1}. This binary
violates the CT policy because it branches on a secret.

\textsc{BinSec}{} converts the program to the DBA IR shown on the right side
of~\Cref{lst:binsec_empty_branch}.  Critically, the DBA program contains an
unconditional \jazz{goto} statement but no longer branches on a secret,
which means that the CT violation is gone.
\textsc{BinSec}{} erroneously reports that the input program is constant-time, a soundness bug.

\begin{figure}
  \centering
  \begin{jasmincode}[outerpos=t,outerwidth=20ex]
    \jasminindent{0}\jasmincomment{; ASM}\\
    \jasminindent{0}\phantom{T1: }cmp rdi, \jasminconstant{1}\\
    \jasminindent{0}\phantom{T1: }je T1\\
    \jasminindent{0}T1: mov rax, \jasminconstant{1}
  \end{jasmincode}\hspace{2em}\begin{jasmincode}[outerpos=t,outerwidth=20ex]
    \jasminindent{0}\jasmincomment{; DBA IR}\\
    \jasminindent{0}\phantom{T1: }cmp rdi, \jasminconstant{1}\\
    \jasminindent{0}\phantom{T1: }goto T1\\
    \jasminindent{0}T1: mov rax, \jasminconstant{1}
  \end{jasmincode}
  \caption{Non-CT program accepted by \textsc{BinSec}{}.}\label{lst:binsec_empty_branch}
\end{figure}

\paragraph*{Discussion} Our finding is unlikely to invalidate the results
of \textsc{BinSec}{} on real-world libraries, as conditionals with two empty
branches are unlikely to exist in the real-world. However, the gap
between theory and practice caused by nontransparent transformations
could be exploited by a malicious agent that could intentionally
introduce such examples into a popular library.

A separate question is whether the synthetic pattern from
\Cref{lst:binsec_empty_branch} could be exploited in practice.  The
answer is that it can, using the BunnyHop
technique~\cite{ZhangTOCGY23}. Specifically, an adversary can monitor
the branch predictor state, which is updated after each branch
execution, to infer whether a secret-dependent branch was taken.  
We confirm the feasibility of such an attack by modifying the artifact 
from~\cite{ZhangTOCGY23} to use our synthetic example as the victim.
The details of the attack and the experimental results are provided 
in \ifAppendix{\cref{appendix:binsec-attack}}\fi{}.

\section{Transparency for Speculative Constant-Time}\label{sec:speculation}
Constant-time programs remain vulnerable to Spectre
attacks~\cite{DBLP:conf/sp/KocherHFGGHHLM019}, where an attacker
manipulates speculative execution mechanisms (e.g., triggering branch
mispredictions) to recover secrets from speculative leakage. A common
way to protect programs against Spectre attacks is to follow the
speculative constant-time (SCT)
policy~\cite{DBLP:conf/pldi/CauligiDGTSRB20,highassuranceSpectre},
which extends the constant-time policy to the speculative
setting---the literature presents a few variants of this policy,
e.g.,~\cite{DBLP:conf/sp/GuarnieriKMRS20}.
There exists various
tools~\cite{DBLP:conf/sp/GuarnieriKMRS20,DBLP:conf/pldi/CauligiDGTSRB20,DBLP:journals/pacmpl/VassenaDGCKJTS21,DBLP:conf/ndss/DanielBR21,CauligiDMBS22,DBLP:conf/sp/ShivakumarBGLOPST23} to check SCT---or its variants.  This raises the question
of whether we can lift our results to verify a transformation is SCT
transparent.

This section discusses how to extend our results to SCT transparency.
We first show how to adapt our definitions of preservation,
reflection, and transparency to SCT\@.
We then lift simulation diagrams and their soundness results to SCT
(leaving full details to \ifAppendix{\cref{appendix:sct-bisim}}\fi{}).
Finally, we summarize our results analyzing several
standard transformations for SCT transparency.

\subsection{Speculative Constant-Time Transparency}\label{sec:sct:transparency}

Speculative semantics such as \cite{highassuranceSpectre,DBLP:conf/sp/ShivakumarBGLOPST23}
instrument the semantics of languages with \emph{directives}
\ensuremath{{\mathcal{D}}}{}, to capture the effect of branch prediction,
and the fact that the attacker can poison said predictor.  We follow
this direction and adapt our language model from \Cref{sec:setting} to
a speculative semantics with branch prediction.  Each step in the
semantics is guided by a directive $d \in \ensuremath{{\mathcal{D}}}$;
\sem{s}{d}{o}{s'} means that \(s\) steps under
directive~\(d\) to \(s'\) and produces the observation \(o\).
We consider two directives: \ensuremath{{\mathsf{step}}}{} and \ensuremath{{\mathsf{force}}}{}, which model
correct and incorrect predictions, respectively.  We use the \ensuremath{{\mathsf{step}}}{}
directive also for steps that require no prediction, such as
assignments.  We present the semantics in \cref{fig:sct-sem}.

\begin{figure}
  \small
  \begin{mathpar}
    \inferrule[Step]{
      s \step{}{o} s'
    }{
      s \step{\ensuremath{{\mathsf{step}}}}{o} s'
    }

    \inferrule[If]{
      b = \ensuremath{{\llbracket\, e \,\rrbracket_{\rho}}}
    }{
   \ensuremath{{\langle \ensuremath{{\mathcolor{mauve}{\mathtt{if}}~{e}~\mathcolor{mauve}{\mathtt{then}}~\ensuremath{{\mathit{c}}}_{\mathsf{t\kern-0.9ptt}}~\mathcolor{mauve}{\mathtt{else}}~\ensuremath{{\mathit{c}}}_{\mathsf{f\kern-1.1ptf}}}}, \rho, \mu \rangle}}
      \step{\ensuremath{{\mathsf{force}}}}{\ensuremath{{\mathsf{br}\; b}}}
      \ensuremath{{\langle \ensuremath{{\mathit{c}}}_{\neg b}, \rho, \mu \rangle}}
    }

    \inferrule[While]{
      \ensuremath{{\llbracket\, e \,\rrbracket_{\rho}}} = b\\
      c_\mathsf{f\kern-1.1ptf} = \ensuremath{{\mathcolor{mauve}{\mathtt{skip}}}}\\\\
      c_\mathsf{t\kern-0.9ptt} = c ; \ensuremath{{\mathcolor{mauve}{\mathtt{while}}~{e}~\mathcolor{mauve}{\mathtt{do}}~{c}}}
    }{
      \ensuremath{{\langle \ensuremath{{\mathcolor{mauve}{\mathtt{while}}~{e}~\mathcolor{mauve}{\mathtt{do}}~{c}}}, \rho, \mu \rangle}}
      \step{\ensuremath{{\mathsf{force}}}}{\ensuremath{{\mathsf{br}\; b}}}
      \ensuremath{{\langle c_{\neg b}, \rho, \mu \rangle}}
    }
  \end{mathpar}
  \caption{Speculative Semantics.}\label{fig:sct-sem}
\end{figure}

\Cref{def:beh,def:ct,def:ct:transparency} readily generalize to the
speculative setting.
Specifically, the speculative behavior of a program collects all
executions possible under any sequence of directives (this formulation
is inspired by \cite{WallM25}).

\begin{definition}\label{def:sbeh}
  The speculative behavior of \(P\) on an input~\(i\) is
  $
    \ensuremath{{\mathit{SBeh}({P}, {i})}} \triangleq
    \{
      \ensuremath{{(\ensuremath{{\boldsymbol{d}}}, \ensuremath{{\boldsymbol{o}}})}}
    \mid
      \sem*{P(i)}{\ensuremath{{\boldsymbol{d}}}}{\ensuremath{{\boldsymbol{o}}}}{s}
    \}
  $.
\end{definition}

A program is speculative constant-time w.r.t.\ an indistinguishability
relation \phi{}, denoted \ensuremath{{\phi\text{-SCT}}}{}, when for all sequences of directives,
related inputs yield equal observations.

\begin{definition}[\ensuremath{{\phi\text{-SCT}}}]\label{def:sct}
\(P\) is \ensuremath{{\phi\text{-SCT}}}{} if for all inputs \(i_1\)~and~\(i_2\),
  \ensuremath{{i_1 \mathrel{\mathphi} i_2 \implies \ensuremath{{\mathit{SBeh}({P}, {i_1})}} = \ensuremath{{\mathit{SBeh}({P}, {i_2})}}}}.
\end{definition}

Reflection, preservation, and transparency are defined
analogously to \Cref{def:ct:transparency}.

\begin{definition}[Reflection, Preservation, and Transparency for SCT]\label{def:sct:transparency}
  We say that a program transformation \ensuremath{{\ensuremath{{\llparenthesis\,{\cdot}\,\rrparenthesis}} : \ensuremath{{\mathcal{L}}}_s \to
    \ensuremath{{\mathcal{L}}}_t}} between an input language $\ensuremath{{\mathcal{L}}}_s$ and an output language $\ensuremath{{\mathcal{L}}}_t$
  \begin{itemize}
      \item Reflects \ensuremath{{\text{SCT}}}{} if for each \(P\) and $\phi$,
    \ensuremath{{\llparenthesis\,{P}\,\rrparenthesis}} is \ensuremath{{\phi\text{-SCT}}}{} implies \(P\) is \ensuremath{{\phi\text{-SCT}}}{};
\item Preserves \ensuremath{{\text{SCT}}}{} if for each $P$ and $\phi$,
    \(P\) is \ensuremath{{\phi\text{-SCT}}}{} implies \ensuremath{{\llparenthesis\,{P}\,\rrparenthesis}} is \ensuremath{{\phi\text{-SCT}}}{}; and
\item Is \ensuremath{{\text{SCT}}}{} transparent if it both reflects and preserves \ensuremath{{\text{SCT}}}{}.
  \end{itemize}
\end{definition}

\subsection{Proof Techniques}
We can extend the CT simulations from \cref{def:ct:diagram} to the
speculative setting, as well as \cref{thm:ct:soundness},
by accounting for the directives in the speculative semantics.
The key idea is to obtain a bisimulation
by linking the directives of the input program and the output program.
For that, we employ \emph{directive transformers} (introduced in~\cite{jasminPOPL25}),
which function analogously to observation transformers from \Cref{sec:proof-techniques}.
They are functions \ensuremath{{\ensuremath{{T_{\ensuremath{{\mathcal{D}}}}}} : \ensuremath{{\mathcal{PC}}} \to (\ensuremath{{\mathcal{D}}}_t \to \ensuremath{{\mathcal{D}}}_s)}}
that translate the directives for the output program into directives
for the input program and vice versa (\ensuremath{{T_{\ensuremath{{\mathcal{D}}}}^{\ensuremath{{\mathsf{pc}}}}}} is bijective for all $\ensuremath{{\mathsf{pc}}} \in \ensuremath{{\mathcal{PC}}}$).
Given two states in simulation \ensuremath{{s \sim t}} and a step \sem{t}{d}{o}{t'},
the directive transformer computes a directive $\ensuremath{{\ensuremath{{T_{\ensuremath{{\mathcal{D}}}}}}({\ensuremath{{\ensuremath{{\mathsf{pc}}}({s})}}}, {d})}}$
to determine the simulating step from \(s\).
\ifAppendix{\cref{appendix:sct:techniques}}\fi{} includes the detailed definitions of
these simulations.

\subsection{Summary of Results}

\begin{table}
  \caption{Summary of decompilation passes analyzed for SCT-transparency.
    A~\textcolor{green}{\ding{51}}{} means that we have proven the pass SCT-transparent using an SCT-simulation,
    a~\textcolor{red}{\Lightning}{} that we found a counterexample to \ensuremath{{\text{SCT}}}{} reflection,
    and~\(\bullet\){} means the pass is \ensuremath{{\text{SCT}}}{} preserving, but we omit
    the proof.}\label{tab:passes-sct}
  \small \begin{tabular}{l c c}
    \toprule
    Pass & \multicolumn{1}{l}{SCT-Reflection} & \multicolumn{1}{l}{SCT-Preservation}\\
    \midrule
Constant Folding & \textcolor{green}{\ding{51}}{} & \textcolor{green}{\ding{51}}{}\\
    Structural Analysis & \textcolor{green}{\ding{51}}{} & \textcolor{green}{\ding{51}}{}\\
Dead Branch Elimination & \textcolor{red}{\Lightning}{} & \(\bullet\){}\\
    Dead Assignment Elimination & \textcolor{red}{\Lightning}{} & \(\bullet\){}\\
    Unspilling & \textcolor{red}{\Lightning}{} & \(\bullet\){}\\
\bottomrule
  \end{tabular}
\end{table}

We studied the SCT transparency of five standard transformations 
and summarize the results in \cref{tab:passes-sct}.
\ifAppendix{\cref{appendix:sct:passes}}\fi{} applies our simulations to Structural
Analysis (as described in \cref{sec:passes}) and Constant Folding (a
special case of Expression Substitution from \cref{sec:passes}) to prove
them SCT transparent.
On the other hand, \ifAppendix{\cref{appendix:sct:counterexamples}}\fi{} shows that the
other transformations are not SCT transparent, as we summarize next.
Dead Branch Elimination is not SCT transparent 
because the body of a dead branch could be executed speculatively and still leak a secret.
Similarly, in Dead Assignment Elimination, an assignment could be
considered as dead if it affects only instructions within a dead branch.
However, during speculative execution, that branch could be executed
speculatively, leading to leakage of a secret propagated by the dead
assignment.
Unspilling is not SCT transparent because data stored on the 
stack could be speculatively accessed via out-of-bounds access and leaked 
later. Unspilling would move such data from the stack to registers, 
which makes it safe from speculative out-of-bounds accesses.

\section{Related Work}\label{sec:related}
Our work draws on the secure compilation literature to address the
question of whether decompiler transformations undermine the soundness
or accuracy of side-channel analysis.
Accordingly, we divide this section into three parts: first, we overview
side-channel analysis and decompiler transformations;
then, we review related work in the field of secure compilation;
finally, we briefly review related work on decompilation and discuss its
relevance.

\paragraph*{Side-Channel Analysis}
\citet{BarbosaBBBCLP21},
\citet{GeimerVRDBM23}, and
\citet{DBLP:conf/sp/JancarFBSSBFA22} survey a large number
of tools to detect side-channel vulnerabilities.
Several of these tools build on robust theoretical foundations, such as
product programs~\cite{DBLP:conf/uss/AlmeidaBBDE16},
type systems~\cite{10.1145/3133956.3134078},
SMT solvers~\cite{DBLP:conf/uss/BondHKLLPRST17},
abstract interpretation~\cite{DBLP:conf/esorics/BlazyPT17}, and
symbolic execution~\cite{DanielBR20,disselkoen2020finding}.
All of these, and the majority of the ones in the survey, operate on
source code or an IR, which means that applying them to assembly
programs requires decompiler transformations.
\Cref{sec:cttools} overviews different tools and discusses
\textsc{CT-Verif}{}~\cite{DBLP:conf/uss/AlmeidaBBDE16} and
\textsc{BinSec}{}~\cite{DanielBR20} in depth.

\paragraph*{Secure Compilation}
It is well known that secure source programs can be transformed into
insecure binary programs, leading to a dangerous
\emph{compiler security gap}~\cite{DBLP:conf/sp/DSilvaPS15,DBLP:conf/eurosp/SimonCA18,DBLP:conf/uss/XuLDDLWPM23}.
The field of secure compilation aims to close this gap by integrating
security considerations into compiler
development~\cite{DBLP:conf/csfw/AbateB0HPT19,DBLP:journals/toplas/AbateBCDGHPTT21}.
Consequently, we draw on this body of work to recast existing
techniques for preservation of side-channel security into a rigorous
approach for CT transparency.

Preservation of CT was first considered
in~\cite{DBLP:conf/csfw/BartheGL18}. The same work introduces
CT simulations as the primary tool for proving preservation. Later
work~\cite{DBLP:conf/ccs/BartheGLP21} uses the equivalent but simpler
technique of leakage transformers. Both the Jasmin compiler and (a
mild modification of) the CompCert compiler are shown to preserve
constant-time in~\cite{DBLP:journals/pacmpl/BartheBGHLPT20}
and~\cite{DBLP:conf/ccs/BartheGLP21}, respectively. Our notion of
simulation uses a variant of CT simulation that adds the
\ensuremath{{\mathcal{PC}}}{}-injectivity condition on leakage transformers.

\paragraph*{Decompilation in Program Analysis}
A number of prior works study the correctness of
decompilers~\cite{BrumleyLSW13,schulte2018evolving,BurkPKV22,DBLP:conf/pldi/VerbeekBFR22,DBLP:conf/uss/ZouKWGBT24}.
However, they focus on the input/output behavior of programs and
yield no guarantees about the security impact of decompilers.
Previous work explores and, to a large extent, confirms the
potential of source analysis of decompiled programs~\cite{DBLP:conf/sp/LiuYWB22,DBLP:conf/asplos/ZhouYHCZ24,DBLP:conf/asiaccs/MantovaniCSB22}.
Nevertheless, they neither formalize nor provide proof techniques for
secure decompilation.

\citet{DBLP:conf/sp/LiuYWB22} explore the impact of lifters
on pointer analysis, and discriminability analysis in the context of
LLVM lifters. Their approach is empirical, comparing
two LLVM IR representations of a source program. The first is obtained
by compiling from source to LLVM IR, and the second
by compiling the program to a binary and subsequently lifting
the binary to LLVM IR\@.

\citet{DBLP:conf/asiaccs/MantovaniCSB22} perform a
complementary exploration. They consider buffer overflows, integer
overflows, null pointer dereference, double free or use after free, and
division by zero. In addition, they provide several
recommendations for decompiler writers. One of their recommendations
resonates strongly with the findings of our work: departing from the
human-centric view of decompilers in favor of improving the soundness
and completeness of decompilers.

\section{Future Work}\label{sec:future}
In this section, we discuss several possible directions for future work.

\paragraph*{Mitigation Strategies}
In this paper, we mitigate the transparency issue in a real-world 
decompiler by finding and disabling nontransparent transformations.
Besides this approach, it is possible to add marks to force the 
decompiler to not optimize certain parts of the code.
This would be similar to how cryptographic 
developers mark secrets with the \emph{volatile} keyword to hint to
the compiler not to optimize operations that process them~\cite{DBLP:conf/eurosp/SimonCA18, intelmitigation}.
The benefit of this approach is that decompiler optimizations can still 
be run on unmarked code, which reduces the size and 
improves the readability of decompiled code.
The downside is that the users of the decompiler may need to 
manually identify cases that require the marks, which could be time-consuming
and error-prone. Furthermore, ensuring that decompilers support such
marks robustly and reliably is far from trivial.

Another possible mitigation strategy is to extend the
IR/high-level-language with a \emph{leak} primitive,
a no-op instruction whose only purpose is to leak a specific value.
It allows the decompiler to optimize the code as usual while inserting  
the \emph{leak} instruction to indicate that a value could be leaked 
at that point.
Again, implementing such a mitigation strategy in decompilers is nontrivial.
We leave the investigation of these mitigation strategies to future work.

\paragraph*{Generalization}
Since we primarily focus on the \textsc{RetDec}{} decompiler in this paper, 
one natural future research direction is to generalize our findings 
to other popular decompilers.
Specifically, it would be interesting to locate their 
nontransparent transformations and develop patches or extensions to 
control or disable them. The investigation could also be 
extended to the semantic gap between binaries and decompiled code,
and differences in undefined behavior, in order to identify their impact on transparency.
Besides finding and mitigating nontransparent 
transformations in decompilers, it would also be valuable to build a 
more comprehensive benchmark set to empirically evaluate 
the transparency of decompilers and the effectiveness of the mitigations.

\paragraph*{Transparency for More Transformations}
Apart from decompilation, an interesting future direction would be 
to study the transparency of program transformations for hardware-software
contracts~\cite{DBLP:conf/sp/GuarnieriKRV21}, a general framework that
encompasses different leakage and execution models.  More broadly, it
would be interesting to study transparency of program transformations
for other classes of security-related properties, e.g., memory
properties. Finally, it would be interesting to develop techniques to
prevent decompilers to eliminate potentially harmful code---similar to
techniques for preventing compilers to introduce harmful code.

\section{Conclusion}\label{sec:conclusion}
In this paper, we initiate the study of CT transparency.  We prove
that state-of-the-art decompilers remove CT violations and correct
this issue with the tool \textsc{CT-RetDec}{}, which transparently decompiles
our benchmark set of 160 binaries.  We
provide rigorous proof methods to assert the transparency
of transformations, and demonstrate them on common transformations.
Our work also emphasizes the need for developers of constant-time
verification tools to ensure transparency \emph{within} their
tools---specifically in the transformations that convert the input
program into the representation used for analysis. Our recommendations
for CT tool developers are: first, extensively test for transparency
of the initial transformations; second, carefully document the initial
transformations; third, if possible, expose these transformations for
scrutiny, ideally by clearly separating the transformation phases and
the analysis phases in the source code of the tool.

\begin{acks}
We thank the reviewers for their time and insightful feedback.
This research was supported
by the\grantsponsor{
  DFG}{
  \textit{Deutsche Forschungsgemeinschaft} (DFG, German research Foundation)}{}
as part of the Excellence Strategy of the German Federal and State Governments
-- \grantnum{DFG}{EXC 2092 CASA - 390781972}.
\end{acks}

\bibliographystyle{ACM-Reference-Format}

\begin{thebibliography}{74}



\ifx \showCODEN    \undefined \def \showCODEN     #1{\unskip}     \fi
\ifx \showISBNx    \undefined \def \showISBNx     #1{\unskip}     \fi
\ifx \showISBNxiii \undefined \def \showISBNxiii  #1{\unskip}     \fi
\ifx \showISSN     \undefined \def \showISSN      #1{\unskip}     \fi
\ifx \showLCCN     \undefined \def \showLCCN      #1{\unskip}     \fi
\ifx \shownote     \undefined \def \shownote      #1{#1}          \fi
\ifx \showarticletitle \undefined \def \showarticletitle #1{#1}   \fi
\ifx \showURL      \undefined \def \showURL       {\relax}        \fi
\providecommand\bibfield[2]{#2}
\providecommand\bibinfo[2]{#2}
\providecommand\natexlab[1]{#1}
\providecommand\showeprint[2][]{arXiv:#2}

\bibitem[35(2016)]{binaryninja}
\bibfield{author}{\bibinfo{person}{Vector 35}.}
  \bibinfo{year}{2016}\natexlab{}.
\newblock \bibinfo{title}{Binary Ninja}.
\newblock
\urldef\tempurl \url{https://binary.ninja/}
\showURL{\tempurl}


\bibitem[35(2025)]{dogbolt}
\bibfield{author}{\bibinfo{person}{Vector 35}.}
  \bibinfo{year}{2025}\natexlab{}.
\newblock \bibinfo{booktitle}{\emph{Dogbolt}}.
\newblock
\urldef\tempurl \url{https://dogbolt.org/}
\showURL{\tempurl}


\bibitem[Abate et~al\mbox{.}(2021)]{DBLP:journals/toplas/AbateBCDGHPTT21}
\bibfield{author}{\bibinfo{person}{Carmine Abate}, \bibinfo{person}{Roberto
  Blanco}, \bibinfo{person}{{\c{S}}tefan Ciob{\^{a}}c{\u{a}}},
  \bibinfo{person}{Adrien Durier}, \bibinfo{person}{Deepak Garg},
  \bibinfo{person}{Catalin Hri\cb{t}cu}, \bibinfo{person}{Marco Patrignani},
  \bibinfo{person}{{\'{E}}ric Tanter}, {and}
  \bibinfo{person}{J{\'{e}}r{\'{e}}my Thibault}.}
  \bibinfo{year}{2021}\natexlab{}.
\newblock \showarticletitle{An Extended Account of Trace-relating Compiler
  Correctness and Secure Compilation}.
\newblock \bibinfo{journal}{\emph{{ACM} Trans. Program. Lang. Syst.}}
  \bibinfo{volume}{43}, \bibinfo{number}{4} (\bibinfo{year}{2021}),
  \bibinfo{pages}{14:1--14:48}.
\newblock
\href{https://doi.org/10.1145/3460860}{doi:\nolinkurl{10.1145/3460860}}


\bibitem[Abate et~al\mbox{.}(2019)]{DBLP:conf/csfw/AbateB0HPT19}
\bibfield{author}{\bibinfo{person}{Carmine Abate}, \bibinfo{person}{Roberto
  Blanco}, \bibinfo{person}{Deepak Garg}, \bibinfo{person}{Catalin
  Hri\cb{t}cu}, \bibinfo{person}{Marco Patrignani}, {and}
  \bibinfo{person}{J{\'{e}}r{\'{e}}my Thibault}.}
  \bibinfo{year}{2019}\natexlab{}.
\newblock \showarticletitle{Journey Beyond Full Abstraction: Exploring Robust
  Property Preservation for Secure Compilation}. In
  \bibinfo{booktitle}{\emph{32nd {IEEE} Computer Security Foundations
  Symposium, {CSF} 2019, Hoboken, NJ, USA, June 25-28, 2019}}.
  \bibinfo{publisher}{{IEEE}}, \bibinfo{pages}{256--271}.
\newblock
\href{https://doi.org/10.1109/CSF.2019.00025}{doi:\nolinkurl{10.1109/CSF.2019.00025}}


\bibitem[Almeida et~al\mbox{.}(2017)]{10.1145/3133956.3134078}
\bibfield{author}{\bibinfo{person}{Jos\'{e}~Bacelar Almeida},
  \bibinfo{person}{Manuel Barbosa}, \bibinfo{person}{Gilles Barthe},
  \bibinfo{person}{Arthur Blot}, \bibinfo{person}{Benjamin Gr\'{e}goire},
  \bibinfo{person}{Vincent Laporte}, \bibinfo{person}{Tiago Oliveira},
  \bibinfo{person}{Hugo Pacheco}, \bibinfo{person}{Benedikt Schmidt}, {and}
  \bibinfo{person}{Pierre-Yves Strub}.} \bibinfo{year}{2017}\natexlab{}.
\newblock \showarticletitle{Jasmin: High-Assurance and High-Speed
  Cryptography}. In \bibinfo{booktitle}{\emph{Proceedings of the 2017 ACM
  SIGSAC Conference on Computer and Communications Security}} (Dallas, Texas,
  USA) \emph{(\bibinfo{series}{CCS '17})}. \bibinfo{publisher}{Association for
  Computing Machinery}, \bibinfo{address}{New York, NY, USA},
  \bibinfo{pages}{1807–1823}.
\newblock
\showISBNx{9781450349468}
\href{https://doi.org/10.1145/3133956.3134078}{doi:\nolinkurl{10.1145/3133956.3134078}}


\bibitem[Almeida et~al\mbox{.}(2016)]{DBLP:conf/uss/AlmeidaBBDE16}
\bibfield{author}{\bibinfo{person}{Jos{\'{e}}~Bacelar Almeida},
  \bibinfo{person}{Manuel Barbosa}, \bibinfo{person}{Gilles Barthe},
  \bibinfo{person}{Fran{\c{c}}ois Dupressoir}, {and} \bibinfo{person}{Michael
  Emmi}.} \bibinfo{year}{2016}\natexlab{}.
\newblock \showarticletitle{Verifying Constant-Time Implementations}. In
  \bibinfo{booktitle}{\emph{{USENIX} Security}}. \bibinfo{pages}{53--70}.
\newblock
\urldef\tempurl \url{https://www.usenix.org/conference/usenixsecurity16/technical-sessions/presentation/almeida}
\showURL{\tempurl}


\bibitem[Andriesse et~al\mbox{.}(2016)]{DBLP:conf/uss/AndriesseCVSB16}
\bibfield{author}{\bibinfo{person}{Dennis Andriesse}, \bibinfo{person}{Xi
  Chen}, \bibinfo{person}{Victor van~der Veen}, \bibinfo{person}{Asia
  Slowinska}, {and} \bibinfo{person}{Herbert Bos}.}
  \bibinfo{year}{2016}\natexlab{}.
\newblock \showarticletitle{An In-Depth Analysis of Disassembly on Full-Scale
  x86/x64 Binaries}. In \bibinfo{booktitle}{\emph{25th {USENIX} Security
  Symposium, {USENIX} Security 16, Austin, TX, USA, August 10-12, 2016}},
  \bibfield{editor}{\bibinfo{person}{Thorsten Holz} {and}
  \bibinfo{person}{Stefan Savage}} (Eds.). \bibinfo{publisher}{{USENIX}
  Association}, \bibinfo{pages}{583--600}.
\newblock
\urldef\tempurl \url{https://www.usenix.org/conference/usenixsecurity16/technical-sessions/presentation/andriesse}
\showURL{\tempurl}


\bibitem[Arranz-Olmos et~al\mbox{.}(2026)]{secdec-artifact}
\bibfield{author}{\bibinfo{person}{Santiago Arranz-Olmos},
  \bibinfo{person}{Gilles Barthe}, \bibinfo{person}{Lionel Blatter},
  \bibinfo{person}{Youcef Bouzid}, \bibinfo{person}{Sören van~der Wall}, {and}
  \bibinfo{person}{Zhiyuan Zhang}.} \bibinfo{year}{2026}\natexlab{}.
\newblock \bibinfo{booktitle}{\emph{Artifact for Paper Decompiling for
  Constant-Time Analysis}}.
\newblock
\href{https://doi.org/10.5281/zenodo.18749373}{doi:\nolinkurl{10.5281/zenodo.18749373}}


\bibitem[Arranz-Olmos et~al\mbox{.}(2025)]{jasminPOPL25}
\bibfield{author}{\bibinfo{person}{Santiago Arranz-Olmos},
  \bibinfo{person}{Gilles Barthe}, \bibinfo{person}{Lionel Blatter},
  \bibinfo{person}{Benjamin Gr\'{e}goire}, {and} \bibinfo{person}{Vincent
  Laporte}.} \bibinfo{year}{2025}\natexlab{}.
\newblock \showarticletitle{Preservation of Speculative Constant-Time by
  Compilation}.
\newblock \bibinfo{journal}{\emph{Proc. ACM Program. Lang.}}
  \bibinfo{volume}{9}, \bibinfo{number}{POPL}, Article \bibinfo{articleno}{44}
  (\bibinfo{date}{Jan.} \bibinfo{year}{2025}), \bibinfo{numpages}{33}~pages.
\newblock
\href{https://doi.org/10.1145/3704880}{doi:\nolinkurl{10.1145/3704880}}


\bibitem[Athanasiou et~al\mbox{.}(2018)]{athanasiou2018sidetrail}
\bibfield{author}{\bibinfo{person}{Konstantinos Athanasiou},
  \bibinfo{person}{Byron Cook}, \bibinfo{person}{Michael Emmi},
  \bibinfo{person}{Colm MacC{\'a}rthaigh}, \bibinfo{person}{Daniel
  Schwartz-Narbonne}, {and} \bibinfo{person}{Serdar Tasiran}.}
  \bibinfo{year}{2018}\natexlab{}.
\newblock \showarticletitle{Sidetrail: Verifying time-balancing of
  cryptosystems}. In \bibinfo{booktitle}{\emph{VSTTE}}.
  \bibinfo{pages}{215--228}.
\newblock


\bibitem[Barbosa et~al\mbox{.}(2021)]{BarbosaBBBCLP21}
\bibfield{author}{\bibinfo{person}{Manuel Barbosa}, \bibinfo{person}{Gilles
  Barthe}, \bibinfo{person}{Karthik Bhargavan}, \bibinfo{person}{Bruno
  Blanchet}, \bibinfo{person}{Cas Cremers}, \bibinfo{person}{Kevin Liao}, {and}
  \bibinfo{person}{Bryan Parno}.} \bibinfo{year}{2021}\natexlab{}.
\newblock \showarticletitle{SoK: Computer-Aided Cryptography}. In
  \bibinfo{booktitle}{\emph{{SP}}}. \bibinfo{pages}{777--795}.
\newblock
\href{https://doi.org/10.1109/SP40001.2021.00008}{doi:\nolinkurl{10.1109/SP40001.2021.00008}}


\bibitem[Barthe et~al\mbox{.}(2014)]{barthe2014system}
\bibfield{author}{\bibinfo{person}{Gilles Barthe}, \bibinfo{person}{Gustavo
  Betarte}, \bibinfo{person}{Juan Campo}, \bibinfo{person}{Carlos Luna}, {and}
  \bibinfo{person}{David Pichardie}.} \bibinfo{year}{2014}\natexlab{}.
\newblock \showarticletitle{System-level non-interference for constant-time
  cryptography}. In \bibinfo{booktitle}{\emph{CCS}}.
  \bibinfo{pages}{1267--1279}.
\newblock


\bibitem[Barthe et~al\mbox{.}(2020)]{DBLP:journals/pacmpl/BartheBGHLPT20}
\bibfield{author}{\bibinfo{person}{Gilles Barthe}, \bibinfo{person}{Sandrine
  Blazy}, \bibinfo{person}{Benjamin Gr{\'{e}}goire},
  \bibinfo{person}{R{\'{e}}mi Hutin}, \bibinfo{person}{Vincent Laporte},
  \bibinfo{person}{David Pichardie}, {and} \bibinfo{person}{Alix Trieu}.}
  \bibinfo{year}{2020}\natexlab{}.
\newblock \showarticletitle{Formal verification of a constant-time preserving
  {C} compiler}.
\newblock \bibinfo{journal}{\emph{Proc. {ACM} Program. Lang.}}
  \bibinfo{volume}{4}, \bibinfo{number}{{POPL}} (\bibinfo{year}{2020}),
  \bibinfo{pages}{7:1--7:30}.
\newblock
\href{https://doi.org/10.1145/3371075}{doi:\nolinkurl{10.1145/3371075}}


\bibitem[Barthe et~al\mbox{.}(2021a)]{highassuranceSpectre}
\bibfield{author}{\bibinfo{person}{Gilles Barthe}, \bibinfo{person}{Sunjay
  Cauligi}, \bibinfo{person}{Benjamin Gr{\'{e}}goire}, \bibinfo{person}{Adrien
  Koutsos}, \bibinfo{person}{Kevin Liao}, \bibinfo{person}{Tiago Oliveira},
  \bibinfo{person}{Swarn Priya}, \bibinfo{person}{Tamara Rezk}, {and}
  \bibinfo{person}{Peter Schwabe}.} \bibinfo{year}{2021}\natexlab{a}.
\newblock \showarticletitle{High-Assurance Cryptography in the Spectre Era}. In
  \bibinfo{booktitle}{\emph{42nd {IEEE} Symposium on Security and Privacy, {SP}
  2021, San Francisco, CA, USA, 24-27 May 2021}}. \bibinfo{publisher}{{IEEE}},
  \bibinfo{pages}{1884--1901}.
\newblock
\href{https://doi.org/10.1109/SP40001.2021.00046}{doi:\nolinkurl{10.1109/SP40001.2021.00046}}


\bibitem[Barthe et~al\mbox{.}(2018)]{DBLP:conf/csfw/BartheGL18}
\bibfield{author}{\bibinfo{person}{Gilles Barthe}, \bibinfo{person}{Benjamin
  Gr{\'{e}}goire}, {and} \bibinfo{person}{Vincent Laporte}.}
  \bibinfo{year}{2018}\natexlab{}.
\newblock \showarticletitle{Secure Compilation of Side-Channel Countermeasures:
  The Case of Cryptographic "Constant-Time"}. In \bibinfo{booktitle}{\emph{31st
  {IEEE} Computer Security Foundations Symposium, {CSF} 2018, Oxford, United
  Kingdom, July 9-12, 2018}}. \bibinfo{publisher}{{IEEE} Computer Society},
  \bibinfo{pages}{328--343}.
\newblock
\href{https://doi.org/10.1109/CSF.2018.00031}{doi:\nolinkurl{10.1109/CSF.2018.00031}}


\bibitem[Barthe et~al\mbox{.}(2021b)]{DBLP:conf/ccs/BartheGLP21}
\bibfield{author}{\bibinfo{person}{Gilles Barthe}, \bibinfo{person}{Benjamin
  Gr{\'{e}}goire}, \bibinfo{person}{Vincent Laporte}, {and}
  \bibinfo{person}{Swarn Priya}.} \bibinfo{year}{2021}\natexlab{b}.
\newblock \showarticletitle{Structured Leakage and Applications to
  Cryptographic Constant-Time and Cost}. In \bibinfo{booktitle}{\emph{{CCS}
  '21: 2021 {ACM} {SIGSAC} Conference on Computer and Communications Security,
  Virtual Event, Republic of Korea, November 15 - 19, 2021}},
  \bibfield{editor}{\bibinfo{person}{Yongdae Kim}, \bibinfo{person}{Jong Kim},
  \bibinfo{person}{Giovanni Vigna}, {and} \bibinfo{person}{Elaine Shi}} (Eds.).
  \bibinfo{publisher}{{ACM}}, \bibinfo{pages}{462--476}.
\newblock
\href{https://doi.org/10.1145/3460120.3484761}{doi:\nolinkurl{10.1145/3460120.3484761}}


\bibitem[Basque et~al\mbox{.}(2024)]{BasqueBGOMBDS024}
\bibfield{author}{\bibinfo{person}{Zion~Leonahenahe Basque},
  \bibinfo{person}{Ati~Priya Bajaj}, \bibinfo{person}{Wil Gibbs},
  \bibinfo{person}{Jude O'Kain}, \bibinfo{person}{Derron Miao},
  \bibinfo{person}{Tiffany Bao}, \bibinfo{person}{Adam Doup{\'{e}}},
  \bibinfo{person}{Yan Shoshitaishvili}, {and} \bibinfo{person}{Ruoyu Wang}.}
  \bibinfo{year}{2024}\natexlab{}.
\newblock \showarticletitle{Ahoy SAILR! There is No Need to {DREAM} of {C:} {A}
  Compiler-Aware Structuring Algorithm for Binary Decompilation}. In
  \bibinfo{booktitle}{\emph{{USENIX} Security}}.
\newblock
\urldef\tempurl \url{https://www.usenix.org/conference/usenixsecurity24/presentation/basque}
\showURL{\tempurl}


\bibitem[Bernstein(2005)]{bernstein2005cache}
\bibfield{author}{\bibinfo{person}{Daniel~J Bernstein}.}
  \bibinfo{year}{2005}\natexlab{}.
\newblock \showarticletitle{Cache-timing attacks on AES}.
\newblock  (\bibinfo{year}{2005}).
\newblock


\bibitem[Blazy et~al\mbox{.}(2017)]{DBLP:conf/esorics/BlazyPT17}
\bibfield{author}{\bibinfo{person}{Sandrine Blazy}, \bibinfo{person}{David
  Pichardie}, {and} \bibinfo{person}{Alix Trieu}.}
  \bibinfo{year}{2017}\natexlab{}.
\newblock \showarticletitle{Verifying Constant-Time Implementations by Abstract
  Interpretation}. In \bibinfo{booktitle}{\emph{Computer Security - {ESORICS}
  2017 - 22nd European Symposium on Research in Computer Security, Oslo,
  Norway, September 11-15, 2017, Proceedings, Part {I}}}
  \emph{(\bibinfo{series}{Lecture Notes in Computer Science},
  Vol.~\bibinfo{volume}{10492})}, \bibfield{editor}{\bibinfo{person}{Simon~N.
  Foley}, \bibinfo{person}{Dieter Gollmann}, {and} \bibinfo{person}{Einar
  Snekkenes}} (Eds.). \bibinfo{publisher}{Springer}, \bibinfo{pages}{260--277}.
\newblock
\href{https://doi.org/10.1007/978-3-319-66402-6\_16}{doi:\nolinkurl{10.1007/978-3-319-66402-6\_16}}


\bibitem[Bond et~al\mbox{.}(2017)]{DBLP:conf/uss/BondHKLLPRST17}
\bibfield{author}{\bibinfo{person}{Barry Bond}, \bibinfo{person}{Chris
  Hawblitzel}, \bibinfo{person}{Manos Kapritsos}, \bibinfo{person}{K.~Rustan~M.
  Leino}, \bibinfo{person}{Jacob~R. Lorch}, \bibinfo{person}{Bryan Parno},
  \bibinfo{person}{Ashay Rane}, \bibinfo{person}{Srinath T.~V. Setty}, {and}
  \bibinfo{person}{Laure Thompson}.} \bibinfo{year}{2017}\natexlab{}.
\newblock \showarticletitle{Vale: Verifying High-Performance Cryptographic
  Assembly Code}. In \bibinfo{booktitle}{\emph{26th {USENIX} Security
  Symposium, {USENIX} Security 2017, Vancouver, BC, Canada, August 16-18,
  2017}}, \bibfield{editor}{\bibinfo{person}{Engin Kirda} {and}
  \bibinfo{person}{Thomas Ristenpart}} (Eds.). \bibinfo{publisher}{{USENIX}
  Association}, \bibinfo{pages}{917--934}.
\newblock
\urldef\tempurl \url{https://www.usenix.org/conference/usenixsecurity17/technical-sessions/presentation/bond}
\showURL{\tempurl}


\bibitem[Brumley et~al\mbox{.}(2013)]{BrumleyLSW13}
\bibfield{author}{\bibinfo{person}{David Brumley}, \bibinfo{person}{JongHyup
  Lee}, \bibinfo{person}{Edward~J. Schwartz}, {and} \bibinfo{person}{Maverick
  Woo}.} \bibinfo{year}{2013}\natexlab{}.
\newblock \showarticletitle{Native x86 Decompilation Using Semantics-Preserving
  Structural Analysis and Iterative Control-Flow Structuring}. In
  \bibinfo{booktitle}{\emph{{USENIX} Security}}. \bibinfo{pages}{353--368}.
\newblock
\urldef\tempurl \url{https://www.usenix.org/conference/usenixsecurity13/technical-sessions/presentation/schwartz}
\showURL{\tempurl}


\bibitem[Burk et~al\mbox{.}(2022)]{BurkPKV22}
\bibfield{author}{\bibinfo{person}{Kevin Burk}, \bibinfo{person}{Fabio Pagani},
  \bibinfo{person}{Christopher Kruegel}, {and} \bibinfo{person}{Giovanni
  Vigna}.} \bibinfo{year}{2022}\natexlab{}.
\newblock \showarticletitle{Decomperson: How Humans Decompile and What We Can
  Learn From It}. In \bibinfo{booktitle}{\emph{{USENIX} Security}}.
  \bibinfo{pages}{2765--2782}.
\newblock
\urldef\tempurl \url{https://www.usenix.org/conference/usenixsecurity22/presentation/burk}
\showURL{\tempurl}


\bibitem[Cai et~al\mbox{.}(2024)]{cai2024towards}
\bibfield{author}{\bibinfo{person}{Luwei Cai}, \bibinfo{person}{Fu Song}, {and}
  \bibinfo{person}{Taolue Chen}.} \bibinfo{year}{2024}\natexlab{}.
\newblock \showarticletitle{Towards Efficient Verification of Constant-Time
  Cryptographic Implementations}.
\newblock \bibinfo{journal}{\emph{Proceedings of the ACM on Software
  Engineering}} \bibinfo{volume}{1}, \bibinfo{number}{FSE}
  (\bibinfo{year}{2024}), \bibinfo{pages}{1019--1042}.
\newblock


\bibitem[Cao et~al\mbox{.}(2024)]{DBLP:conf/issta/CaoZL024}
\bibfield{author}{\bibinfo{person}{Ying Cao}, \bibinfo{person}{Runze Zhang},
  \bibinfo{person}{Ruigang Liang}, {and} \bibinfo{person}{Kai Chen}.}
  \bibinfo{year}{2024}\natexlab{}.
\newblock \showarticletitle{Evaluating the Effectiveness of Decompilers}. In
  \bibinfo{booktitle}{\emph{Proceedings of the 33rd {ACM} {SIGSOFT}
  International Symposium on Software Testing and Analysis, {ISSTA} 2024,
  Vienna, Austria, September 16-20, 2024}},
  \bibfield{editor}{\bibinfo{person}{Maria Christakis} {and}
  \bibinfo{person}{Michael Pradel}} (Eds.). \bibinfo{publisher}{{ACM}},
  \bibinfo{pages}{491--502}.
\newblock
\href{https://doi.org/10.1145/3650212.3652144}{doi:\nolinkurl{10.1145/3650212.3652144}}


\bibitem[Cauligi et~al\mbox{.}(2022)]{CauligiDMBS22}
\bibfield{author}{\bibinfo{person}{Sunjay Cauligi}, \bibinfo{person}{Craig
  Disselkoen}, \bibinfo{person}{Daniel Moghimi}, \bibinfo{person}{Gilles
  Barthe}, {and} \bibinfo{person}{Deian Stefan}.}
  \bibinfo{year}{2022}\natexlab{}.
\newblock \showarticletitle{SoK: Practical Foundations for Software Spectre
  Defenses}. In \bibinfo{booktitle}{\emph{{SP}}}. \bibinfo{pages}{666--680}.
\newblock
\href{https://doi.org/10.1109/SP46214.2022.9833707}{doi:\nolinkurl{10.1109/SP46214.2022.9833707}}


\bibitem[Cauligi et~al\mbox{.}(2020)]{DBLP:conf/pldi/CauligiDGTSRB20}
\bibfield{author}{\bibinfo{person}{Sunjay Cauligi}, \bibinfo{person}{Craig
  Disselkoen}, \bibinfo{person}{Klaus von Gleissenthall},
  \bibinfo{person}{Dean~M. Tullsen}, \bibinfo{person}{Deian Stefan},
  \bibinfo{person}{Tamara Rezk}, {and} \bibinfo{person}{Gilles Barthe}.}
  \bibinfo{year}{2020}\natexlab{}.
\newblock \showarticletitle{Constant-time foundations for the new spectre era}.
  In \bibinfo{booktitle}{\emph{Proceedings of the 41st {ACM} {SIGPLAN}
  International Conference on Programming Language Design and Implementation,
  {PLDI} 2020, London, UK, June 15-20, 2020}},
  \bibfield{editor}{\bibinfo{person}{Alastair~F. Donaldson} {and}
  \bibinfo{person}{Emina Torlak}} (Eds.). \bibinfo{publisher}{{ACM}},
  \bibinfo{pages}{913--926}.
\newblock
\href{https://doi.org/10.1145/3385412.3385970}{doi:\nolinkurl{10.1145/3385412.3385970}}


\bibitem[Daniel et~al\mbox{.}(2020)]{DanielBR20}
\bibfield{author}{\bibinfo{person}{Lesly{-}Ann Daniel},
  \bibinfo{person}{S{\'{e}}bastien Bardin}, {and} \bibinfo{person}{Tamara
  Rezk}.} \bibinfo{year}{2020}\natexlab{}.
\newblock \showarticletitle{Binsec/Rel: Efficient Relational Symbolic Execution
  for Constant-Time at Binary-Level}. In \bibinfo{booktitle}{\emph{{SP}}}.
  \bibinfo{pages}{1021--1038}.
\newblock
\href{https://doi.org/10.1109/SP40000.2020.00074}{doi:\nolinkurl{10.1109/SP40000.2020.00074}}


\bibitem[Daniel et~al\mbox{.}(2021)]{DBLP:conf/ndss/DanielBR21}
\bibfield{author}{\bibinfo{person}{Lesly{-}Ann Daniel},
  \bibinfo{person}{S{\'{e}}bastien Bardin}, {and} \bibinfo{person}{Tamara
  Rezk}.} \bibinfo{year}{2021}\natexlab{}.
\newblock \showarticletitle{Hunting the Haunter - Efficient Relational Symbolic
  Execution for Spectre with Haunted RelSE}. In
  \bibinfo{booktitle}{\emph{{NDSS}}}.
\newblock
\urldef\tempurl \url{https://www.ndss-symposium.org/ndss-paper/hunting-the-haunter-efficient-relational-symbolic-execution-for-spectre-with-haunted-relse/}
\showURL{\tempurl}


\bibitem[Disselkoen et~al\mbox{.}(2020)]{disselkoen2020finding}
\bibfield{author}{\bibinfo{person}{Craig Disselkoen}, \bibinfo{person}{Sunjay
  Cauligi}, \bibinfo{person}{Dean Tullsen}, {and} \bibinfo{person}{Deian
  Stefan}.} \bibinfo{year}{2020}\natexlab{}.
\newblock \showarticletitle{Finding and eliminating timing side-channels in
  crypto code with pitchfork}. In \bibinfo{booktitle}{\emph{TECHCON}}.
\newblock


\bibitem[Dramko et~al\mbox{.}(2024)]{DramkoLSVG24}
\bibfield{author}{\bibinfo{person}{Luke Dramko}, \bibinfo{person}{Jeremy
  Lacomis}, \bibinfo{person}{Edward~J. Schwartz}, \bibinfo{person}{Bogdan
  Vasilescu}, {and} \bibinfo{person}{Claire {Le Goues}}.}
  \bibinfo{year}{2024}\natexlab{}.
\newblock \showarticletitle{A Taxonomy of {C} Decompiler Fidelity Issues}. In
  \bibinfo{booktitle}{\emph{{USENIX} Security}}.
\newblock
\urldef\tempurl \url{https://www.usenix.org/conference/usenixsecurity24/presentation/dramko}
\showURL{\tempurl}


\bibitem[D'Silva et~al\mbox{.}(2015)]{DBLP:conf/sp/DSilvaPS15}
\bibfield{author}{\bibinfo{person}{Vijay D'Silva}, \bibinfo{person}{Mathias
  Payer}, {and} \bibinfo{person}{Dawn~Xiaodong Song}.}
  \bibinfo{year}{2015}\natexlab{}.
\newblock \showarticletitle{The Correctness-Security Gap in Compiler
  Optimization}. In \bibinfo{booktitle}{\emph{2015 {IEEE} Symposium on Security
  and Privacy Workshops, {SPW} 2015, San Jose, CA, USA, May 21-22, 2015}}.
  \bibinfo{publisher}{{IEEE} Computer Society}, \bibinfo{pages}{73--87}.
\newblock
\href{https://doi.org/10.1109/SPW.2015.33}{doi:\nolinkurl{10.1109/SPW.2015.33}}


\bibitem[Enders et~al\mbox{.}(2022)]{dewolf}
\bibfield{author}{\bibinfo{person}{Steffen Enders},
  \bibinfo{person}{Eva{-}Maria~C. Behner}, \bibinfo{person}{Niklas Bergmann},
  \bibinfo{person}{Mariia Rybalka}, \bibinfo{person}{Elmar Padilla},
  \bibinfo{person}{Er~Xue Hui}, \bibinfo{person}{Henry Low}, {and}
  \bibinfo{person}{Nicholas Sim}.} \bibinfo{year}{2022}\natexlab{}.
\newblock \showarticletitle{dewolf: Improving Decompilation by leveraging User
  Surveys}.
\newblock \bibinfo{journal}{\emph{CoRR}}  \bibinfo{volume}{abs/2205.06719}
  (\bibinfo{year}{2022}).
\newblock
\showeprint[arXiv]{2205.06719}
\href{https://doi.org/10.48550/ARXIV.2205.06719}{doi:\nolinkurl{10.48550/ARXIV.2205.06719}}


\bibitem[Engel et~al\mbox{.}(2011)]{DBLP:conf/scopes/EngelLAFB11}
\bibfield{author}{\bibinfo{person}{Felix Engel}, \bibinfo{person}{Rainer
  Leupers}, \bibinfo{person}{Gerd Ascheid}, \bibinfo{person}{Max Ferger}, {and}
  \bibinfo{person}{Marcel Beemster}.} \bibinfo{year}{2011}\natexlab{}.
\newblock \showarticletitle{Enhanced structural analysis for {C} code
  reconstruction from {IR} code}. In \bibinfo{booktitle}{\emph{14th
  International Workshop on Software and Compilers for Embedded Systems,
  {SCOPES} '11, St. Goar, Germany, June 27-28, 2011}},
  \bibfield{editor}{\bibinfo{person}{Henk Corporaal} {and}
  \bibinfo{person}{Sander Stuijk}} (Eds.). \bibinfo{publisher}{{ACM}},
  \bibinfo{pages}{21--27}.
\newblock
\href{https://doi.org/10.1145/1988932.1988936}{doi:\nolinkurl{10.1145/1988932.1988936}}


\bibitem[Geimer and Maurice(2025)]{geimer2025fun}
\bibfield{author}{\bibinfo{person}{Antoine Geimer} {and}
  \bibinfo{person}{Clementine Maurice}.} \bibinfo{year}{2025}\natexlab{}.
\newblock \showarticletitle{Fun with flags: How Compilers Break and Fix
  Constant-Time Code}.
\newblock \bibinfo{journal}{\emph{arXiv preprint arXiv:2507.06112}}
  (\bibinfo{year}{2025}).
\newblock


\bibitem[Geimer et~al\mbox{.}(2023)]{GeimerVRDBM23}
\bibfield{author}{\bibinfo{person}{Antoine Geimer},
  \bibinfo{person}{Math{\'{e}}o Vergnolle},
  \bibinfo{person}{Fr{\'{e}}d{\'{e}}ric Recoules}, \bibinfo{person}{Lesly{-}Ann
  Daniel}, \bibinfo{person}{S{\'{e}}bastien Bardin}, {and}
  \bibinfo{person}{Cl{\'{e}}mentine Maurice}.} \bibinfo{year}{2023}\natexlab{}.
\newblock \showarticletitle{A Systematic Evaluation of Automated Tools for
  Side-Channel Vulnerabilities Detection in Cryptographic Libraries}. In
  \bibinfo{booktitle}{\emph{{CCS}}}. \bibinfo{publisher}{{ACM}},
  \bibinfo{pages}{1690--1704}.
\newblock
\href{https://doi.org/10.1145/3576915.3623112}{doi:\nolinkurl{10.1145/3576915.3623112}}


\bibitem[Guarnieri et~al\mbox{.}(2020)]{DBLP:conf/sp/GuarnieriKMRS20}
\bibfield{author}{\bibinfo{person}{Marco Guarnieri}, \bibinfo{person}{Boris
  K{\"{o}}pf}, \bibinfo{person}{Jos{\'{e}}~F. Morales}, \bibinfo{person}{Jan
  Reineke}, {and} \bibinfo{person}{Andr{\'{e}}s S{\'{a}}nchez}.}
  \bibinfo{year}{2020}\natexlab{}.
\newblock \showarticletitle{Spectector: Principled Detection of Speculative
  Information Flows}. In \bibinfo{booktitle}{\emph{2020 {IEEE} Symposium on
  Security and Privacy, {SP} 2020, San Francisco, CA, USA, May 18-21, 2020}}.
  \bibinfo{publisher}{{IEEE}}, \bibinfo{pages}{1--19}.
\newblock
\href{https://doi.org/10.1109/SP40000.2020.00011}{doi:\nolinkurl{10.1109/SP40000.2020.00011}}


\bibitem[Guarnieri et~al\mbox{.}(2021)]{DBLP:conf/sp/GuarnieriKRV21}
\bibfield{author}{\bibinfo{person}{Marco Guarnieri}, \bibinfo{person}{Boris
  K{\"{o}}pf}, \bibinfo{person}{Jan Reineke}, {and} \bibinfo{person}{Pepe
  Vila}.} \bibinfo{year}{2021}\natexlab{}.
\newblock \showarticletitle{Hardware-Software Contracts for Secure
  Speculation}. In \bibinfo{booktitle}{\emph{42nd {IEEE} Symposium on Security
  and Privacy, {SP} 2021, San Francisco, CA, USA, 24-27 May 2021}}.
  \bibinfo{publisher}{{IEEE}}, \bibinfo{pages}{1868--1883}.
\newblock
\href{https://doi.org/10.1109/SP40001.2021.00036}{doi:\nolinkurl{10.1109/SP40001.2021.00036}}


\bibitem[Gussoni et~al\mbox{.}(2020)]{DBLP:conf/ccs/GussoniFFA20}
\bibfield{author}{\bibinfo{person}{Andrea Gussoni},
  \bibinfo{person}{Alessandro~Di Federico}, \bibinfo{person}{Pietro Fezzardi},
  {and} \bibinfo{person}{Giovanni Agosta}.} \bibinfo{year}{2020}\natexlab{}.
\newblock \showarticletitle{A Comb for Decompiled {C} Code}. In
  \bibinfo{booktitle}{\emph{{ASIA} {CCS} '20: The 15th {ACM} Asia Conference on
  Computer and Communications Security, Taipei, Taiwan, October 5-9, 2020}},
  \bibfield{editor}{\bibinfo{person}{Hung{-}Min Sun},
  \bibinfo{person}{Shiuh{-}Pyng Shieh}, \bibinfo{person}{Guofei Gu}, {and}
  \bibinfo{person}{Giuseppe Ateniese}} (Eds.). \bibinfo{publisher}{{ACM}},
  \bibinfo{pages}{637--651}.
\newblock
\href{https://doi.org/10.1145/3320269.3384766}{doi:\nolinkurl{10.1145/3320269.3384766}}


\bibitem[Hex-Rays(2023)]{hexrays}
\bibfield{author}{\bibinfo{person}{Hex-Rays}.} \bibinfo{year}{2023}\natexlab{}.
\newblock \bibinfo{title}{Ida Pro}.
\newblock
\urldef\tempurl \url{https://www. hex-rays.com/products/ida}
\showURL{\tempurl}


\bibitem[Intel(2022)]{intelmitigation}
\bibfield{author}{\bibinfo{person}{Intel}.} \bibinfo{year}{2022}\natexlab{}.
\newblock \bibinfo{booktitle}{\emph{Guidelines for Mitigating Timing Side
  Channels Against Cryptographic Implementations}}.
\newblock
\urldef\tempurl \url{https://www.intel.com/content/www/us/en/developer/articles/technical/software-security-guidance/secure-coding/mitigate-timing-side-channel-crypto-implementation.html}
\showURL{\tempurl}
\newblock
\shownote{Accessed: 2026-01-22}.


\bibitem[Intel(2023)]{DOIT}
\bibfield{author}{\bibinfo{person}{Intel}.} \bibinfo{year}{2023}\natexlab{}.
\newblock \bibinfo{booktitle}{\emph{Data Operand Independent Timing
  Instructions}}.
\newblock
\urldef\tempurl \url{https://www.intel.com/content/www/us/en/developer/articles/technical/software-security-guidance/resources/data-operand-independent-timing-instructions.html}
\showURL{\tempurl}
\newblock
\shownote{Accessed: 2025-10-06}.


\bibitem[Jancar et~al\mbox{.}(2022)]{DBLP:conf/sp/JancarFBSSBFA22}
\bibfield{author}{\bibinfo{person}{Jan Jancar}, \bibinfo{person}{Marcel
  Fourn{\'{e}}}, \bibinfo{person}{Daniel De~Almeida Braga},
  \bibinfo{person}{Mohamed Sabt}, \bibinfo{person}{Peter Schwabe},
  \bibinfo{person}{Gilles Barthe}, \bibinfo{person}{Pierre{-}Alain Fouque},
  {and} \bibinfo{person}{Yasemin Acar}.} \bibinfo{year}{2022}\natexlab{}.
\newblock \showarticletitle{"They're not that hard to mitigate": What
  Cryptographic Library Developers Think About Timing Attacks}. In
  \bibinfo{booktitle}{\emph{43rd {IEEE} Symposium on Security and Privacy, {SP}
  2022, San Francisco, CA, USA, May 22-26, 2022}}. \bibinfo{publisher}{{IEEE}},
  \bibinfo{pages}{632--649}.
\newblock
\href{https://doi.org/10.1109/SP46214.2022.9833713}{doi:\nolinkurl{10.1109/SP46214.2022.9833713}}


\bibitem[Kocher et~al\mbox{.}(2019)]{DBLP:conf/sp/KocherHFGGHHLM019}
\bibfield{author}{\bibinfo{person}{Paul Kocher}, \bibinfo{person}{Jann Horn},
  \bibinfo{person}{Anders Fogh}, \bibinfo{person}{Daniel Genkin},
  \bibinfo{person}{Daniel Gruss}, \bibinfo{person}{Werner Haas},
  \bibinfo{person}{Mike Hamburg}, \bibinfo{person}{Moritz Lipp},
  \bibinfo{person}{Stefan Mangard}, \bibinfo{person}{Thomas Prescher},
  \bibinfo{person}{Michael Schwarz}, {and} \bibinfo{person}{Yuval Yarom}.}
  \bibinfo{year}{2019}\natexlab{}.
\newblock \showarticletitle{Spectre Attacks: Exploiting Speculative Execution}.
  In \bibinfo{booktitle}{\emph{{IEEE} {SP}}}. \bibinfo{pages}{1--19}.
\newblock
\href{https://doi.org/10.1109/SP.2019.00002}{doi:\nolinkurl{10.1109/SP.2019.00002}}


\bibitem[Kocher(1996)]{Koc96}
\bibfield{author}{\bibinfo{person}{Paul~C. Kocher}.}
  \bibinfo{year}{1996}\natexlab{}.
\newblock \showarticletitle{Timing Attacks on Implementations of
  {D}iffie-{H}ellman, {RSA}, {DSS}, and Other Systems}. In
  \bibinfo{booktitle}{\emph{Advances in Cryptology -- {CRYPTO'96}}}
  \emph{(\bibinfo{series}{LNCS}, Vol.~\bibinfo{volume}{1109})}.
  \bibinfo{publisher}{SV}, \bibinfo{pages}{104--113}.
\newblock
\urldef\tempurl \url{http://www.cryptography.com/public/pdf/TimingAttacks.pdf}
\showURL{\tempurl}


\bibitem[K{\v{r}}oustek et~al\mbox{.}(2017)]{retdec}
\bibfield{author}{\bibinfo{person}{Jakub K{\v{r}}oustek},
  \bibinfo{person}{Peter Matula}, {and} \bibinfo{person}{Petr Zemek}.}
  \bibinfo{year}{2017}\natexlab{}.
\newblock \showarticletitle{Retdec: An open-source machine-code decompiler}.
\newblock
\urldef\tempurl \url{https://github.com/avast/retdec}
\showURL{\tempurl}


\bibitem[Liu and Wang(2020)]{DBLP:conf/issta/LiuW20}
\bibfield{author}{\bibinfo{person}{Zhibo Liu} {and} \bibinfo{person}{Shuai
  Wang}.} \bibinfo{year}{2020}\natexlab{}.
\newblock \showarticletitle{How far we have come: testing decompilation
  correctness of {C} decompilers}. In \bibinfo{booktitle}{\emph{{ISSTA} '20:
  29th {ACM} {SIGSOFT} International Symposium on Software Testing and
  Analysis, Virtual Event, USA, July 18-22, 2020}},
  \bibfield{editor}{\bibinfo{person}{Sarfraz Khurshid} {and}
  \bibinfo{person}{Corina~S. Pasareanu}} (Eds.). \bibinfo{publisher}{{ACM}},
  \bibinfo{pages}{475--487}.
\newblock
\href{https://doi.org/10.1145/3395363.3397370}{doi:\nolinkurl{10.1145/3395363.3397370}}


\bibitem[Liu et~al\mbox{.}(2022)]{DBLP:conf/sp/LiuYWB22}
\bibfield{author}{\bibinfo{person}{Zhibo Liu}, \bibinfo{person}{Yuanyuan Yuan},
  \bibinfo{person}{Shuai Wang}, {and} \bibinfo{person}{Yuyan Bao}.}
  \bibinfo{year}{2022}\natexlab{}.
\newblock \showarticletitle{SoK: Demystifying Binary Lifters Through the Lens
  of Downstream Applications}. In \bibinfo{booktitle}{\emph{43rd {IEEE}
  Symposium on Security and Privacy, {SP} 2022, San Francisco, CA, USA, May
  22-26, 2022}}. \bibinfo{publisher}{{IEEE}}, \bibinfo{pages}{1100--1119}.
\newblock
\href{https://doi.org/10.1109/SP46214.2022.9833799}{doi:\nolinkurl{10.1109/SP46214.2022.9833799}}


\bibitem[Mantovani et~al\mbox{.}(2022)]{DBLP:conf/asiaccs/MantovaniCSB22}
\bibfield{author}{\bibinfo{person}{Alessandro Mantovani}, \bibinfo{person}{Luca
  Compagna}, \bibinfo{person}{Yan Shoshitaishvili}, {and}
  \bibinfo{person}{Davide Balzarotti}.} \bibinfo{year}{2022}\natexlab{}.
\newblock \showarticletitle{The Convergence of Source Code and Binary
  Vulnerability Discovery - {A} Case Study}. In
  \bibinfo{booktitle}{\emph{{ASIA} {CCS} '22: {ACM} Asia Conference on Computer
  and Communications Security, Nagasaki, Japan, 30 May 2022 - 3 June 2022}},
  \bibfield{editor}{\bibinfo{person}{Yuji Suga}, \bibinfo{person}{Kouichi
  Sakurai}, \bibinfo{person}{Xuhua Ding}, {and} \bibinfo{person}{Kazue Sako}}
  (Eds.). \bibinfo{publisher}{{ACM}}, \bibinfo{pages}{602--615}.
\newblock
\href{https://doi.org/10.1145/3488932.3497764}{doi:\nolinkurl{10.1145/3488932.3497764}}


\bibitem[Mattei et~al\mbox{.}(2022)]{DBLP:conf/acsac/MatteiMKV22}
\bibfield{author}{\bibinfo{person}{James Mattei}, \bibinfo{person}{Madeline
  McLaughlin}, \bibinfo{person}{Samantha Katcher}, {and}
  \bibinfo{person}{Daniel Votipka}.} \bibinfo{year}{2022}\natexlab{}.
\newblock \showarticletitle{A Qualitative Evaluation of Reverse Engineering
  Tool Usability}. In \bibinfo{booktitle}{\emph{Annual Computer Security
  Applications Conference, {ACSAC} 2022, Austin, TX, USA, December 5-9, 2022}}.
  \bibinfo{publisher}{{ACM}}, \bibinfo{pages}{619--631}.
\newblock
\href{https://doi.org/10.1145/3564625.3567993}{doi:\nolinkurl{10.1145/3564625.3567993}}


\bibitem[Molnar et~al\mbox{.}(2005)]{DBLP:conf/icisc/MolnarPSW05}
\bibfield{author}{\bibinfo{person}{David Molnar}, \bibinfo{person}{Matt
  Piotrowski}, \bibinfo{person}{David Schultz}, {and} \bibinfo{person}{David~A.
  Wagner}.} \bibinfo{year}{2005}\natexlab{}.
\newblock \showarticletitle{The Program Counter Security Model: Automatic
  Detection and Removal of Control-Flow Side Channel Attacks}. In
  \bibinfo{booktitle}{\emph{Information Security and Cryptology - {ICISC} 2005,
  8th International Conference, Seoul, Korea, December 1-2, 2005, Revised
  Selected Papers}} \emph{(\bibinfo{series}{Lecture Notes in Computer Science},
  Vol.~\bibinfo{volume}{3935})}, \bibfield{editor}{\bibinfo{person}{Dongho Won}
  {and} \bibinfo{person}{Seungjoo Kim}} (Eds.). \bibinfo{publisher}{Springer},
  \bibinfo{pages}{156--168}.
\newblock
\href{https://doi.org/10.1007/11734727\_14}{doi:\nolinkurl{10.1007/11734727\_14}}


\bibitem[Nasrabadi et~al\mbox{.}(2025)]{abs-2511-11385}
\bibfield{author}{\bibinfo{person}{Faezeh Nasrabadi}, \bibinfo{person}{Robert
  K{\"{u}}nnemann}, {and} \bibinfo{person}{Hamed Nemati}.}
  \bibinfo{year}{2025}\natexlab{}.
\newblock \showarticletitle{Automated Side-Channel Analysis of Cryptographic
  Protocol Implementations}.
\newblock \bibinfo{journal}{\emph{CoRR}}  \bibinfo{volume}{abs/2511.11385}
  (\bibinfo{year}{2025}).
\newblock
\href{https://doi.org/10.48550/ARXIV.2511.11385}{doi:\nolinkurl{10.48550/ARXIV.2511.11385}}


\bibitem[(NSA)(2018)]{ghidra}
\bibfield{author}{\bibinfo{person}{National Security~Agency (NSA)}.}
  \bibinfo{year}{2018}\natexlab{}.
\newblock \bibinfo{title}{Ghidra}.
\newblock
\urldef\tempurl \url{https://www.nsa.gov/resources/everyone/ghidra/}
\showURL{\tempurl}


\bibitem[Purnal(2024)]{clangover}
\bibfield{author}{\bibinfo{person}{Antoon Purnal}.}
  \bibinfo{year}{2024}\natexlab{}.
\newblock \bibinfo{booktitle}{\emph{Clangover CVE}}.
\newblock
\urldef\tempurl \url{https://nvd.nist.gov/vuln/detail/CVE-2024-37880}
\showURL{\tempurl}
\newblock
\shownote{Accessed: 2025-10-03}.


\bibitem[Rakamaric and Emmi(2014)]{DBLP:conf/cav/RakamaricE14}
\bibfield{author}{\bibinfo{person}{Zvonimir Rakamaric} {and}
  \bibinfo{person}{Michael Emmi}.} \bibinfo{year}{2014}\natexlab{}.
\newblock \showarticletitle{{SMACK:} Decoupling Source Language Details from
  Verifier Implementations}. In \bibinfo{booktitle}{\emph{Computer Aided
  Verification - 26th International Conference, {CAV} 2014, Held as Part of the
  Vienna Summer of Logic, {VSL} 2014, Vienna, Austria, July 18-22, 2014.
  Proceedings}} \emph{(\bibinfo{series}{Lecture Notes in Computer Science},
  Vol.~\bibinfo{volume}{8559})}, \bibfield{editor}{\bibinfo{person}{Armin
  Biere} {and} \bibinfo{person}{Roderick Bloem}} (Eds.).
  \bibinfo{publisher}{Springer}, \bibinfo{pages}{106--113}.
\newblock
\href{https://doi.org/10.1007/978-3-319-08867-9\_7}{doi:\nolinkurl{10.1007/978-3-319-08867-9\_7}}


\bibitem[Rodrigues et~al\mbox{.}(2016)]{RodriguesPA16}
\bibfield{author}{\bibinfo{person}{Bruno Rodrigues}, \bibinfo{person}{Fernando
  Magno~Quint{\~{a}}o Pereira}, {and} \bibinfo{person}{Diego~F. Aranha}.}
  \bibinfo{year}{2016}\natexlab{}.
\newblock \showarticletitle{Sparse representation of implicit flows with
  applications to side-channel detection}. In \bibinfo{booktitle}{\emph{{CC}}}.
  \bibinfo{pages}{110--120}.
\newblock
\href{https://doi.org/10.1145/2892208.2892230}{doi:\nolinkurl{10.1145/2892208.2892230}}


\bibitem[Schneider et~al\mbox{.}(2024)]{schneider2024breaking}
\bibfield{author}{\bibinfo{person}{Moritz Schneider}, \bibinfo{person}{Daniele
  Lain}, \bibinfo{person}{Ivan Puddu}, \bibinfo{person}{Nicolas Dutly}, {and}
  \bibinfo{person}{Srdjan Capkun}.} \bibinfo{year}{2024}\natexlab{}.
\newblock \bibinfo{title}{Breaking Bad: How Compilers Break
  Constant-Time~Implementations}.
\newblock
\showeprint[arxiv]{2410.13489}~[cs.CR]
\urldef\tempurl \url{https://arxiv.org/abs/2410.13489}
\showURL{\tempurl}


\bibitem[Schulte et~al\mbox{.}(2018)]{schulte2018evolving}
\bibfield{author}{\bibinfo{person}{Eric Schulte}, \bibinfo{person}{Jason
  Ruchti}, \bibinfo{person}{Matt Noonan}, \bibinfo{person}{David Ciarletta},
  {and} \bibinfo{person}{Alexey Loginov}.} \bibinfo{year}{2018}\natexlab{}.
\newblock \showarticletitle{Evolving exact decompilation}. In
  \bibinfo{booktitle}{\emph{Workshop on Binary Analysis Research (BAR)}}.
\newblock


\bibitem[Schwartz and Lee(2013)]{DBLP:conf/uss/SchwartzL13}
\bibfield{author}{\bibinfo{person}{Edward~J. Schwartz} {and}
  \bibinfo{person}{JongHyup Lee}.} \bibinfo{year}{2013}\natexlab{}.
\newblock \showarticletitle{Decompilation of Binary Programs Using Control-Flow
  Structuring and Semantic-Preserving Transformations}. In
  \bibinfo{booktitle}{\emph{{USENIX} Security}}. \bibinfo{pages}{369--384}.
\newblock
\urldef\tempurl \url{https://www.usenix.org/conference/usenixsecurity13/technical-sessions/presentation/schwartz}
\showURL{\tempurl}


\bibitem[Shivakumar et~al\mbox{.}(2023)]{DBLP:conf/sp/ShivakumarBGLOPST23}
\bibfield{author}{\bibinfo{person}{Basavesh~Ammanaghatta Shivakumar},
  \bibinfo{person}{Gilles Barthe}, \bibinfo{person}{Benjamin Gr{\'{e}}goire},
  \bibinfo{person}{Vincent Laporte}, \bibinfo{person}{Tiago Oliveira},
  \bibinfo{person}{Swarn Priya}, \bibinfo{person}{Peter Schwabe}, {and}
  \bibinfo{person}{Lucas Tabary{-}Maujean}.} \bibinfo{year}{2023}\natexlab{}.
\newblock \showarticletitle{Typing High-Speed Cryptography against Spectre v1}.
  In \bibinfo{booktitle}{\emph{44th {IEEE} Symposium on Security and Privacy,
  {SP} 2023, San Francisco, CA, USA, May 21-25, 2023}}.
  \bibinfo{publisher}{{IEEE}}, \bibinfo{pages}{1094--1111}.
\newblock
\href{https://doi.org/10.1109/SP46215.2023.10179418}{doi:\nolinkurl{10.1109/SP46215.2023.10179418}}


\bibitem[Shoshitaishvili et~al\mbox{.}(2016)]{angr}
\bibfield{author}{\bibinfo{person}{Yan Shoshitaishvili}, \bibinfo{person}{Ruoyu
  Wang}, \bibinfo{person}{Christopher Salls}, \bibinfo{person}{Nick Stephens},
  \bibinfo{person}{Mario Polino}, \bibinfo{person}{Andrew Dutcher},
  \bibinfo{person}{John Grosen}, \bibinfo{person}{Siji Feng},
  \bibinfo{person}{Christophe Hauser}, \bibinfo{person}{Christopher
  Kr{\"{u}}gel}, {and} \bibinfo{person}{Giovanni Vigna}.}
  \bibinfo{year}{2016}\natexlab{}.
\newblock \showarticletitle{{SOK:} (State of) The Art of War: Offensive
  Techniques in Binary Analysis}. In \bibinfo{booktitle}{\emph{{SP}}}.
  \bibinfo{pages}{138--157}.
\newblock
\href{https://doi.org/10.1109/SP.2016.17}{doi:\nolinkurl{10.1109/SP.2016.17}}


\bibitem[Simon et~al\mbox{.}(2018)]{DBLP:conf/eurosp/SimonCA18}
\bibfield{author}{\bibinfo{person}{Laurent Simon}, \bibinfo{person}{David
  Chisnall}, {and} \bibinfo{person}{Ross~J. Anderson}.}
  \bibinfo{year}{2018}\natexlab{}.
\newblock \showarticletitle{What You Get is What You {C:} Controlling Side
  Effects in Mainstream {C} Compilers}. In \bibinfo{booktitle}{\emph{2018
  {IEEE} European Symposium on Security and Privacy, EuroS{\&}P 2018, London,
  United Kingdom, April 24-26, 2018}}. \bibinfo{publisher}{{IEEE}},
  \bibinfo{pages}{1--15}.
\newblock
\href{https://doi.org/10.1109/EuroSP.2018.00009}{doi:\nolinkurl{10.1109/EuroSP.2018.00009}}


\bibitem[Sung et~al\mbox{.}(2018)]{sung2018canal}
\bibfield{author}{\bibinfo{person}{Chungha Sung}, \bibinfo{person}{Brandon
  Paulsen}, {and} \bibinfo{person}{Chao Wang}.}
  \bibinfo{year}{2018}\natexlab{}.
\newblock \showarticletitle{CANAL: a cache timing analysis framework via LLVM
  transformation}. In \bibinfo{booktitle}{\emph{Proceedings of the 33rd
  ACM/IEEE International Conference on Automated Software Engineering}}
  (Montpellier, France) \emph{(\bibinfo{series}{ASE '18})}.
  \bibinfo{publisher}{Association for Computing Machinery},
  \bibinfo{address}{New York, NY, USA}, \bibinfo{pages}{904–907}.
\newblock
\showISBNx{9781450359375}
\href{https://doi.org/10.1145/3238147.3240485}{doi:\nolinkurl{10.1145/3238147.3240485}}


\bibitem[van~der Wall and Meyer(2025)]{WallM25}
\bibfield{author}{\bibinfo{person}{S\"{o}ren van~der Wall} {and}
  \bibinfo{person}{Roland Meyer}.} \bibinfo{year}{2025}\natexlab{}.
\newblock \showarticletitle{SNIP: Speculative Execution and Non-Interference
  Preservation for Compiler Transformations}.
\newblock \bibinfo{journal}{\emph{Proc. ACM Program. Lang.}}
  \bibinfo{volume}{9}, \bibinfo{number}{POPL}, Article \bibinfo{articleno}{51}
  (\bibinfo{date}{Jan.} \bibinfo{year}{2025}), \bibinfo{numpages}{30}~pages.
\newblock
\href{https://doi.org/10.1145/3704887}{doi:\nolinkurl{10.1145/3704887}}


\bibitem[Vassena et~al\mbox{.}(2021)]{DBLP:journals/pacmpl/VassenaDGCKJTS21}
\bibfield{author}{\bibinfo{person}{Marco Vassena}, \bibinfo{person}{Craig
  Disselkoen}, \bibinfo{person}{Klaus von Gleissenthall},
  \bibinfo{person}{Sunjay Cauligi}, \bibinfo{person}{Rami~G{\"{o}}khan Kici},
  \bibinfo{person}{Ranjit Jhala}, \bibinfo{person}{Dean~M. Tullsen}, {and}
  \bibinfo{person}{Deian Stefan}.} \bibinfo{year}{2021}\natexlab{}.
\newblock \showarticletitle{Automatically eliminating speculative leaks from
  cryptographic code with blade}.
\newblock \bibinfo{journal}{\emph{Proc. {ACM} Program. Lang.}}
  \bibinfo{volume}{5}, \bibinfo{number}{{POPL}} (\bibinfo{year}{2021}),
  \bibinfo{pages}{1--30}.
\newblock
\href{https://doi.org/10.1145/3434330}{doi:\nolinkurl{10.1145/3434330}}


\bibitem[Verbeek et~al\mbox{.}(2022)]{DBLP:conf/pldi/VerbeekBFR22}
\bibfield{author}{\bibinfo{person}{Freek Verbeek}, \bibinfo{person}{Joshua~A.
  Bockenek}, \bibinfo{person}{Zhoulai Fu}, {and} \bibinfo{person}{Binoy
  Ravindran}.} \bibinfo{year}{2022}\natexlab{}.
\newblock \showarticletitle{Formally verified lifting of C-compiled x86-64
  binaries}. In \bibinfo{booktitle}{\emph{{PLDI} '22: 43rd {ACM} {SIGPLAN}
  International Conference on Programming Language Design and Implementation,
  San Diego, CA, USA, June 13 - 17, 2022}},
  \bibfield{editor}{\bibinfo{person}{Ranjit Jhala} {and} \bibinfo{person}{Isil
  Dillig}} (Eds.). \bibinfo{publisher}{{ACM}}, \bibinfo{pages}{934--949}.
\newblock
\href{https://doi.org/10.1145/3519939.3523702}{doi:\nolinkurl{10.1145/3519939.3523702}}


\bibitem[Wang et~al\mbox{.}(2019)]{0011BL0ZW19}
\bibfield{author}{\bibinfo{person}{Shuai Wang}, \bibinfo{person}{Yuyan Bao},
  \bibinfo{person}{Xiao Liu}, \bibinfo{person}{Pei Wang},
  \bibinfo{person}{Danfeng Zhang}, {and} \bibinfo{person}{Dinghao Wu}.}
  \bibinfo{year}{2019}\natexlab{}.
\newblock \showarticletitle{Identifying Cache-Based Side Channels through
  Secret-Augmented Abstract Interpretation}. In
  \bibinfo{booktitle}{\emph{{USENIX} Security}}. \bibinfo{pages}{657--674}.
\newblock
\urldef\tempurl \url{https://www.usenix.org/conference/usenixsecurity19/presentation/wang-shuai}
\showURL{\tempurl}


\bibitem[Xu et~al\mbox{.}(2023)]{DBLP:conf/uss/XuLDDLWPM23}
\bibfield{author}{\bibinfo{person}{Jianhao Xu}, \bibinfo{person}{Kangjie Lu},
  \bibinfo{person}{Zhengjie Du}, \bibinfo{person}{Zhu Ding},
  \bibinfo{person}{Linke Li}, \bibinfo{person}{Qiushi Wu},
  \bibinfo{person}{Mathias Payer}, {and} \bibinfo{person}{Bing Mao}.}
  \bibinfo{year}{2023}\natexlab{}.
\newblock \showarticletitle{Silent Bugs Matter: {A} Study of
  Compiler-Introduced Security Bugs}. In \bibinfo{booktitle}{\emph{32nd
  {USENIX} Security Symposium, {USENIX} Security 2023, Anaheim, CA, USA, August
  9-11, 2023}}, \bibfield{editor}{\bibinfo{person}{Joseph~A. Calandrino} {and}
  \bibinfo{person}{Carmela Troncoso}} (Eds.). \bibinfo{publisher}{{USENIX}
  Association}, \bibinfo{pages}{3655--3672}.
\newblock
\urldef\tempurl \url{https://www.usenix.org/conference/usenixsecurity23/presentation/xu-jianhao}
\showURL{\tempurl}


\bibitem[Yakdan et~al\mbox{.}(2016)]{yakdan2016helping}
\bibfield{author}{\bibinfo{person}{Khaled Yakdan}, \bibinfo{person}{Sergej
  Dechand}, \bibinfo{person}{Elmar Gerhards-Padilla}, {and}
  \bibinfo{person}{Matthew Smith}.} \bibinfo{year}{2016}\natexlab{}.
\newblock \showarticletitle{Helping johnny to analyze malware: A
  usability-optimized decompiler and malware analysis user study}. In
  \bibinfo{booktitle}{\emph{2016 IEEE Symposium on Security and Privacy (SP)}}.
  IEEE, \bibinfo{pages}{158--177}.
\newblock


\bibitem[Yakdan et~al\mbox{.}(2015)]{YakdanEGS15}
\bibfield{author}{\bibinfo{person}{Khaled Yakdan}, \bibinfo{person}{Sebastian
  Eschweiler}, \bibinfo{person}{Elmar Gerhards{-}Padilla}, {and}
  \bibinfo{person}{Matthew Smith}.} \bibinfo{year}{2015}\natexlab{}.
\newblock \showarticletitle{No More Gotos: Decompilation Using
  Pattern-Independent Control-Flow Structuring and Semantic-Preserving
  Transformations}. In \bibinfo{booktitle}{\emph{{NDSS}}}.
\newblock
\urldef\tempurl \url{https://www.ndss-symposium.org/ndss2015/no-more-gotos-decompilation-using-pattern-independent-control-flow-structuring-and-semantics}
\showURL{\tempurl}


\bibitem[Zhang and Barthe(2025)]{cryptoeprint:2025/338}
\bibfield{author}{\bibinfo{person}{Zhiyuan Zhang} {and} \bibinfo{person}{Gilles
  Barthe}.} \bibinfo{year}{2025}\natexlab{}.
\newblock \bibinfo{title}{{CT}-{LLVM}: Automatic Large-Scale Constant-Time
  Analysis}.
\newblock \bibinfo{howpublished}{Cryptology {ePrint} Archive, Paper 2025/338}.
\newblock
\urldef\tempurl \url{https://eprint.iacr.org/2025/338}
\showURL{\tempurl}


\bibitem[Zhang et~al\mbox{.}(2023)]{ZhangTOCGY23}
\bibfield{author}{\bibinfo{person}{Zhiyuan Zhang}, \bibinfo{person}{Mingtian
  Tao}, \bibinfo{person}{Sioli O'Connell}, \bibinfo{person}{Chitchanok
  Chuengsatiansup}, \bibinfo{person}{Daniel Genkin}, {and}
  \bibinfo{person}{Yuval Yarom}.} \bibinfo{year}{2023}\natexlab{}.
\newblock \showarticletitle{BunnyHop: Exploiting the Instruction Prefetcher}.
  In \bibinfo{booktitle}{\emph{{USENIX} Security}}.
  \bibinfo{pages}{7321--7337}.
\newblock
\urldef\tempurl \url{https://www.usenix.org/conference/usenixsecurity23/presentation/zhang-zhiyuan-bunnyhop}
\showURL{\tempurl}


\bibitem[Zhou et~al\mbox{.}(2024b)]{DBLP:conf/asplos/ZhouYHCZ24}
\bibfield{author}{\bibinfo{person}{Anshunkang Zhou}, \bibinfo{person}{Chengfeng
  Ye}, \bibinfo{person}{Heqing Huang}, \bibinfo{person}{Yuandao Cai}, {and}
  \bibinfo{person}{Charles Zhang}.} \bibinfo{year}{2024}\natexlab{b}.
\newblock \showarticletitle{Plankton: Reconciling Binary Code and Debug
  Information}. In \bibinfo{booktitle}{\emph{Proceedings of the 29th {ACM}
  International Conference on Architectural Support for Programming Languages
  and Operating Systems, Volume 2, {ASPLOS} 2024, La Jolla, CA, USA, 27 April
  2024- 1 May 2024}}, \bibfield{editor}{\bibinfo{person}{Rajiv Gupta},
  \bibinfo{person}{Nael~B. Abu{-}Ghazaleh}, \bibinfo{person}{Madan Musuvathi},
  {and} \bibinfo{person}{Dan Tsafrir}} (Eds.). \bibinfo{publisher}{{ACM}},
  \bibinfo{pages}{912--928}.
\newblock
\href{https://doi.org/10.1145/3620665.3640382}{doi:\nolinkurl{10.1145/3620665.3640382}}


\bibitem[Zhou et~al\mbox{.}(2024a)]{ZhouDZ24}
\bibfield{author}{\bibinfo{person}{Quan Zhou}, \bibinfo{person}{Sixuan Dang},
  {and} \bibinfo{person}{Danfeng Zhang}.} \bibinfo{year}{2024}\natexlab{a}.
\newblock \showarticletitle{CtChecker: A Precise, Sound and Efficient Static
  Analysis for Constant-Time Programming}. In
  \bibinfo{booktitle}{\emph{{ECOOP}}} \emph{(\bibinfo{series}{LIPIcs},
  Vol.~\bibinfo{volume}{313})}. \bibinfo{pages}{46:1--46:26}.
\newblock
\href{https://doi.org/10.4230/LIPICS.ECOOP.2024.46}{doi:\nolinkurl{10.4230/LIPICS.ECOOP.2024.46}}


\bibitem[Zou et~al\mbox{.}(2024)]{DBLP:conf/uss/ZouKWGBT24}
\bibfield{author}{\bibinfo{person}{Muqi Zou}, \bibinfo{person}{Arslan Khan},
  \bibinfo{person}{Ruoyu Wu}, \bibinfo{person}{Han Gao},
  \bibinfo{person}{Antonio Bianchi}, {and} \bibinfo{person}{Dave~(Jing) Tian}.}
  \bibinfo{year}{2024}\natexlab{}.
\newblock \showarticletitle{D-Helix: {A} Generic Decompiler Testing Framework
  Using Symbolic Differentiation}. In \bibinfo{booktitle}{\emph{33rd {USENIX}
  Security Symposium, {USENIX} Security 2024, Philadelphia, PA, USA, August
  14-16, 2024}}, \bibfield{editor}{\bibinfo{person}{Davide Balzarotti} {and}
  \bibinfo{person}{Wenyuan Xu}} (Eds.). \bibinfo{publisher}{{USENIX}
  Association}.
\newblock
\urldef\tempurl \url{https://www.usenix.org/conference/usenixsecurity24/presentation/zou}
\showURL{\tempurl}


\end{thebibliography}

\appendix
\newcommand*{\ellexit}{\ell_{\mathtt{exit}}}

\section{Proof Details for the Transformations in \Cref{sec:prog-optimizations}}\label{appendix:proofs-prog-optimizations}

We provide more detailed definitions and proofs for a selection of the transformations in \Cref{sec:prog-optimizations}.
To avoid repetition we cut the details for the other transformations, and we chose our selection according to their uniqueness.

\subsection{Expression Substitution}
We define $\ensuremath{{\llparenthesis\,{\cdot}\,\rrparenthesis}}$ inductively via
\begin{gather*}
    \begin{align*}
        \ensuremath{{\llparenthesis\,{\ensuremath{{\mathcolor{mauve}{\mathtt{skip}}}}}\,\rrparenthesis}} & \;=\; \ensuremath{{\mathcolor{mauve}{\mathtt{skip}}}}
                     &
        \ensuremath{{\llparenthesis\,{\ensuremath{{{x}\mathcolor{cyan}{=}{\ensuremath{{\{{e_s}\mathbin{\text{\normalfont\guilsinglright}}{e_t}\}}}e}}}}\,\rrparenthesis}} & \;=\;  \ensuremath{{{x}\mathcolor{cyan}{=}{\ensuremath{{{e}[{e_s}\mathbin{\text{\normalfont\guilsinglright}}{e_t}]}}}}}&
        \\
        \ensuremath{{\llparenthesis\,{\ensuremath{{{x}\mathcolor{cyan}{=}{\left[\ensuremath{{\{{e_s}\mathbin{\text{\normalfont\guilsinglright}}{e_t}\}}}e\right]}}}}\,\rrparenthesis}} & \;=\;  \ensuremath{{{x}\mathcolor{cyan}{=}{\left[\ensuremath{{{e}[{e_s}\mathbin{\text{\normalfont\guilsinglright}}{e_t}]}}\right]}}}
                                                &
        \ensuremath{{\llparenthesis\,{\ensuremath{{{\left[\ensuremath{{\{{e_s}\mathbin{\text{\normalfont\guilsinglright}}{e_t}\}}}e\right]}\mathcolor{cyan}{=}{x}}}}\,\rrparenthesis}} & \;=\;  \ensuremath{{{\ensuremath{{{e}[{e_s}\mathbin{\text{\normalfont\guilsinglright}}{e_t}]}}}\mathcolor{cyan}{=}{\left[x\right]}}}&
    \end{align*}
    \\
    \begin{align*}
        \ensuremath{{\llparenthesis\,{\ensuremath{{\mathcolor{mauve}{\mathtt{if}}~{\ensuremath{{\{{e_s}\mathbin{\text{\normalfont\guilsinglright}}{e_t}\}}}e}~\mathcolor{mauve}{\mathtt{then}}~c_\mathsf{t\kern-0.9ptt}~\mathcolor{mauve}{\mathtt{else}}~c_\mathsf{f\kern-1.1ptf}}}}\,\rrparenthesis}} &\;=\;  \ensuremath{{\mathcolor{mauve}{\mathtt{if}}~{\ensuremath{{{e}[{e_s}\mathbin{\text{\normalfont\guilsinglright}}{e_t}]}}}~\mathcolor{mauve}{\mathtt{then}}~\ensuremath{{\llparenthesis\,{c_\mathsf{t\kern-0.9ptt}}\,\rrparenthesis}}~\mathcolor{mauve}{\mathtt{else}}~\ensuremath{{\llparenthesis\,{c_\mathsf{f\kern-1.1ptf}}\,\rrparenthesis}}}}&
        \\
        \ensuremath{{\llparenthesis\,{\ensuremath{{\mathcolor{mauve}{\mathtt{while}}~{\ensuremath{{\{{e_s}\mathbin{\text{\normalfont\guilsinglright}}{e_t}\}}}e}~\mathcolor{mauve}{\mathtt{do}}~{c}}}}\,\rrparenthesis}} &\;=\;  \ensuremath{{\mathcolor{mauve}{\mathtt{while}}~{\ensuremath{{{e}[{e_s}\mathbin{\text{\normalfont\guilsinglright}}{e_t}]}}}~\mathcolor{mauve}{\mathtt{do}}~{\ensuremath{{\llparenthesis\,{c_\mathsf{t\kern-0.9ptt}}\,\rrparenthesis}}}}}&
        \\
        \ensuremath{{\llparenthesis\,{c_1 \mathop{;} c_2}\,\rrparenthesis}} &\;=\;  \ensuremath{{\llparenthesis\,{c_1}\,\rrparenthesis}} \mathop{;} \ensuremath{{\llparenthesis\,{c_2}\,\rrparenthesis}} &
    \end{align*}
\end{gather*}
where $\ensuremath{{{e}[{e_s}\mathbin{\text{\normalfont\guilsinglright}}{e_t}]}}$ is standard substitution of every occurrence of $e_s$ in $e$ by $e_t$.

We defined $\sim$ to hold only at $\ensuremath{{\langle c, \ensuremath{{\rho}}, \ensuremath{{\mu}} \rangle}} \sim \ensuremath{{\langle \ensuremath{{\llparenthesis\,{c}\,\rrparenthesis}}, \ensuremath{{\rho}}, \ensuremath{{\mu}} \rangle}}$
and the leakage transformer as \ensuremath{{\Tobs{}{o_s} = o_s}}.

\begin{proof}[Proof of \Cref{thm:ct:lock-step-soundness:simulation}]
    Consider $s \sim t$, i.e., $s = \ensuremath{{\langle c, \ensuremath{{\rho}}, \ensuremath{{\mu}} \rangle}}$ and $t = \ensuremath{{\langle \ensuremath{{\llparenthesis\,{c}\,\rrparenthesis}}, \ensuremath{{\rho}}, \ensuremath{{\mu}} \rangle}}$,
    and a step $t\step{}{o_t}t' = \ensuremath{{\langle c', \ensuremath{{\rho}}', \ensuremath{{\mu}}' \rangle}}$.
    We prove that $s \step{}{o_t} s'$ with $s' \sim t'$ by induction on $c$ ($\Tobs{}{o} = o$).
    \begin{enumerate}[label=$\blacktriangleright$,leftmargin=*]
        \item \emph{\refassign}\hskip 1.2em Then, $c \;=\; \ensuremath{{{x}\mathcolor{cyan}{=}{\ensuremath{{\{{e_s}\mathbin{\text{\normalfont\guilsinglright}}{e_t}\}}}e}}}$ and $\ensuremath{{\llparenthesis\,{c}\,\rrparenthesis}} \;= \;\ensuremath{{{x}\mathcolor{cyan}{=}{\ensuremath{{{e}[{e_s}\mathbin{\text{\normalfont\guilsinglright}}{e_t}]}}}}}$.
        By definition, $\ensuremath{{\mu}} = \ensuremath{{\mu}}'$, and $\ensuremath{{\rho}}' = \ensuremath{{{\ensuremath{{\rho}}}[{x}\mathbin{\text{\normalfont\guilsinglright}}{\ensuremath{{\llbracket\, \ensuremath{{{e}[{e_s}\mathbin{\text{\normalfont\guilsinglright}}{e_t}]}} \,\rrbracket_{\ensuremath{{\rho}}}}}}]}}$ and $o_t = \ensuremath{{\bullet}}$.
        Due to correctness of the annotation, \Cref{prop:expr-subst-guarantee}, we have that $\ensuremath{{\llbracket\, e \,\rrbracket_{\ensuremath{{\rho}}}}} = \ensuremath{{\llbracket\, \ensuremath{{{e}[{e_s}\mathbin{\text{\normalfont\guilsinglright}}{e_t}]}} \,\rrbracket_{\ensuremath{{\rho}}}}}$.
        Thus, $s \step{}{\ensuremath{{\bullet}}} \ensuremath{{\langle c', \ensuremath{{\rho}}', \ensuremath{{\mu}}' \rangle}} = s'$, and $s' \sim t'$ as required.
        \item \emph{\refload}\hskip 1.2em Then, $c \;=\; \ensuremath{{{x}\mathcolor{cyan}{=}{\left[\ensuremath{{\{{e_s}\mathbin{\text{\normalfont\guilsinglright}}{e_t}\}}}e\right]}}}$ and $\ensuremath{{\llparenthesis\,{c}\,\rrparenthesis}} \;= \;\ensuremath{{{x}\mathcolor{cyan}{=}{\left[\ensuremath{{{e}[{e_s}\mathbin{\text{\normalfont\guilsinglright}}{e_t}]}}\right]}}}$.
        By definition, $\ensuremath{{\mu}} = \ensuremath{{\mu}}'$, $a = \ensuremath{{\llbracket\, \ensuremath{{{e}[{e_s}\mathbin{\text{\normalfont\guilsinglright}}{e_t}]}} \,\rrbracket_{\ensuremath{{\rho}}}}}$ and $\ensuremath{{\rho}}' = \ensuremath{{{\ensuremath{{\rho}}}[{x}\mathbin{\text{\normalfont\guilsinglright}}{\ensuremath{{\mu}}(a)}]}}$ and $o_t = \ensuremath{{\mathsf{adr}\; a}}$.
        Due to \Cref{prop:expr-subst-guarantee}, we have that $\ensuremath{{\llbracket\, e \,\rrbracket_{\ensuremath{{\rho}}}}} = \ensuremath{{\llbracket\, \ensuremath{{{e}[{e_s}\mathbin{\text{\normalfont\guilsinglright}}{e_t}]}} \,\rrbracket_{\ensuremath{{\rho}}}}} = a$.
        Thus, $s \step{}{\ensuremath{{\mathsf{adr}\; a}}} \ensuremath{{\langle c', \ensuremath{{\rho}}', \ensuremath{{\mu}}' \rangle}} = s'$, and $s' \sim t'$.
        \item \emph{\refstore}\hskip 1.2em Then, $c \;=\; \ensuremath{{{\left[\ensuremath{{\{{e_s}\mathbin{\text{\normalfont\guilsinglright}}{e_t}\}}}e\right]}\mathcolor{cyan}{=}{x}}}$ and $\ensuremath{{\llparenthesis\,{c}\,\rrparenthesis}} \;= \;\ensuremath{{{\left[\ensuremath{{{e}[{e_s}\mathbin{\text{\normalfont\guilsinglright}}{e_t}]}}\right]}\mathcolor{cyan}{=}{x}}}$.
        By definition, $\ensuremath{{\rho}} = \ensuremath{{\rho}}'$, $a = \ensuremath{{\llbracket\, \ensuremath{{{e}[{e_s}\mathbin{\text{\normalfont\guilsinglright}}{e_t}]}} \,\rrbracket_{\ensuremath{{\rho}}}}}$ and $\ensuremath{{\mu}}' = \ensuremath{{{\ensuremath{{\mu}}}[{a}\mathbin{\text{\normalfont\guilsinglright}}{\ensuremath{{\rho}}(x)}]}}$ and $o_t = \ensuremath{{\mathsf{adr}\; a}}$.
        Due to \Cref{prop:expr-subst-guarantee}, we have that $\ensuremath{{\llbracket\, e \,\rrbracket_{\ensuremath{{\rho}}}}} = \ensuremath{{\llbracket\, \ensuremath{{{e}[{e_s}\mathbin{\text{\normalfont\guilsinglright}}{e_t}]}} \,\rrbracket_{\ensuremath{{\rho}}}}} = a$.
        Thus, $s \step{}{\ensuremath{{\mathsf{adr}\; a}}} \ensuremath{{\langle c', \ensuremath{{\rho}}', \ensuremath{{\mu}}' \rangle}} = s'$, and $s' \sim t'$.
        \item \emph{\refif}\hskip 1.2em Then, $c \;=\; \ensuremath{{\mathcolor{mauve}{\mathtt{if}}~{\ensuremath{{\{{e_s}\mathbin{\text{\normalfont\guilsinglright}}{e_t}\}}}e}~\mathcolor{mauve}{\mathtt{then}}~c_\mathsf{t\kern-0.9ptt}~\mathcolor{mauve}{\mathtt{else}}~c_\mathsf{f\kern-1.1ptf}}}$ and $\ensuremath{{\llparenthesis\,{c}\,\rrparenthesis}} \;= \;\ensuremath{{\mathcolor{mauve}{\mathtt{if}}~{\ensuremath{{{e}[{e_s}\mathbin{\text{\normalfont\guilsinglright}}{e_t}]}}}~\mathcolor{mauve}{\mathtt{then}}~\ensuremath{{\llparenthesis\,{c_\mathsf{t\kern-0.9ptt}}\,\rrparenthesis}}~\mathcolor{mauve}{\mathtt{else}}~\ensuremath{{\llparenthesis\,{c_\mathsf{f\kern-1.1ptf}}\,\rrparenthesis}}}}$.
        By definition, $\ensuremath{{\rho}} = \ensuremath{{\rho}}'$, $\ensuremath{{\mu}} = \ensuremath{{\mu}}'$, and $b = \ensuremath{{\llbracket\, \ensuremath{{{e}[{e_s}\mathbin{\text{\normalfont\guilsinglright}}{e_t}]}} \,\rrbracket_{\ensuremath{{\rho}}}}}$ chooses $c'' = \ensuremath{{\llparenthesis\,{c_b}\,\rrparenthesis}}$ and $o_t = \ensuremath{{\mathsf{br}\; b}}$.
        Due to \Cref{prop:expr-subst-guarantee}, we have that $\ensuremath{{\llbracket\, e \,\rrbracket_{\ensuremath{{\rho}}}}} = \ensuremath{{\llbracket\, \ensuremath{{{e}[{e_s}\mathbin{\text{\normalfont\guilsinglright}}{e_t}]}} \,\rrbracket_{\ensuremath{{\rho}}}}} = b$.
        Thus, $s \step{}{\ensuremath{{\mathsf{br}\; b}}} \ensuremath{{\langle c_b, \ensuremath{{\rho}}, \ensuremath{{\mu}} \rangle}} = s'$, and $s' \sim t'$.
        \item \emph{\refwhile}\hskip 1.2em Then, $c \;=\; \ensuremath{{\mathcolor{mauve}{\mathtt{while}}~{\ensuremath{{\{{e_s}\mathbin{\text{\normalfont\guilsinglright}}{e_t}\}}}e}~\mathcolor{mauve}{\mathtt{do}}~{c_\mathsf{t\kern-0.9ptt}}}}$ and $\ensuremath{{\llparenthesis\,{c}\,\rrparenthesis}} \;= \;\ensuremath{{\mathcolor{mauve}{\mathtt{while}}~{\ensuremath{{{e}[{e_s}\mathbin{\text{\normalfont\guilsinglright}}{e_t}]}}}~\mathcolor{mauve}{\mathtt{do}}~{\ensuremath{{\llparenthesis\,{c_\mathsf{t\kern-0.9ptt}}\,\rrparenthesis}}}}}$.
        By definition, $\ensuremath{{\rho}} = \ensuremath{{\rho}}'$, $\ensuremath{{\mu}} = \ensuremath{{\mu}}'$, and $b = \ensuremath{{\llbracket\, \ensuremath{{{e}[{e_s}\mathbin{\text{\normalfont\guilsinglright}}{e_t}]}} \,\rrbracket_{\ensuremath{{\rho}}}}}$ and $o_t = \ensuremath{{\mathsf{br}\; b}}$.
        Due to \Cref{prop:expr-subst-guarantee}, we have that $\ensuremath{{\llbracket\, e \,\rrbracket_{\ensuremath{{\rho}}}}} = \ensuremath{{\llbracket\, \ensuremath{{{e}[{e_s}\mathbin{\text{\normalfont\guilsinglright}}{e_t}]}} \,\rrbracket_{\ensuremath{{\rho}}}}} = b$.
        If $b = \mathsf{t\kern-0.9ptt}$ chooses $c' = \ensuremath{{\llparenthesis\,{c_\mathsf{t\kern-0.9ptt}}\,\rrparenthesis}} \mathop{;} \ensuremath{{\mathcolor{mauve}{\mathtt{while}}~{\ensuremath{{{e}[{e_s}\mathbin{\text{\normalfont\guilsinglright}}{e_t}]}}}~\mathcolor{mauve}{\mathtt{do}}~{\ensuremath{{\llparenthesis\,{c_\mathsf{t\kern-0.9ptt}}\,\rrparenthesis}}}}}$, then
        Thus, $s \step{}{\ensuremath{{\mathsf{br}\; \mathsf{t\kern-0.9ptt}}}} \ensuremath{{\langle c_\mathsf{t\kern-0.9ptt}\mathop{;}\ensuremath{{\mathcolor{mauve}{\mathtt{while}}~{\ensuremath{{\{{e_s}\mathbin{\text{\normalfont\guilsinglright}}{e_t}\}}}e}~\mathcolor{mauve}{\mathtt{do}}~{c}}}{c_\mathsf{t\kern-0.9ptt}}, \ensuremath{{\rho}}, \ensuremath{{\mu}} \rangle}} = s'$, and $s' \sim t'$.
        Similarly, if $b = \mathsf{f\kern-1.1ptf}$ chooses $c' = \ensuremath{{\mathcolor{mauve}{\mathtt{skip}}}}$, then 
        $s \step{}{\ensuremath{{\mathsf{br}\; \mathsf{f\kern-1.1ptf}}}} \ensuremath{{\langle \ensuremath{{\mathcolor{mauve}{\mathtt{skip}}}}, \ensuremath{{\rho}}, \ensuremath{{\mu}} \rangle}} = s'$, and $s' \sim t'$.
        \item \emph{\refseq}\hskip 1.2em Then, $c = c_1 \mathop{;} c_2$ and $\ensuremath{{\llparenthesis\,{c}\,\rrparenthesis}} = \ensuremath{{\llparenthesis\,{c_1}\,\rrparenthesis}}{c_2}$ and $\ensuremath{{\langle \ensuremath{{\llparenthesis\,{c_1}\,\rrparenthesis}}, \ensuremath{{\rho}}, \ensuremath{{\mu}} \rangle}} \step{}{o_t} \ensuremath{{\langle c_1'', \ensuremath{{\rho}}', \ensuremath{{\mu}}' \rangle}}$.
        By induction hypothesis, $\ensuremath{{\langle c_1, \ensuremath{{\rho}}, \ensuremath{{\mu}} \rangle}} \step{}{o_t} \ensuremath{{\langle c_1', \ensuremath{{\rho}}', \ensuremath{{\mu}}' \rangle}}$ so that $\ensuremath{{\langle c_1', \ensuremath{{\rho}}', \ensuremath{{\mu}}' \rangle}} \sim \ensuremath{{\langle c_1'', \ensuremath{{\rho}}', \ensuremath{{\mu}}' \rangle}}$, i.e., $c_1'' = \ensuremath{{\llparenthesis\,{c_1'}\,\rrparenthesis}}$.
        Finally, we have~$s' = \ensuremath{{\langle c_1' \mathop{;} c_2, \ensuremath{{\rho}}', \ensuremath{{\mu}}' \rangle}} \sim \ensuremath{{\langle \ensuremath{{\llparenthesis\,{c_1'}\,\rrparenthesis}}\mathop{;}\ensuremath{{\llparenthesis\,{c_2}\,\rrparenthesis}}, \ensuremath{{\rho}}', \ensuremath{{\mu}}' \rangle}} = t'$. \qedhere
    \end{enumerate}
\end{proof}
\subsection{Dead Branch Elimination}
\label{appendix:dead-branch-elimination}

We define $\ensuremath{{\llparenthesis\,{\cdot}\,\rrparenthesis}}$ formally, where $e' \notin \ensuremath{{\{\mathsf{t\kern-0.9ptt}, \mathsf{f\kern-1.1ptf}\}}}$ is no constant boolean, whereas $b \in \ensuremath{{\{\mathsf{t\kern-0.9ptt}, \mathsf{f\kern-1.1ptf}\}}}$ is:
\begin{gather*}
    \begin{align*}
        \ensuremath{{\llparenthesis\,{\ensuremath{{\mathcolor{mauve}{\mathtt{skip}}}}}\,\rrparenthesis}} & \;=\; \ensuremath{{\mathcolor{mauve}{\mathtt{skip}}}}
                     &
        \ensuremath{{\llparenthesis\,{\ensuremath{{{x}\mathcolor{cyan}{=}{e}}}}\,\rrparenthesis}} & \;=\;  \ensuremath{{{x}\mathcolor{cyan}{=}{e}}}&
        \\
        \ensuremath{{\llparenthesis\,{\ensuremath{{{x}\mathcolor{cyan}{=}{\left[e\right]}}}}\,\rrparenthesis}} & \;=\;  \ensuremath{{{x}\mathcolor{cyan}{=}{\left[e\right]}}}
                                                &
        \ensuremath{{\llparenthesis\,{\ensuremath{{{\left[e\right]}\mathcolor{cyan}{=}{x}}}}\,\rrparenthesis}} & \;=\;  \ensuremath{{{e}\mathcolor{cyan}{=}{\left[x\right]}}}&
    \end{align*}
    \\
    \begin{align*}
        \ensuremath{{\llparenthesis\,{\ensuremath{{\mathcolor{mauve}{\mathtt{if}}~{b}~\mathcolor{mauve}{\mathtt{then}}~c_\mathsf{t\kern-0.9ptt}~\mathcolor{mauve}{\mathtt{else}}~c_\mathsf{f\kern-1.1ptf}}}}\,\rrparenthesis}} &\;=\;  \ensuremath{{\llparenthesis\,{c_b}\,\rrparenthesis}}&
        \\
        \ensuremath{{\llparenthesis\,{\ensuremath{{\mathcolor{mauve}{\mathtt{if}}~{e'}~\mathcolor{mauve}{\mathtt{then}}~c_\mathsf{t\kern-0.9ptt}~\mathcolor{mauve}{\mathtt{else}}~c_\mathsf{f\kern-1.1ptf}}}}\,\rrparenthesis}} &\;=\;  \ensuremath{{\mathcolor{mauve}{\mathtt{if}}~{e}~\mathcolor{mauve}{\mathtt{then}}~\ensuremath{{\llparenthesis\,{c_\mathsf{t\kern-0.9ptt}}\,\rrparenthesis}}~\mathcolor{mauve}{\mathtt{else}}~\ensuremath{{\llparenthesis\,{c_\mathsf{f\kern-1.1ptf}}\,\rrparenthesis}}}}&
        \\
        \ensuremath{{\llparenthesis\,{\ensuremath{{\mathcolor{mauve}{\mathtt{while}}~{e}~\mathcolor{mauve}{\mathtt{do}}~{c}}}}\,\rrparenthesis}} &\;=\;  \ensuremath{{\mathcolor{mauve}{\mathtt{while}}~{e}~\mathcolor{mauve}{\mathtt{do}}~{\ensuremath{{\llparenthesis\,{c}\,\rrparenthesis}}}}}&
        \\
        \ensuremath{{\llparenthesis\,{c_1 \mathop{;} c_2}\,\rrparenthesis}} &\;=\;  \ensuremath{{\llparenthesis\,{c_1}\,\rrparenthesis}} \mathop{;} \ensuremath{{\llparenthesis\,{c_2}\,\rrparenthesis}} &
    \end{align*}
\end{gather*}

We defined $\sim$ to hold for $\ensuremath{{\langle c, \ensuremath{{\rho}}, \ensuremath{{\mu}} \rangle}} \sim \ensuremath{{\langle \ensuremath{{\llparenthesis\,{c}\,\rrparenthesis}}, \ensuremath{{\rho}}, \ensuremath{{\mu}} \rangle}}$,
i.e., variable and memory coincides.
Transformer $\Tobsname{}$, number-of-steps function $\ensuremath{{\mathsf{ns}}}$, and suffixes $\ensuremath{{\mathsf{sf}}}$ are defined in \Cref{fig:sim-dbe}.

\begin{proof}[Proof of \Cref{thm:ct:soundness:simulation}]
    Let $s = \ensuremath{{\langle c, \ensuremath{{\rho}}, \ensuremath{{\mu}} \rangle}} \sim \ensuremath{{\langle \ensuremath{{\llparenthesis\,{c}\,\rrparenthesis}}, \ensuremath{{\rho}}, \ensuremath{{\mu}} \rangle}} = t$.

    For the first requirement of \Cref{def:ct:diagram}, let further $t \step{}{o_t} t' = \ensuremath{{\langle c'', \ensuremath{{\rho}}', \ensuremath{{\mu}}' \rangle}}$.
    We apply case distinction on $c$,
    which is either $c \neq \ensuremath{{\mathcolor{mauve}{\mathtt{if}}~{b}~\mathcolor{mauve}{\mathtt{then}}~c_\mathsf{t\kern-0.9ptt}~\mathcolor{mauve}{\mathtt{else}}~c_\mathsf{f\kern-1.1ptf}}}$,
    or it is $c = \ensuremath{{\mathcolor{mauve}{\mathtt{if}}~{b}~\mathcolor{mauve}{\mathtt{then}}~c_\mathsf{t\kern-0.9ptt}~\mathcolor{mauve}{\mathtt{else}}~c_\mathsf{f\kern-1.1ptf}}}$ with $\ensuremath{{\llparenthesis\,{c}\,\rrparenthesis}} = \ensuremath{{\llparenthesis\,{c_b}\,\rrparenthesis}}$.

    \begin{enumerate}[label=$\blacktriangleright$,leftmargin=*]
        \item \emph{$c \neq \ensuremath{{\mathcolor{mauve}{\mathtt{if}}~{b}~\mathcolor{mauve}{\mathtt{then}}~c_\mathsf{t\kern-0.9ptt}~\mathcolor{mauve}{\mathtt{else}}~c_\mathsf{f\kern-1.1ptf}}}$}\hskip 1.2em then, by definition of $\ensuremath{{\llparenthesis\,{\cdot}\,\rrparenthesis}}$, $s$ executes the same instruction,
        i.e., $s \step{}{o_t} \ensuremath{{\langle c', \ensuremath{{\rho}}', \ensuremath{{\mu}}' \rangle}} = s'$ and $\ensuremath{{\llparenthesis\,{c'}\,\rrparenthesis}} = c''$, so $s' \sim t'$.
        \item \emph{$c = \ensuremath{{\mathcolor{mauve}{\mathtt{if}}~{b}~\mathcolor{mauve}{\mathtt{then}}~c_\mathsf{t\kern-0.9ptt}~\mathcolor{mauve}{\mathtt{else}}~c_\mathsf{f\kern-1.1ptf}}}$}\hskip 1.2em we have $s \step{}{\ensuremath{{\mathsf{br}\; b}}} \ensuremath{{\langle c_b, \ensuremath{{\rho}}, \ensuremath{{\mu}} \rangle}}$.
        Due to $\ensuremath{{\llparenthesis\,{c_b}\,\rrparenthesis}} = \ensuremath{{\llparenthesis\,{c}\,\rrparenthesis}}$, and thus $\ensuremath{{\langle c_b, \ensuremath{{\rho}}, \ensuremath{{\mu}} \rangle}} \sim t$, we can rely on induction to provide a sequence $\ensuremath{{\langle c_b, \ensuremath{{\rho}}, \ensuremath{{\mu}} \rangle}} \step*{}{\ensuremath{{\boldsymbol{o}}}} s'$
        with $s' \sim t'$ and $\Tobs{c_b}{\ensuremath{{\boldsymbol{o}}}} = o_t$.
        We add the first step to obtain $s \step*{}{\ensuremath{{\mathsf{br}\; b}} \mathop{\cdot} \ensuremath{{\boldsymbol{o}}}} s'$ and, by definition of $\Tobsname{}$, $\Tobs{c}{\ensuremath{{\mathsf{br}\; b}} \mathop{\cdot} \ensuremath{{\boldsymbol{o}}}} = o_t$.
    \end{enumerate}

    For the second requirement of \Cref{def:ct:diagram}, let $t$ be final, i.e.,
    $t = \ensuremath{{\langle \ensuremath{{\mathcolor{mauve}{\mathtt{skip}}}}, \ensuremath{{\rho}}, \ensuremath{{\mu}} \rangle}}$.
    We do case distinction on $c$ as well, where either $c = \ensuremath{{\mathcolor{mauve}{\mathtt{skip}}}}$ as well,
    or $c = \ensuremath{{\mathcolor{mauve}{\mathtt{if}}~{b}~\mathcolor{mauve}{\mathtt{then}}~c_\mathsf{t\kern-0.9ptt}~\mathcolor{mauve}{\mathtt{else}}~c_\mathsf{f\kern-1.1ptf}}}$ and again $\ensuremath{{\llparenthesis\,{c_b}\,\rrparenthesis}} = \ensuremath{{\mathcolor{mauve}{\mathtt{skip}}}}$.

    \begin{enumerate}[label=$\blacktriangleright$,leftmargin=*]
        \item \emph{$c = \ensuremath{{\mathcolor{mauve}{\mathtt{skip}}}}$}\hskip 1.2em Trivial.
        \item \emph{$c = \ensuremath{{\mathcolor{mauve}{\mathtt{if}}~{b}~\mathcolor{mauve}{\mathtt{then}}~c_\mathsf{t\kern-0.9ptt}~\mathcolor{mauve}{\mathtt{else}}~c_\mathsf{f\kern-1.1ptf}}}$ and $\ensuremath{{\llparenthesis\,{c_b}\,\rrparenthesis}} = \ensuremath{{\llparenthesis\,{c}\,\rrparenthesis}}$}\hskip 1.2em
        Again, we have $s \step{}{\ensuremath{{\mathsf{br}\; b}}} \ensuremath{{\langle c_b, \ensuremath{{\rho}}, \ensuremath{{\mu}} \rangle}}$ and $\ensuremath{{\langle c_b, \ensuremath{{\rho}}, \ensuremath{{\mu}} \rangle}} \sim t$.
        Further, $\ensuremath{{\ensuremath{{\mathsf{sf}}}({c})}} = \ensuremath{{\mathsf{br}\; b}} \mathop{\cdot} \ensuremath{{\ensuremath{{\mathsf{sf}}}({c_b})}}$.
        Induction yields $\ensuremath{{\langle c_b, \ensuremath{{\rho}}, \ensuremath{{\mu}} \rangle}} \step*{}{\ensuremath{{\ensuremath{{\mathsf{sf}}}({c_b})}}} s'$ with $s' = \ensuremath{{\langle \ensuremath{{\mathcolor{mauve}{\mathtt{skip}}}}, \ensuremath{{\rho}}, \ensuremath{{\mu}} \rangle}}$.
        We stitch them together for $s \step*{}{\ensuremath{{\mathsf{br}\; b}}\mathop{\cdot}\ensuremath{{\ensuremath{{\mathsf{sf}}}({c_b})}}} s'$, and by definition $\ensuremath{{\ensuremath{{\mathsf{sf}}}({c})}} = \ensuremath{{\mathsf{br}\; b}}\mathop{\cdot}\ensuremath{{\ensuremath{{\mathsf{sf}}}({c_b})}}$.\qedhere
    \end{enumerate}
\end{proof}
\subsection{Dead Assignment Elimination}
\label{appendix:dead-assignment-elimination}

\newcommand{\annot}[1]{\{\,#1\,\}}
\newcommand{\cnilannot}[1]{\annot{#1}\ensuremath{{\mathcolor{mauve}{\mathtt{skip}}}}}
\newcommand{\iassignannot}[3]{\annot{#3}\ensuremath{{{#1}\mathcolor{cyan}{=}{#2}}}}
\newcommand{\iloadannot}[3]{\annot{#3}\ensuremath{{{#1}\mathcolor{cyan}{=}{\left[#2\right]}}}}
\newcommand{\istoreannot}[3]{\annot{#3}\ensuremath{{{\left[#1\right]}\mathcolor{cyan}{=}{#2}}}}
\newcommand{\iifannot}[4]{\annot{#4}\ensuremath{{\mathcolor{mauve}{\mathtt{if}}~{#1}~\mathcolor{mauve}{\mathtt{then}}~#2~\mathcolor{mauve}{\mathtt{else}}~#3}}}
\newcommand{\iwhileannot}[3]{\annot{#3}\ensuremath{{\mathcolor{mauve}{\mathtt{while}}~{#1}~\mathcolor{mauve}{\mathtt{do}}~{#2}}}}
\newcommand{\headinst}[1]{\mathit{hd}(#1)}
\newcommand{\cfreach}[1]{\mathit{cfr}(#1)}
As a matter of fact, the analysis of the input program must label every program point of the program with the analysis result.
An annotated input program thus stems from the syntax:
\begin{align}
    a &\Coloneqq \iassignannot{x}{e}{D} \mid \iloadannot{x}{e}{D} \mid \istoreannot{e}{x}{D} \\
    c &\Coloneqq \cnilannot{D} \mid c \mathop{;} c \mid a \mid \iifannot{e}{c_\mathsf{t\kern-0.9ptt}}{c_\mathsf{f\kern-1.1ptf}}{D} \mid \iwhileannot{e}{c}{D},
\end{align}
where $a$ are instructions as before and $D \subseteq \ensuremath{{\mathit{Var}}}\cup\ensuremath{{\mathit{Adr}}}$ is the set of annotated dead variables and memory addresses.
The annotations do not modify the semantics.

To express the correctness guarantee of Dead Assignment Elimination, we need to define the head instruction $\headinst{c}$ of a program $c$.
It is the next instruction that will be executed. 
\begin{align*}
    \headinst{c_1\mathop{;} c_2} &= \headinst{c_1} & \headinst{c} &= c \quad\text{for all $c \neq c_1 \mathop{;} c_2$}\,.
\end{align*}
We write $\annot{D}c$ for a command $c$, where $\headinst{c}$ is annotated by $D$.
Further, we define $=_D$, which expresses equality of two states up to a set of dead variables and memory addresses.
\begin{align*}
    \ensuremath{{\rho}} &=_D \ensuremath{{\rho}}' & \text{when } & \forall x\in\ensuremath{{\mathit{Var}}} \setminus D. \ensuremath{{\rho}}(x) = \ensuremath{{\rho}}'(x) \\
    \ensuremath{{\mu}} &=_D \ensuremath{{\mu}}' & \text{when } & \forall n\in\ensuremath{{\mathit{Adr}}} \setminus D. \ensuremath{{\mu}}(n) = \ensuremath{{\mu}}'(n)\,.
\end{align*}
At last, we define, when a program point is control flow reachable in program $P$.
Intuitively, $c'$ is control flow reachable from $c$ by unfolding $c$ and removing instructions from the front.
Formally, we define the control flow reachable set $\cfreach{c}$ as the smallest set, so that:
\begin{align*}
    &&&&c & \in \cfreach{c} \\
    \ensuremath{{\mathcolor{mauve}{\mathtt{if}}~{e}~\mathcolor{mauve}{\mathtt{then}}~c_{\mathsf{t\kern-0.9ptt}}~\mathcolor{mauve}{\mathtt{else}}~c_{\mathsf{f\kern-1.1ptf}}}} & \in \cfreach{c} &&\implies & c_\mathsf{t\kern-0.9ptt}, c_\mathsf{f\kern-1.1ptf} & \in \cfreach{c} \\
    \ensuremath{{\mathcolor{mauve}{\mathtt{while}}~{e}~\mathcolor{mauve}{\mathtt{do}}~{c'}}} & \in \cfreach{c} &&\implies & c'& \in \cfreach{c} \\
    c_1 \mathop{;} c_2 & \in \cfreach{c} &&\implies & \cfreach{c_1} \mathop{;} c_2 & \subseteq \cfreach{c} \land \cfreach{c_2} \subseteq \cfreach{c}\,.
\end{align*}

The input program annotations need to satisfy the following correctness guarantee:
\begin{proposition}
    For all $\annot{D_1}c_1 \in \cfreach{P}$, when
    $\ensuremath{{\langle \annot{D_1}c_1, \ensuremath{{\rho}}_1, \ensuremath{{\mu}}_1 \rangle}} \step{}{o} \ensuremath{{\langle \annot{D_2}c_2, \ensuremath{{\rho}}_2, \ensuremath{{\mu}}_2 \rangle}}$ holds,
    then for any $\ensuremath{{\rho}}_1' =_{D_1} \ensuremath{{\rho}}_1$, $\ensuremath{{\mu}}_1' =_{D_1} \ensuremath{{\mu}}_1$,
    we have $\ensuremath{{\langle \annot{D_1}c_1, \ensuremath{{\rho}}_1', \ensuremath{{\mu}}_1' \rangle}} \step{}{o} \ensuremath{{\langle \annot{D_2}c_2, \ensuremath{{\rho}}_2', \ensuremath{{\mu}}_2' \rangle}}$
    with $\ensuremath{{\rho}}_2 =_{D_2} \ensuremath{{\rho}}_2'$, $\ensuremath{{\mu}}_2 =_{D_2} \ensuremath{{\mu}}_2'$.
\end{proposition}

We define $\ensuremath{{\llparenthesis\,{\cdot}\,\rrparenthesis}}$.
Note, that to eliminate an assignment, the variable needs to be dead \emph{after} the assignment.
To simplify the presentation, we only eliminate assignments which are at the left of $\mathop{;}$, where the succeeding set of dead variables is easily attainable.
\begin{gather*}
    \begin{align*}
        \ensuremath{{\llparenthesis\,{\cnilannot{D}}\,\rrparenthesis}} & \;=\; \ensuremath{{\mathcolor{mauve}{\mathtt{skip}}}}
                     &
        \ensuremath{{\llparenthesis\,{\iassignannot{x}{e}{D}}\,\rrparenthesis}} & = \ensuremath{{{x}\mathcolor{cyan}{=}{e}}}
        \\
            \ensuremath{{\llparenthesis\,{\iloadannot{x}{e}{D}}\,\rrparenthesis}} & \;=\;  \ensuremath{{{x}\mathcolor{cyan}{=}{\left[e\right]}}}
                                                &
            \ensuremath{{\llparenthesis\,{\istoreannot{e}{x}{D}}\,\rrparenthesis}} & \;=\;  \ensuremath{{{e}\mathcolor{cyan}{=}{\left[x\right]}}}&
    \end{align*}
    \\
    \begin{align*}
        \ensuremath{{\llparenthesis\,{\iifannot{e}{c_\mathsf{t\kern-0.9ptt}}{c_\mathsf{f\kern-1.1ptf}}{D}}\,\rrparenthesis}} &\;=\;  \ensuremath{{\mathcolor{mauve}{\mathtt{if}}~{e}~\mathcolor{mauve}{\mathtt{then}}~c_\mathsf{t\kern-0.9ptt}~\mathcolor{mauve}{\mathtt{else}}~c_\mathsf{f\kern-1.1ptf}}}&
        \\
        \ensuremath{{\llparenthesis\,{\iwhileannot{e}{c}{D}}\,\rrparenthesis}} &\;=\;  \ensuremath{{\mathcolor{mauve}{\mathtt{while}}~{e}~\mathcolor{mauve}{\mathtt{do}}~{\ensuremath{{\llparenthesis\,{c}\,\rrparenthesis}}}}} \\
        \ensuremath{{\llparenthesis\,{c_1 \mathop{;} c_2}\,\rrparenthesis}} &\;=\;  \ensuremath{{\llparenthesis\,{c_1}\,\rrparenthesis}} \mathop{;} \ensuremath{{\llparenthesis\,{c_2}\,\rrparenthesis}} \qquad\text{where $c_1 \neq \iassignannot{x}{e}{D}$}\\
        \ensuremath{{\llparenthesis\,{\iassignannot{x}{e}{D}\mathop{;}\annot{D'}{c}}\,\rrparenthesis}} & \;=\; 
        \begin{cases}
                \ensuremath{{{x}\mathcolor{cyan}{=}{e}}}\mathop{;}\ensuremath{{\llparenthesis\,{c}\,\rrparenthesis}}& x \notin D'\\
                \ensuremath{{\llparenthesis\,{c}\,\rrparenthesis}}& x\in D'
            \end{cases}&
    \end{align*}
\end{gather*}

We define $\sim$ to hold at
\begin{align*}
    \ensuremath{{\langle \annot{D}c, \ensuremath{{\rho}}_s, \ensuremath{{\mu}}_s \rangle}} &\sim \ensuremath{{\langle \ensuremath{{\llparenthesis\,{c}\,\rrparenthesis}}, \ensuremath{{\rho}}_t, \ensuremath{{\mu}}_t \rangle}} &&\text{when} && \ensuremath{{\rho}}_s =_D \ensuremath{{\rho}}_t \land \ensuremath{{\mu}}_s =_D \ensuremath{{\mu}}_t \land \annot{D}c\in\cfreach{P}\,.
\end{align*}

Transformer $\Tobsname{}$, number-of-steps function $\ensuremath{{\mathsf{ns}}}$, and suffixes $\ensuremath{{\mathsf{sf}}}$ are defined in \Cref{fig:sim-dae}.

\begin{proof}[Proof of \Cref{thm:ct:soundness:simulation}] 
Let $s = \ensuremath{{\langle \annot{D}c, \ensuremath{{\rho}}_s, \ensuremath{{\mu}}_s \rangle}} \sim \ensuremath{{\langle \ensuremath{{\llparenthesis\,{c}\,\rrparenthesis}}, \ensuremath{{\rho}}_t, \ensuremath{{\mu}}_t \rangle}} = t$.

For the first requirement of \Cref{def:ct:diagram}, let further $t \step{}{o_t} t' = \ensuremath{{\langle c'', \ensuremath{{\rho}}_t', \ensuremath{{\mu}}_t' \rangle}}$.
We apply case distinction on $c$,
which is either $c = \iassignannot{x}{e}{D}\mathop{;} \annot{D'}c'$ with $x\in D'$,
or not, in which case $\ensuremath{{\llparenthesis\,{c}\,\rrparenthesis}} = \ensuremath{{\llparenthesis\,{c'}\,\rrparenthesis}}$, or $\headinst{c}$ is equal to $\headinst{\ensuremath{{\llparenthesis\,{c}\,\rrparenthesis}}}$ except that there are no annotations and the embedded subcommands are compiled as well.
The first case can only happen a finite number of times, before $\headinst{c'}$ is not a dead assignment.
We apply induction using the second case as the base case.
\begin{enumerate}[label=$\blacktriangleright$,leftmargin=*]
    \item \emph{$c \neq \iassignannot{x}{e}{D}\mathop{;}\annot{D'}{c'}$ with $x\in D'$}\hskip 1.2em
    It is straightforward to see that the equal head instruction yields $s \step{}{o_s} s'$
    where $s' = \ensuremath{{\langle \annot{D'''}c''', \ensuremath{{\rho}}_s', \ensuremath{{\mu}}_s' \rangle}}$ and $c'' =\ensuremath{{\llparenthesis\,{c'''}\,\rrparenthesis}}$.
    We have $\ensuremath{{\rho}}_s =_{D} \ensuremath{{\rho}}_t$ and $\ensuremath{{\mu}}_s =_D \ensuremath{{\mu}}_t$.
    Due to the correctness guarantee, we have $\ensuremath{{\langle c, \ensuremath{{\rho}}_t, \ensuremath{{\mu}}_t \rangle}} \step{}{o_t} \ensuremath{{\langle \annot{D'''}c''', \ensuremath{{\rho}}_t', \ensuremath{{\mu}}_t' \rangle}}$ with $\ensuremath{{\rho}}_s' =_{D'''} \ensuremath{{\rho}}_t'$ and $\ensuremath{{\mu}}_s' =_{D'''} \ensuremath{{\mu}}_t'$.
    That is $s' \sim t'$. Moreover, $o_s = o_t = \Tobs{c}{o_s}$.
    \item \emph{$c = \iassignannot{x}{e}{D}\mathop{;}\annot{D'}c'$ with $x\in D'$}\hskip 1.2em We have $\ensuremath{{\langle c, \ensuremath{{\rho}}_s, \ensuremath{{\mu}}_s \rangle}} \step{}{\ensuremath{{\bullet}}} \ensuremath{{\langle \annot{D'}c', \ensuremath{{\rho}}_s'', \ensuremath{{\mu}}_s \rangle}}$ with $\ensuremath{{\rho}}_s'' = \ensuremath{{{\ensuremath{{\rho}}_s}[{x}\mathbin{\text{\normalfont\guilsinglright}}{\ensuremath{{\llbracket\, e \,\rrbracket_{\ensuremath{{\rho}}_s}}}}]}}$.
    We show that $\ensuremath{{\rho}}_s'' =_{D'} \ensuremath{{\rho}}_t$.
    The correctness guarantee implies the standard kill/gen inequality $(D \setminus \ensuremath{{\{x\}}}) \cup \mathit{vars}(e) \subseteq D'$.
    But $x \in D'$, so $D \subseteq D'$.
    And because $\ensuremath{{\rho}}_s$ differs from $\ensuremath{{\rho}}_s''$ only on $x \in D'$, we have that $\ensuremath{{\rho}}_s'' =_{D'} \ensuremath{{\rho}}_t$.
    This establishes $\ensuremath{{\langle c', \ensuremath{{\rho}}_s'', \ensuremath{{\mu}}_s \rangle}} \sim \ensuremath{{\langle \ensuremath{{\llparenthesis\,{c}\,\rrparenthesis}}, \ensuremath{{\rho}}_t, \ensuremath{{\rho}}_t' \rangle}}$.
    We rely on induction to provide a sequence $\ensuremath{{\langle c', \ensuremath{{\rho}}_s'', \ensuremath{{\mu}} \rangle}} \step*{}{\ensuremath{{\boldsymbol{o}}}} s'$ with $s' \sim t'$ and $\Tobs{c}{\ensuremath{{\boldsymbol{o}}}} = o_t$.
    We add the first step to obtain $s \step*{}{\ensuremath{{\bullet}} \mathop{\cdot} \ensuremath{{\boldsymbol{o}}}} s'$ and, by definition of $\Tobsname{}$, $\Tobs{c}{\ensuremath{{\bullet}} \mathop{\cdot} \ensuremath{{\boldsymbol{o}}}} = o_t$.
\end{enumerate}

For the second requirement of \Cref{def:ct:diagram}, let $t$ be final, i.e.,
$t = \ensuremath{{\langle \ensuremath{{\mathcolor{mauve}{\mathtt{skip}}}}, \ensuremath{{\rho}}, \ensuremath{{\mu}} \rangle}}$.
We do case distinction on $c$ as well, where either $c = \annot{D}\ensuremath{{\mathcolor{mauve}{\mathtt{skip}}}}$ as well,
or $c = \iassignannot{x}{e}{D}\mathop{;} \annot{D'}c$, $x\in D'$ and again $\ensuremath{{\llparenthesis\,{\annot{D'}c}\,\rrparenthesis}} = \ensuremath{{\mathcolor{mauve}{\mathtt{skip}}}}$.
Again, we use the first case as base case for an induction.

\begin{enumerate}[label=$\blacktriangleright$,leftmargin=*]
    \item \emph{$c = \ensuremath{{\mathcolor{mauve}{\mathtt{skip}}}}$}\hskip 1.2em Trivial.
    \item \emph{$c = \iassignannot{x}{e}{D}\mathop{;} \annot{D'}c$ with $c \in D'$}\hskip 1.2em
        As before, $\ensuremath{{\langle c, \ensuremath{{\rho}}_s, \ensuremath{{\mu}}_s \rangle}} \step{}{\ensuremath{{\bullet}}} \ensuremath{{\langle \annot{D'}c', \ensuremath{{\rho}}_s'', \ensuremath{{\mu}}_s \rangle}}$ with $\ensuremath{{\rho}}_s'' = \ensuremath{{{\ensuremath{{\rho}}_s}[{x}\mathbin{\text{\normalfont\guilsinglright}}{\ensuremath{{\llbracket\, e \,\rrbracket_{\ensuremath{{\rho}}_s}}}}]}}$.
        As before, $D \subseteq D'$ yields $\ensuremath{{\rho}}_s'' =_{D'} \ensuremath{{\rho}}_t$, thus $\ensuremath{{\langle c', \ensuremath{{\rho}}_s'', \ensuremath{{\mu}}_s \rangle}} \sim \ensuremath{{\langle \ensuremath{{\llparenthesis\,{c}\,\rrparenthesis}}, \ensuremath{{\rho}}_t, \ensuremath{{\rho}}_t' \rangle}}$.
        We rely on induction to provide a sequence $\ensuremath{{\langle c', \ensuremath{{\rho}}_s'', \ensuremath{{\mu}} \rangle}} \step*{}{\ensuremath{{\boldsymbol{o}}}} s'$ with $s'$ final and $s' \sim t'$ and $\ensuremath{{\ensuremath{{\mathsf{sf}}}({c'})}} = \ensuremath{{\boldsymbol{o}}}$.
        We add the first step to obtain $s \step*{}{\ensuremath{{\bullet}} \mathop{\cdot} \ensuremath{{\boldsymbol{o}}}} s'$ and, by definition of $\Tobsname{}$, $\ensuremath{{\ensuremath{{\mathsf{sf}}}({c})}} =\ensuremath{{\bullet}} \mathop{\cdot} \ensuremath{{\boldsymbol{o}}}$. \qedhere
\end{enumerate}
\end{proof}

\subsection{Structural analysis}\label{appendix:structural-analysis}

\newcommand{\sgraph}[1]{\mathit{SG}(#1)}
\newcommand{\nentry}{n_\mathit{entry}}
We begin by presenting the construction of $\ensuremath{{\mathit{struct}_{G}}} : L_{G} \to \ensuremath{{\mathit{Struc}}}$, that takes a label of $G$ to produce the corresponding structured program that sits at the same instruction.
When done, it also provides $\ensuremath{{\llparenthesis\,{G}\,\rrparenthesis}} = \ensuremath{{\ensuremath{{\mathit{struct}_{G}}}(\ell_\ensuremath{{\mathit{init}}})}}$.
As discussed in the main paper, we only search for simple loop and branch patterns in $G$ (\Cref{fig:example-patterns}).
For that, we first define the successor graph $\sgraph{G} = (L_G,E_G)$ of $G$.
The set of nodes $L_G$ is the set of labels in $G$ and $E_G$ is the set of successor edges, i.e.,
\begin{align*}
    L_G &= \{\ell, \ell'\mid\ensuremath{{\ell: i \mathrel{\text{\normalfont\guilsinglright}}\ell'}}\} \cup \{\ell, \ell_{\mathsf{t\kern-0.9ptt}}, \ell_{\mathsf{f\kern-1.1ptf}}\mid\ensuremath{{\ell: e \mathrel{\text{\normalfont\guilsinglright}}\ell_{\mathsf{t\kern-0.9ptt}}}}{\ell_\mathsf{f\kern-1.1ptf}}\},
    \\
    E_G &= \{ (\ell, \ell') \mid\ensuremath{{\ell: i \mathrel{\text{\normalfont\guilsinglright}}\ell'}}\} \cup \{ (\ell, \ell_{\mathsf{t\kern-0.9ptt}}), (\ell, \ell_{\mathsf{f\kern-1.1ptf}})\mid\ensuremath{{\ell: e \mathrel{\text{\normalfont\guilsinglright}}\ell_{\mathsf{t\kern-0.9ptt}}}}{\ell_\mathsf{f\kern-1.1ptf}}\}.
\end{align*}
We require that an analysis algorithm has annotated the labels $\ell \in L_G$, that are conditional branchings $\ensuremath{{\ell: e \mathrel{\text{\normalfont\guilsinglright}}\ell_{\mathsf{t\kern-0.9ptt}}}}{\ell_{\mathsf{f\kern-1.1ptf}}}$.
The annotations identify regions $L \subseteq L_G$ that correspond to control structures.
There are two types of annotations for $\ensuremath{{\ell: e \mathrel{\text{\normalfont\guilsinglright}}\ell_{\mathsf{t\kern-0.9ptt}}}}{\ell_{\mathsf{f\kern-1.1ptf}}}$:
\[
   \ensuremath{{\ell: \ensuremath{{\mathcolor{mauve}{\mathtt{while}}~{\ensuremath{{\mathit{e}}}}~\mathcolor{mauve}{\mathtt{do}}~{L_w}}} \mathrel{\text{\normalfont\guilsinglright}}\ell_{\mathsf{t\kern-0.9ptt}},\ell_{\mathsf{f\kern-1.1ptf}}}},
   \qquad
   \ensuremath{{\ell: \ensuremath{{\mathcolor{mauve}{\mathtt{if}}~{\ensuremath{{\mathit{e}}}}~\mathcolor{mauve}{\mathtt{then}}~L_{\mathsf{t\kern-0.9ptt}}~\mathcolor{mauve}{\mathtt{else}}~L_{\mathsf{f\kern-1.1ptf}}}} \mathrel{\text{\normalfont\guilsinglright}}\ell_{\mathsf{t\kern-0.9ptt}},\ell_{\mathsf{f\kern-1.1ptf}}\mathop{;}\ell'}}.
\]
The first asserts that $\ell$ is the entry and exit condition to a while loop, where $L_w$ is the loop body region, a strongly connected component including $\ell \in L_w$.
The requirement is, that the loop can only be entered and exited through $\ell$, i.e.,
\[
    E_G \cap (L_w \times (L_G \setminus L_w)) = \ensuremath{{\{(\ell, \ell_{\mathsf{f\kern-1.1ptf}})\}}}
    \qquad\text{and}\qquad
    E_G \cap ((L_G \setminus L_w) \times L_w) = E_G \cap ((L_G \setminus L_w) \times \ensuremath{{\{\ell\}}}).
\]
The second asserts that $\ell$ is a conditional branch, where $L_\mathsf{t\kern-0.9ptt}$ and $L_\mathsf{f\kern-1.1ptf}$ are the disjoint regions that form the branch bodies.
The requirement is, that the regions can only be entered through $\ell$,
i.e.,
\[
    E_G \cap ((L_G \setminus L_\mathsf{t\kern-0.9ptt}) \times L_\mathsf{t\kern-0.9ptt}) = \ensuremath{{\{(\ell, \ell_{\mathsf{t\kern-0.9ptt}})\}}}
    \qquad\text{and}\qquad
    E_G \cap ((L_G \setminus L_\mathsf{f\kern-1.1ptf}) \times L_\mathsf{f\kern-1.1ptf}) = \ensuremath{{\{(\ell, \ell_{\mathsf{f\kern-1.1ptf}})\}}},
\]
and that they exit to the common join point $\ell'$,
i.e.,
\[
    E_G \cap ( L_\mathsf{t\kern-0.9ptt} \times (L_G \setminus L_\mathsf{t\kern-0.9ptt})) = E_G \cap ( L_\mathsf{t\kern-0.9ptt} \times \ensuremath{{\{\ell'\}}})
    \qquad\text{and}\qquad
    E_G \cap ( L_\mathsf{t\kern-0.9ptt} \times (L_G \setminus L_\mathsf{f\kern-1.1ptf})) = E_G \cap ( L_\mathsf{f\kern-1.1ptf} \times \ensuremath{{\{\ell'\}}}).
\]
Lastly, we require that all regions annotated are well-nested.
That means, that for all annotated regions $L_1$, $L_2$ throughout all annotations,
either $L_1 \subseteq L_2$, or $L_2 \subseteq L_1$, or $L_1 \cap L_2 = \varnothing$.

Provided the analysis, we strive to define $\ensuremath{{\mathit{struct}_{G}}}$.
To do so, we define an intermediate function $C_{L,\ellexit} : L \cup \ensuremath{{\{\ellexit\}}} \to \ensuremath{{\mathit{Struc}}}$.
It is parametric in an annotated control region $L \subseteq L_G$ with the region's single exit point $\ellexit$ 
(When $L$ is a while-region, then $\ellexit \in L$ is the entry and exit point, the loop condition;
when $L$ belongs to a conditional branch, then $\ellexit \notin L$ is the join point of both branches).
It recursively translates each control region $L$ with exit point $\ellexit$ into a structured program that terminates when reaching the exit point.
Then, it stitches the obtained subprograms together to obtain the structured program for the full program $P$:
\begin{align*}
    C_{L,\ell'}(\ell) & = 
    \begin{cases}
        \ensuremath{{\mathcolor{mauve}{\mathtt{skip}}}} & \ell = \ell' \\
        \ensuremath{{\mathit{a}}}\mathop{;}C_{L,\ell'}(\ell_{s}) 
              & \ensuremath{{\ell: \ensuremath{{\mathit{a}}} \mathrel{\text{\normalfont\guilsinglright}}\ell_{s}}} \\
        \ensuremath{{\mathcolor{mauve}{\mathtt{while}}~{\ensuremath{{\mathit{e}}}}~\mathcolor{mauve}{\mathtt{do}}~{C_{L_w,\ell}(\ell_\mathsf{t\kern-0.9ptt})}}}  \mathop{;} C_{L,\ell'}(\ell_{\mathsf{f\kern-1.1ptf}})
              & \ensuremath{{\ell: \ensuremath{{\mathcolor{mauve}{\mathtt{while}}~{\ensuremath{{\mathit{e}}}}~\mathcolor{mauve}{\mathtt{do}}~{L_w}}} \mathrel{\text{\normalfont\guilsinglright}}\ell_{\mathsf{t\kern-0.9ptt}},\ell_{\mathsf{f\kern-1.1ptf}}}}\\
        \ensuremath{{\mathcolor{mauve}{\mathtt{if}}~{\ensuremath{{\mathit{e}}}}~\mathcolor{mauve}{\mathtt{then}}~C_{L_\mathsf{t\kern-0.9ptt},\ell''}(\ell_{\mathsf{t\kern-0.9ptt}})~\mathcolor{mauve}{\mathtt{else}}~C_{L_\mathsf{f\kern-1.1ptf},\ell''}(\ell_\mathsf{f\kern-1.1ptf})}} \mathop{;} C_{L,\ell'}(\ell'')
              & \ensuremath{{\ell: \ensuremath{{\mathcolor{mauve}{\mathtt{if}}~{\ensuremath{{\mathit{e}}}}~\mathcolor{mauve}{\mathtt{then}}~L_{\mathsf{t\kern-0.9ptt}}~\mathcolor{mauve}{\mathtt{else}}~L_{\mathsf{f\kern-1.1ptf}}}} \mathrel{\text{\normalfont\guilsinglright}}\ell_{\mathsf{t\kern-0.9ptt}},\ell_{\mathsf{f\kern-1.1ptf}},\ell''}}\,.
    \end{cases}
\end{align*}
At the exit point of $L$ no program is left to execute within $L$.
Basic instructions $\ensuremath{{\mathit{a}}}$ are just prepended to the translation of their successor.
Encountering an annotated control structure first recursively translates the body of the structure before appending its successor's translation.
Structural analysis transforms the prorgram via $\ensuremath{{\llparenthesis\,{P}\,\rrparenthesis}} = C_{L_G,\ell_\ensuremath{{\mathit{ret}}}}(\ell_\ensuremath{{\mathit{init}}})$,
i.e., the command corresponding to $\ell_\ensuremath{{\mathit{init}}}$ in the full region $L_G$ with single exit point being the final label $\ell_\ensuremath{{\mathit{ret}}}$.

The desired function $\ensuremath{{\mathit{struct}_{G}}}$ is more advanced.
It also constructs the structured programs encountered during execution of $\ensuremath{{\llparenthesis\,{G}\,\rrparenthesis}}$ (i.e., the $c \in \cfreach{\ensuremath{{\llparenthesis\,{G}\,\rrparenthesis}}}$.
This is important, because loop bodies are unrolled in the structured language, and compiler transformation aims to only translate full control structures.
\[
    \ensuremath{{\ensuremath{{\mathit{struct}_{G}}}(\ell)}} = \begin{cases}
        C_{L,\ell'}(\ell)\mathop{;} \ensuremath{{\ensuremath{{\mathit{struct}_{G}}}(\ell')}} & \begin{aligned}
            &\text{smallest annotated $L$ with $\ell \in L$, and $\ell' \neq \ell$, and }\\
            &\ensuremath{{\ell': \ensuremath{{\mathcolor{mauve}{\mathtt{while}}~{\ensuremath{{\mathit{e}}}}~\mathcolor{mauve}{\mathtt{do}}~{L}}} \mathrel{\text{\normalfont\guilsinglright}}\ell_{\mathsf{t\kern-0.9ptt}},\ell_{\mathsf{f\kern-1.1ptf}}}} \end{aligned}\\
        C_{L_G,\ell_\ensuremath{{\mathit{ret}}}}(\ell) & \text{if none exists}
    \end{cases}
\]

We prove the following Lemma, which links the head instruction of $\ensuremath{{\ensuremath{{\mathit{struct}_{G}}}(\ell)}}$ to $\ell$.
\begin{lemma}
    Given a label $\ell \in L_G$, we have that 
    \[
        \ensuremath{{\ensuremath{{\mathit{struct}_{G}}}(\ell)}} = \begin{cases}
            a \mathop{;} \ensuremath{{\ensuremath{{\mathit{struct}_{G}}}(\ell')}} & \ensuremath{{\ell: a \mathrel{\text{\normalfont\guilsinglright}}\ell'}} \\
            \ensuremath{{\mathcolor{mauve}{\mathtt{while}}~{e}~\mathcolor{mauve}{\mathtt{do}}~{C_{L,\ell_\mathsf{f\kern-1.1ptf}}(\ell_{\mathsf{t\kern-0.9ptt}})}}}\mathop{;} \ensuremath{{\ensuremath{{\mathit{struct}_{G}}}(\ell_\mathsf{f\kern-1.1ptf})}} & \ensuremath{{\ell: \ensuremath{{\mathcolor{mauve}{\mathtt{while}}~{e}~\mathcolor{mauve}{\mathtt{do}}~{L}}} \mathrel{\text{\normalfont\guilsinglright}}\ell_{\mathsf{t\kern-0.9ptt}},\ell_{\mathsf{f\kern-1.1ptf}}}} \\
            \begin{aligned}
                & \ensuremath{{\mathcolor{mauve}{\mathtt{if}}~{e}~\mathcolor{mauve}{\mathtt{then}}~C_{L_{\mathsf{t\kern-0.9ptt}},\ell'}(\ell_{\mathsf{t\kern-0.9ptt}})~\mathcolor{mauve}{\mathtt{else}}~C_{L_{\mathsf{f\kern-1.1ptf}},\ell'}(\ell_{\mathsf{f\kern-1.1ptf}})}} \vphantom{\ell'}\\
                & \hspace{3.5cm}\mathop{;}\smash{\ensuremath{{\ensuremath{{\mathit{struct}_{G}}}(\ell')}}}
            \end{aligned}
            & \ensuremath{{\ell: \ensuremath{{\mathcolor{mauve}{\mathtt{if}}~{e}~\mathcolor{mauve}{\mathtt{then}}~L_\mathsf{t\kern-0.9ptt}~\mathcolor{mauve}{\mathtt{else}}~L_{\mathsf{f\kern-1.1ptf}}}} \mathrel{\text{\normalfont\guilsinglright}}\ell_{\mathsf{t\kern-0.9ptt}},\ell_{\mathsf{f\kern-1.1ptf}}\mathop{;}\ell'}}
                                          \end{cases}
    \]
\end{lemma}

\begin{lemma}
    $\ensuremath{{\langle \ell, \ensuremath{{\rho}}, \ensuremath{{\mu}} \rangle}} \step{}{o} \ensuremath{{\langle \ell', \ensuremath{{\rho}}', \ensuremath{{\mu}}' \rangle}}$ if and only if $\ensuremath{{\langle \ensuremath{{\ensuremath{{\mathit{struct}_{G}}}(\ell)}}, \ensuremath{{\rho}}, \ensuremath{{\mu}} \rangle}} \step{}{o} \ensuremath{{\langle \ensuremath{{\ensuremath{{\mathit{struct}_{G}}}(\ell')}}, \ensuremath{{\rho}}', \ensuremath{{\mu}}' \rangle}}$
\end{lemma}
\begin{proof}
    We do case distinction on the node of $\ell$.
    \begin{enumerate}[label=$\blacktriangleright$,leftmargin=*]
        \item \emph{$\ensuremath{{\ell: a \mathrel{\text{\normalfont\guilsinglright}}\ell'}}$}\hskip 1.2em Then, $\ensuremath{{\ensuremath{{\mathit{struct}_{G}}}(\ell)}} = a \mathop{;} \ensuremath{{\ensuremath{{\mathit{struct}_{G}}}(\ell')}}$. The result follows from the fact that $a$ transforms $\ensuremath{{\rho}}$ and $\ensuremath{{\mu}}$ equally in both semantics.
        \item \emph{$\ensuremath{{\ell: \ensuremath{{\mathcolor{mauve}{\mathtt{while}}~{e}~\mathcolor{mauve}{\mathtt{do}}~{L}}} \mathrel{\text{\normalfont\guilsinglright}}\ell_{\mathsf{t\kern-0.9ptt}},\ell_{\mathsf{f\kern-1.1ptf}}}}$}\hskip 1.2em
        Then, $\ensuremath{{\ensuremath{{\mathit{struct}_{G}}}(\ell)}} = \ensuremath{{\mathcolor{mauve}{\mathtt{while}}~{e}~\mathcolor{mauve}{\mathtt{do}}~{C_{L,\ell_\mathsf{f\kern-1.1ptf}}(\ell_{\mathsf{t\kern-0.9ptt}})}}}\mathop{;} \ensuremath{{\ensuremath{{\mathit{struct}_{G}}}(\ell_\mathsf{f\kern-1.1ptf})}}$.
        Again, for both semantics, $\ensuremath{{\llbracket\, e \,\rrbracket_{\ensuremath{{\rho}}}}}$ is either $\mathsf{t\kern-0.9ptt}$, or $\mathsf{f\kern-1.1ptf}$.
        If it is $\mathsf{t\kern-0.9ptt}$, then
        \[
            \ensuremath{{\langle \ensuremath{{\ensuremath{{\mathit{struct}_{G}}}(\ell)}}, \ensuremath{{\rho}}, \ensuremath{{\mu}} \rangle}}
            \step{}{\ensuremath{{\mathsf{br}\; \mathsf{t\kern-0.9ptt}}}}
            \ensuremath{{\langle C_{L,\ell_{\mathsf{f\kern-1.1ptf}}}(\ell_{\mathsf{t\kern-0.9ptt}}) \mathop{;} \ensuremath{{\mathcolor{mauve}{\mathtt{while}}~{e}~\mathcolor{mauve}{\mathtt{do}}~{C_{L,\ell_{\mathsf{f\kern-1.1ptf}}}(\ell_{\mathsf{t\kern-0.9ptt}})}}} \mathop{;} \ensuremath{{\ensuremath{{\mathit{struct}_{G}}}(\ell_{\mathsf{f\kern-1.1ptf}})}}, \ensuremath{{\rho}}, \ensuremath{{\mu}} \rangle}},
        \]
        Indeed, 
        \[
            \ensuremath{{\ensuremath{{\mathit{struct}_{G}}}(\ell_{\mathsf{t\kern-0.9ptt}})}} =C_{L,\ell_{\mathsf{f\kern-1.1ptf}}}(\ell_{\mathsf{t\kern-0.9ptt}}) \mathop{;} \ensuremath{{\ensuremath{{\mathit{struct}_{G}}}(\ell)}} = C_{L,\ell_{\mathsf{f\kern-1.1ptf}}}(\ell_{\mathsf{t\kern-0.9ptt}}) \mathop{;} \ensuremath{{\mathcolor{mauve}{\mathtt{while}}~{e}~\mathcolor{mauve}{\mathtt{do}}~{C_{L,\ell_{\mathsf{f\kern-1.1ptf}}}(\ell_{\mathsf{t\kern-0.9ptt}})}}} \mathop{;} \ensuremath{{\ensuremath{{\mathit{struct}_{G}}}(\ell_{\mathsf{f\kern-1.1ptf}})}},
        \]
        because $\ell_{\mathsf{t\kern-0.9ptt}} \in L$, and $L$ is the smallest while-body region annotated to an $\ell'' \neq \ell_{\mathsf{t\kern-0.9ptt}}$ containing $\ell_{\mathsf{t\kern-0.9ptt}}$.

        If it is $\mathsf{f\kern-1.1ptf}$, then $\ensuremath{{\langle \ensuremath{{\ensuremath{{\mathit{struct}_{G}}}(\ell)}}, \ensuremath{{\rho}}, \ensuremath{{\mu}} \rangle}} \step{}{\ensuremath{{\mathsf{br}\; \mathsf{f\kern-1.1ptf}}}} \ensuremath{{\langle \ensuremath{{\ensuremath{{\mathit{struct}_{G}}}(\ell_{\mathsf{f\kern-1.1ptf}})}}, \ensuremath{{\rho}}, \ensuremath{{\mu}} \rangle}}$ as required.
        \item \emph{$\ensuremath{{\ell: \ensuremath{{\mathcolor{mauve}{\mathtt{if}}~{e}~\mathcolor{mauve}{\mathtt{then}}~L_\mathsf{t\kern-0.9ptt}~\mathcolor{mauve}{\mathtt{else}}~L_{\mathsf{f\kern-1.1ptf}}}} \mathrel{\text{\normalfont\guilsinglright}}\ell_{\mathsf{t\kern-0.9ptt}},\ell_{\mathsf{f\kern-1.1ptf}}\mathop{;}\ell'}}$}\hskip 1.2em
        Then, 
        \[
            \ensuremath{{\ensuremath{{\mathit{struct}_{G}}}(\ell)}} = (\ensuremath{{\mathcolor{mauve}{\mathtt{if}}~{e}~\mathcolor{mauve}{\mathtt{then}}~C_{L_{\mathsf{t\kern-0.9ptt}},\ell'}(\ell_{\mathsf{t\kern-0.9ptt}})~\mathcolor{mauve}{\mathtt{else}}~C_{L_{\mathsf{f\kern-1.1ptf}},\ell'}(\ell_{\mathsf{f\kern-1.1ptf}})}}) \mathop{;} \ensuremath{{\ensuremath{{\mathit{struct}_{G}}}(\ell')}}\,.
        \]
        Wlog, let $\ensuremath{{\llbracket\, e \,\rrbracket_{\ensuremath{{\rho}}}}} = \mathsf{t\kern-0.9ptt}$.
        We have $\ensuremath{{\langle \ensuremath{{\ensuremath{{\mathit{struct}_{G}}}(\ell)}}, \ensuremath{{\rho}}, \ensuremath{{\mu}} \rangle}} \step{}{\ensuremath{{\mathsf{br}\; \mathsf{t\kern-0.9ptt}}}} \ensuremath{{\langle C_{L_{\mathsf{t\kern-0.9ptt}},\ell'}(\ell_{\mathsf{t\kern-0.9ptt}})\mathop{;}\ensuremath{{\ensuremath{{\mathit{struct}_{G}}}(\ell')}}, \ensuremath{{\rho}}, \ensuremath{{\mu}} \rangle}}$.
        We inspect the latter, to see $\ensuremath{{\ensuremath{{\mathit{struct}_{G}}}(\ell_{\mathsf{t\kern-0.9ptt}})}} = C_{L_{\mathsf{t\kern-0.9ptt}},\ell'}(\ell_{\mathsf{t\kern-0.9ptt}})\mathop{;}\ensuremath{{\ensuremath{{\mathit{struct}_{G}}}(\ell')}}$. \qedhere
    \end{enumerate}
\end{proof}

We define $\ensuremath{{\langle \ell, \ensuremath{{\rho}}, \ensuremath{{\mu}} \rangle}} \sim \ensuremath{{\langle \ensuremath{{\ensuremath{{\mathit{struct}_{G}}}(\ell)}}, \ensuremath{{\rho}}, \ensuremath{{\mu}} \rangle}}$ and $\Tobs{}{o} = o$.
\Cref{thm:ct:lock-step-soundness:simulation} follows from the previous lemma as an immediate consequence.

\subsection{Loop Rotation}

We formalize Loop rotation in more detail.
A loop $(\ell_{\textit{pred}}, \ell_{\textit{entry}}, L)$ with $\ell_{\textit{pred}}, \ell_{\textit{entry}} \in L_G$ and $L \subseteq L_G$ has the following properties:
\begin{enumerate*}[label=\textbf{(\arabic*)}]
    \item $L$ is a strongly connected component in $\sgraph{G}$,
    \item $\ell_{\textit{pred}}$ has only one successor in $\sgraph{G}$, namely
    \item $(\ell_{\textit{pred}},\ell_{\textit{entry}}) \in E_G$ which is the only edge going into $L$, i.e., $((L_G \setminus L) \times L) \cap E_G = \ensuremath{{\{(\ell_{\textit{pred}},\ell_{\textit{entry}})\}}}$.
\end{enumerate*}
Given a loop $(\ell_{\textit{pred}}, \ell_{\textit{entry}}, L)$ of $G$ to rotate, we can define the transformation $\ensuremath{{\llparenthesis\,{\cdot}\,\rrparenthesis}}$:
\[
    \ensuremath{{\llparenthesis\,{G}\,\rrparenthesis}} = (G \setminus \ensuremath{{\{(\ell_{\textit{pred}},i, \ell_{\textit{entry}})\}}}) \cup \ensuremath{{\{(\ell_{\textit{pred}},i,\ell_{\textit{copy}})\}}} \cup
    \begin{cases}
        \ensuremath{{\{(\ell_{\textit{copy}}, i, \ell')\}}} & (\ell_{\textit{entry}}, i, \ell') \in G \\
        \ensuremath{{\{(\ell_{\textit{copy}}, e, \ell_\mathsf{t\kern-0.9ptt}, \ell_\mathsf{f\kern-1.1ptf})\}}} & (\ell_{\textit{entry}}, e, \ell_\mathsf{t\kern-0.9ptt}, \ell_\mathsf{f\kern-1.1ptf}) \in G 
    \end{cases}
\]
where $\ell_{\textit{copy}} \notin L_G$ is fresh.

With the defined simulation 
\[
    \ensuremath{{\langle \ell, \ensuremath{{\rho}}, \ensuremath{{\mu}} \rangle}} \sim \ensuremath{{\langle \ell', \ensuremath{{\rho}}, \ensuremath{{\mu}} \rangle}} \quad\text{when}\quad \ell = \ell' \lor (\ell = \ell_{\textit{entry}} \land \ell' = \ell_{\textit{copy}})
\]
and the observation transformer $\Tobs{}{o} = o$ we prove \Cref{thm:ct:lock-step-soundness:simulation}:
Let $s = \ensuremath{{\langle \ell, \ensuremath{{\rho}}, \ensuremath{{\mu}} \rangle}} \sim \ensuremath{{\langle \ell', \ensuremath{{\rho}}, \ensuremath{{\mu}} \rangle}} = t$, and $t \step{}{o} t'$.

When $\ell = \ell' \neq \ell_{\textit{pred}}$, then $s = t$ and the nodes of $\ell$ are the same in $G$ and $\ensuremath{{\llparenthesis\,{G}\,\rrparenthesis}}$, thus $s \step{}{o} s' = t'$ and $s' \sim t'$.
Similarly, when $\ell = \ell' = \ell_{\textit{pred}}$, then $t' = (\ell_{\textit{copy}}, \ensuremath{{\rho}}', \ensuremath{{\mu}}')$.
$\ell_{\textit{pred}}$ has the same instruction in $G$, but the successor is $\ell_{\textit{entry}}$, thus $s \step{}{o} s' = \ensuremath{{\langle \ell_{\textit{entry}}, \ensuremath{{\rho}}', \ensuremath{{\mu}}' \rangle}}$ and thus $s' \sim t'$.

If, on the other hand, $\ell = \ell_{\textit{entry}}$ and $\ell' = \ell_{\textit{copy}}$, then instruction/branching condition and successors are identical for $\ell_{\textit{entry}}$ in $G$ and $\ell_{\textit{copy}}$ in $\ensuremath{{\llparenthesis\,{G}\,\rrparenthesis}}$,
so $s \step{}{o} s' = t'$.

\newcommand*{\tableone}[0]{
  \footnotesize
  \begin{tabular}{@{}lcc@{}}
    \toprule
    \textbf{\makecell{Code \\ Simplification/Elimination\\ Passes}} & \textbf{Theoretical} & \textbf{Empirical} \\
    \midrule
    adce &  \textcolor{red}{\Lightning}{}  &  \\
    dse  &  \textcolor{red}{\Lightning}{} &  \\

    globalopt &  & \textcolor{red}{\Lightning}{} \\
    gvn       &  & \textcolor{red}{\Lightning}{} \\

    instcombine &  & \textcolor{red}{\Lightning}{} \\
    early-cse   & \textcolor{red}{\Lightning}{} &  \\
    reassociate &  & \textcolor{red}{\Lightning}{} \\
    bdce & \textcolor{red}{\Lightning}{} &  \\

    globaldce &  & \textcolor{green}{\ding{51}}  \\ strip-dead-prototypes & \textcolor{green}{\ding{51}} &  \\
    \bottomrule
  \end{tabular}
}

\newcommand*{\tabletwo}[0]{
  \footnotesize
  \begin{tabular}{@{}lcc@{}}
    \toprule
    \textbf{\makecell{Loop \\ Optimizations \\ passes}} & \textbf{Theoretical} & \textbf{Empirical} \\
    \midrule
    loops &  &  \textcolor{green}{\ding{51}} \\
    loop-simplify &  & \textcolor{green}{\ding{51}} \\ lcssa &  &  \textcolor{green}{\ding{51}} \\ loop-rotate & \textcolor{green}{\ding{51}} &  \\ 
    licm &  &  \textcolor{green}{\ding{51}} \\ loop-load-elim &  \textcolor{red}{\Lightning}{} &  \\ indvars &  & \textcolor{green}{\ding{51}} \\ loop-idiom &  &  \textcolor{green}{\ding{51}} \\ loop-deletion & \textcolor{red}{\Lightning}{} &  \\ scalar-evolution &  & \textcolor{green}{\ding{51}} \\ \bottomrule
  \end{tabular}

}

\newcommand*{\tablethree}[0]{
  \footnotesize
  \begin{tabular}{@{}lcc@{}}
    \toprule
    \textbf{\makecell{LLVM \\Analysis \\Passes}} & \textbf{Theoretical} & \textbf{Empirical} \\
    \midrule
    tbaa & \textcolor{green}{\ding{51}} &  \\
    basicaa & \textcolor{green}{\ding{51}} &  \\
    aa &  \textcolor{green}{\ding{51}} &  \\
    loop-accesses & \textcolor{green}{\ding{51}} &  \\
    verify &  \textcolor{green}{\ding{51}} &  \\
    lazy-value-info & \textcolor{green}{\ding{51}} &  \\
    \bottomrule
  \end{tabular}
}

\newcommand*{\tablefour}[0]{
  \footnotesize
  \begin{tabular}{@{}lcc@{}}
    \toprule
    \textbf{\makecell{Expression\\ Substitution\\ Passes}} & \textbf{Theoretical} & \textbf{Empirical} \\
    \midrule
    constprop & \textcolor{green}{\ding{51}} &  \\
    sccp & \textcolor{green}{\ding{51}} &  \\ correlated-propagation & \textcolor{green}{\ding{51}} &  \\ constmerge & \textcolor{green}{\ding{51}} &  \\ \bottomrule
  \end{tabular}
}

\newcommand*{\tablefive}[0]{
  \footnotesize
  \begin{tabular}{@{}lcc@{}}
    \toprule
    \textbf{\makecell{Control Flow \\Simplification\\ Passes}} & \textbf{Theoretical} & \textbf{Empirical} \\
    \midrule
    simplifycfg & \textcolor{red}{\Lightning}{} &  \\
    loop-simplifycfg &  & \textcolor{green}{\ding{51}} \\ jump-threading &  & \textcolor{green}{\ding{51}} \\ sink & \textcolor{red}{\Lightning}{} &   \\ \bottomrule
  \end{tabular}
}

\newcommand*{\tablesix}[0]{
  \footnotesize
  \begin{tabular}{@{}lcc@{}}
    \toprule
    \textbf{\makecell{Memory and \\Stack Optimizations \\ passes}} & \textbf{Theoretical} & \textbf{Empirical} \\
    \midrule
    mem2reg & \textcolor{green}{\ding{51}} &  \\ \bottomrule
  \end{tabular}
}

\newcommand*{\tableseven}[0]{
  \footnotesize
  \begin{tabular}{@{}lcc@{}}
    \toprule
    \textbf{RetDec Utility Passes} & \textbf{Theoretical} & \textbf{Empirical} \\
    \midrule
    retdec-provider-init &  & \textcolor{green}{\ding{51}} \\
    retdec-decoder &  & \textcolor{green}{\ding{51}} \\
    retdec-write-ll &  & \textcolor{green}{\ding{51}} \\
    retdec-write-bc &  & \textcolor{green}{\ding{51}} \\
    retdec-llvmir2hll &  & \textcolor{green}{\ding{51}} \\
    retdec-x86-addr-spaces &  & \textcolor{green}{\ding{51}}  \\
    retdec-x87-fpu &  &  \textcolor{green}{\ding{51}} \\
    retdec-syscalls &  &  \textcolor{green}{\ding{51}} \\
    retdec-simple-types &  & \textcolor{green}{\ding{51}} \\
    retdec-param-return &  & \textcolor{green}{\ding{51}} \\
    retdec-select-fncs &  &  \textcolor{green}{\ding{51}} \\
    retdec-class-hierarchy &  & \textcolor{green}{\ding{51}} \\
    retdec-unreachable-funcs & & \textcolor{green}{\ding{51}} \\
    retdec-main-detection &  & \textcolor{green}{\ding{51}} \\
    retdec-idioms &  & \textcolor{green}{\ding{51}} \\
    retdec-idioms-libgcc &  & \textcolor{green}{\ding{51}} \\
    retdec-remove-phi &  & \textcolor{green}{\ding{51}} \\
    retdec-inst-opt &  &  \textcolor{green}{\ding{51}} \\
    retdec-inst-opt-rda &  & \textcolor{red}{\Lightning}{} \\
    retdec-value-protect &  & \textcolor{green}{\ding{51}} \\
    retdec-write-dsm &  &  \textcolor{green}{\ding{51}} \\
    retdec-cond-branch-opt &  & \textcolor{green}{\ding{51}} \\ retdec-constants &  &  \textcolor{green}{\ding{51}} \\ retdec-stack &  &  \textcolor{green}{\ding{51}} \\
    retdec-stack-ptr-op-remove &  &  \textcolor{green}{\ding{51}} \\
    retdec-register-localization &  &  \textcolor{green}{\ding{51}} \\
    retdec-value-protect &  &  \textcolor{green}{\ding{51}} \\
    \bottomrule
  \end{tabular}
}

\section{The Passes of RetDec}
\label{appendix:retdec-passes}

\begin{tabular}{cc}
    \tableone{}
  &
    \tabletwo{}
  \\
  & \\
    \tableseven{}
  &
    \begin{tabular}{c}
      \tablethree{}
      \\
      \\
      \tablefour{}
      \\
      \\
      \tablefive{}
      \\
      \\
      \tablesix{}
    \end{tabular}
\end{tabular}

\section{BunnyHop PoC Attack}\label{appendix:binsec-attack}
To
demonstrate the feasibility of the BunnyHop attack on
\Cref{lst:binsec_empty_branch}, we adapt the artifact
from~\cite{ZhangTOCGY23} and use our example as the victim.  We
conduct experiments running our victim on an Intel Core i7-10710U CPU,
and run the attack 2000 times.  Half of the time we set the secret bit
\texttt{rdi=1} and the other half, we set \texttt{rdi=0}.  Every time,
the attack runs the victim 100 times following the \textsc{Prime+Probe}
approach: before running the victim, we prime the branch target
buffer; then, after executing the victim, we probe the branch target
buffer to check whether the primed branches were evicted (resulting in
probe misses).  This indicates that the victim took the
secret-dependent branch.  The results of our experiments are shown in
\Cref{fig:binsec_exploit}, plotting the distribution of the probing
miss-rate, i.e., how many of the 100 probes missed, throughout our
attacks performed.  As we can see, the distributions for
\texttt{rdi=1} (the branch in the victim program is taken) leads to
significantly more probe misses than when \texttt{rdi=0} (the branch
is not taken), effectively leaking the value of the secret
\texttt{rdi} register.

\begin{figure}
  \centering
  \includegraphics[width=0.9\linewidth]{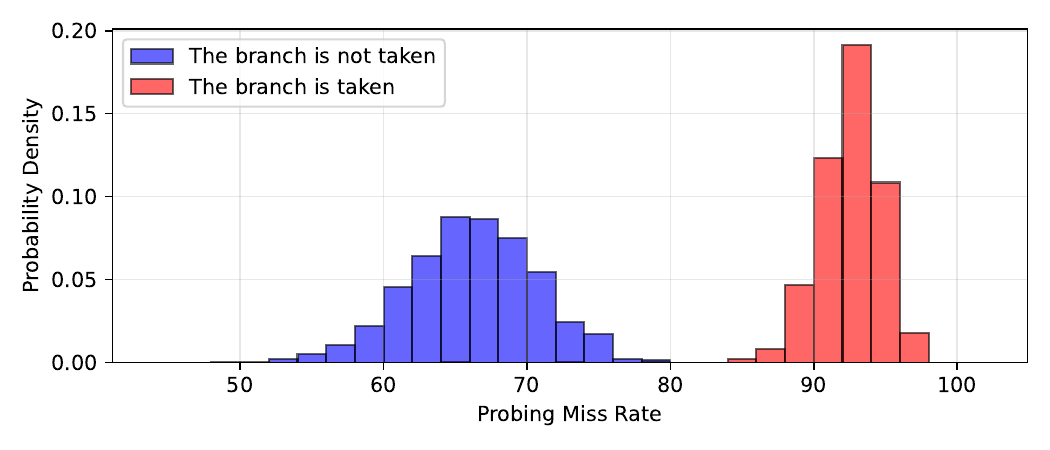}
  \caption{Attacking the synthetic example presented in 
  \Cref{lst:binsec_empty_branch} with the BunnyHop technique~\cite{ZhangTOCGY23}.
  The distribution of probe misses is distinguishable and can tell if 
  a secret-dependent branch was taken.}\label{fig:binsec_exploit}
\end{figure}

\section{Speculative Constant-Time Transparency}\label{appendix:sct-bisim}

In this section, we present a method to prove
that a program transformation is \ensuremath{{\text{SCT}}}{} transparent,
and apply this method to transformation from
\cref{sec:prog-optimizations} or provide counterexamples.

\subsection{Proof Techniques}\label{appendix:sct:techniques}

To prove \ensuremath{{\text{SCT}}}{} tranparency we aim to establish a bisimulation that
links source and target directives in a one-to-one manner:
We require that $\ensuremath{{T_{\ensuremath{{\mathcal{D}}}}^{\ensuremath{{\mathsf{pc}}}}}}$ is a bijective function that links the
directives between source states $s$ at $\ensuremath{{\ensuremath{{\mathsf{pc}}}({s})}}  = \ensuremath{{\mathsf{pc}}}$ and
target states $t$ when $s \sim t$.
We require that the directive transformer only inspects the program point
of source state in order to map the directives.
The notion $\ensuremath{{d_s \mathrel{\ensuremath{{T_{\ensuremath{{\mathcal{D}}}}^{\ensuremath{{\mathsf{pc}}}}}}} d_t}}$ states that the bijective $\ensuremath{{T_{\ensuremath{{\mathcal{D}}}}^{\ensuremath{{\mathsf{pc}}}}}}$
maps $d_s$ to $d_t$ and its inverse maps $d_t$ to $d_s$.

To allow for non-lockstep transformations, we split the program points of
the source program into two partitions $\ensuremath{{\mathcal{PC}}}_{\mathsf{sim}}$ and $\ensuremath{{\mathcal{PC}}}_{\mathsf{int}}$.
Program points belonging to $\ensuremath{{\mathsf{pc}}} \in \ensuremath{{\mathcal{PC}}}_{\mathsf{sim}}$ are those where
$\ensuremath{{T_{\ensuremath{{\mathcal{D}}}}^{\ensuremath{{\mathsf{pc}}}}}}$ links directives of source and target states so they step
forward together. Program points from $\ensuremath{{\mathsf{pc}}} \in \ensuremath{{\mathcal{PC}}}_{\mathsf{int}}$, however, are
intermediate steps, where only the source program takes a step to catch up
with the target program. The intermediate directives $\ensuremath{{d_{\ensuremath{{\mathsf{pc}}}}}}$ and
observations $\ensuremath{{o_{\ensuremath{{\mathsf{pc}}}}}}$ also depend only on $\ensuremath{{\mathsf{pc}}}$.

\begin{definition}[Speculative bisimulation diagram]\label{def:sct:diagram}
  A relation \ensuremath{{{\sim} \subseteq \ensuremath{{\mathcal{S}}}_s \times \ensuremath{{\mathcal{S}}}_t}}
  satisfies a \emph{simulation diagram} w.r.t.\
  a family of bijective directive transformers
  $\ensuremath{{T_{\ensuremath{{\mathcal{D}}}}^{\ensuremath{{\mathsf{pc}}}}}}: \ensuremath{{\mathcal{D}}}_t \leftrightarrow \ensuremath{{\mathcal{D}}}_s$,
  an observation transformer
  $\Tobsname : \ensuremath{{\mathcal{O}}}_s \to \ensuremath{{\mathcal{O}}}_t$, and
  intermediate directives $\ensuremath{{d_{\ensuremath{{\mathsf{pc}}}}}}$ and observations \ensuremath{{o_{\ensuremath{{\mathsf{pc}}}}}},
  if for every pair of related states \ensuremath{{s \sim t}} with $\ensuremath{{\mathsf{pc}}} = \ensuremath{{\ensuremath{{\mathsf{pc}}}({s})}}$ the following hold:
  \begin{requirelist}
      \item if $\ensuremath{{\mathsf{pc}}} \in \ensuremath{{\mathcal{PC}}}_{\mathsf{sim}}$ then for every $\ensuremath{{d_s \mathrel{\ensuremath{{T_{\ensuremath{{\mathcal{D}}}}^{\ensuremath{{\mathsf{pc}}}}}}} d_t}}$
          the input program steps $\sem{s}{d_s}{o_s}{s'}$
          if and only if
          the output program steps $\sem{t}{d_t}{o_t}{t'}$,
          and if so, \ensuremath{{s' \sim t'}} and $o_t = \Tobs{}{o_s}$; and
        \item if $\ensuremath{{\mathsf{pc}}} \in \ensuremath{{\mathcal{PC}}}_{\mathsf{int}}$, then the only step the input program can
          perform is $\sem{s}{\ensuremath{{d_{\ensuremath{{\mathsf{pc}}}}}}}{\ensuremath{{o_{\ensuremath{{\mathsf{pc}}}}}}}{s'}$ where $s' \sim t$.
  \end{requirelist}
\end{definition}

The following result recasts \cref{thm:ct:soundness} in the speculative
setting.
\begin{theorem}[Soundness of speculative simulations]\label{thm:sct:soundness}
    A program transformation
    ${\ensuremath{{\llparenthesis\,{\cdot}\,\rrparenthesis}}} : \ensuremath{{\mathcal{L}}}_s \to \ensuremath{{\mathcal{L}}}_t$
    is \ensuremath{{\phi\text{-SCT}}}{} transparent if for every input program $P$
    there exist a relation $\sim$,
    directive transformers $\ensuremath{{T_{\ensuremath{{\mathcal{D}}}}^{\ensuremath{{\mathsf{pc}}}}}}$,
    an observation transformer $\Tobsname{}$,
    and intermediate directives $\ensuremath{{d_{\ensuremath{{\mathsf{pc}}}}}}$ observations $\ensuremath{{o_{\ensuremath{{\mathsf{pc}}}}}}$,
    such that
    \begin{requirelist}
    \item $\sim$ is a speculative bisimulation w.r.t.\ \ensuremath{{T_{\ensuremath{{\mathcal{D}}}}^{\ensuremath{{\mathsf{pc}}}}}},
        \Tobsname{}, \ensuremath{{d_{\ensuremath{{\mathsf{pc}}}}}}, and \ensuremath{{o_{\ensuremath{{\mathsf{pc}}}}}};\label{thm:sct:soundness:simulation}
    \item initial states are related: \ensuremath{{P(i) \sim Q(i)}} for every $i$; and\label{thm:sct:soundness:initial}
    \item \Tobsname{} is \ensuremath{{\ensuremath{{\mathcal{PC}}}\text{-injective}}}{}.\label{thm:sct:soundness:injectivity}
    \end{requirelist}
\end{theorem}
\begin{proof}
    We prove only reflection.
    For a proof of preservation, \Cref{thm:sct:soundness:injectivity} is not
    needed and the directive transformer needs not be bijective
    (see one of \cite{WallM25,jasminPOPL25}).
    Let us assume that \(Q\) is \ensuremath{{\phi\text{-SCT}}}{}, we show that $P$ is \ensuremath{{\phi\text{-SCT}}}{} as well.
    Take two related inputs $i_1 \mathrel{\mathphi} i_2$
    and any two executions $\sem*{P(i_1)}{\ensuremath{{\boldsymbol{d}}}_s}{\ensuremath{{\boldsymbol{o}}}_1}{s_1}$,
    $\sem*{P(i_2)}{\ensuremath{{\boldsymbol{d}}}_s}{\ensuremath{{\boldsymbol{o}}}_2}{s_2}$ on the same sequence of directives $\ensuremath{{\boldsymbol{d}}}_s$.
    Notice that any states $t_1$, $t_2$ reachable from $Q(i_1)$, $Q(i_2)$
    under the same directives $\ensuremath{{\boldsymbol{d}}}_t$ will produce the same observations when
    executed. We denote this property with $t_1 =_\ensuremath{{\mathcal{O}}} t_2$.

    We will show a generalization of our theorem: for any directives
    $\ensuremath{{\boldsymbol{d}}}_s$, observations $\ensuremath{{\boldsymbol{o}}}_1$, $\ensuremath{{\boldsymbol{o}}}_2$,
    input program states $s_1$, $s_2$ and $s_1'$, $s_2'$,
    and output program states $t_1$, $t_2$, if
    $\ensuremath{{\ensuremath{{\mathsf{pc}}}({s_1})}} = \ensuremath{{\mathsf{pc}}} = \ensuremath{{\ensuremath{{\mathsf{pc}}}({s_2})}}$,
    $s_k \sim t_k$,
    $\sem*{s_k}{\ensuremath{{\boldsymbol{d}}}_s}{\ensuremath{{\boldsymbol{o}}}_k}{s_k'}$, and
    $t_1 =_\ensuremath{{\mathcal{O}}} t_2$,
    then $\ensuremath{{\boldsymbol{o}}}_1 = \ensuremath{{\boldsymbol{o}}}_2$.

    We proceed by induction on $\ensuremath{{\boldsymbol{d}}}_s$.
    The base case is trivial since $\ensuremath{{\boldsymbol{o}}}_1 = \ensuremath{{\epsilon}} = \ensuremath{{\boldsymbol{o}}}_2$.
    In the inductive case, we have that
    $\ensuremath{{\ensuremath{{\mathsf{pc}}}({s_1})}} = \ensuremath{{\ensuremath{{\mathsf{pc}}}({s_2})}}$,
    $s_k \sim t_k$,
    $s_k \step{d_s}{o_k} s_k' \step*{\ensuremath{{\boldsymbol{d}}}_s}{\ensuremath{{\boldsymbol{o}}}_k} s_k''$, and
    $t_1 =_\ensuremath{{\mathcal{O}}} t_2$.
    We show that $o_1 = o_2$ and use the inductive hypothesis for $\ensuremath{{\boldsymbol{o}}}_1 = \ensuremath{{\boldsymbol{o}}}_2$.
    Applying the diagram on the first step of the input program,
    we get one of two cases:

    \paragraph{\textbf{Case:} {$\ensuremath{{\mathsf{pc}}} \in \ensuremath{{\mathcal{PC}}}_{\mathsf{int}}$}}
        We have $o_1 = \ensuremath{{o_{\ensuremath{{\mathsf{pc}}}}}} = o_2$ and $s_k' \sim t_k$.

    \paragraph{\textbf{Case:} {$\ensuremath{{\mathsf{pc}}} \in \ensuremath{{\mathcal{PC}}}_{\mathsf{sim}}$}}
        In this case we have $\sem{s_k}{d_s}{o_k'}{s_k'}$ with $s_k' \sim t'$, and
        $\ensuremath{{d_s \mathrel{\ensuremath{{T_{\ensuremath{{\mathcal{D}}}}^{\ensuremath{{\mathsf{pc}}}}}}} d_t}}$.
We can apply the diagram to get $\sem{t_k}{d_t}{o_{t,k}}{t_k'}$, $s_k' \sim t_k'$,
        $o_{t,k} = \Tobs{}{o_k'}$.
As $t_1 =_\ensuremath{{\mathcal{O}}} t_2$, we know that $o_{t,1} = o_{t,2}$.
Since \Tobsname{} is \ensuremath{{\ensuremath{{\mathcal{PC}}}\text{-injective}}}{}, we get $o_1 = o_2$.

    Due to leakage of $\ensuremath{{\mathsf{pc}}}$, we have that $\ensuremath{{\ensuremath{{\mathsf{pc}}}({s_1'})}} = \ensuremath{{\ensuremath{{\mathsf{pc}}}({s_2'})}}$ and $\ensuremath{{\ensuremath{{\mathsf{pc}}}({t_1'})}} = \ensuremath{{\ensuremath{{\mathsf{pc}}}({t_2'})}}$.
    Further, $t_1' =_{\ensuremath{{\mathcal{O}}}} t_2'$.
    We can now apply the induction hypothesis to obtain that $\ensuremath{{\boldsymbol{o}}}_1 = \ensuremath{{\boldsymbol{o}}}_2$.
\end{proof}

\subsection{Transparent Transformations}\label{appendix:sct:passes}

For each transformation, we define a simulation $\sim$,
the corresponding sets $\ensuremath{{\mathcal{PC}}}_{\mathsf{sim}}$ and $\ensuremath{{\mathcal{PC}}}_{\mathsf{int}}$,
the observation transformer $\Tobsname : \ensuremath{{\mathcal{O}}}_s \to \ensuremath{{\mathcal{O}}}_t$,
the directive transformers $\ensuremath{{T_{\ensuremath{{\mathcal{D}}}}^{\ensuremath{{\mathsf{pc}}}}}}$ for each $\ensuremath{{\mathsf{pc}}} \in \ensuremath{{\mathcal{PC}}}_{\mathsf{sim}}$,
and the intermediate directives $\ensuremath{{d_{\ensuremath{{\mathsf{pc}}}}}}$ and observations $\ensuremath{{o_{\ensuremath{{\mathsf{pc}}}}}}$ for
intermediate program points $\ensuremath{{\mathsf{pc}}} \in \ensuremath{{\mathcal{PC}}}_{\mathsf{int}}$,
so that $\sim$ forms a bisimulation according to \Cref{def:sct:diagram} and
\Cref{thm:sct:soundness}.

\paragraph*{Structural Analysis}
The simulation $\sim$ and the injective observation transformer $\Tobsname$ are
the same as in \Cref{sec:passes}.
The set $\ensuremath{{\mathcal{PC}}}_{\mathsf{int}}$ is empty, so every program point belongs to $\ensuremath{{\mathcal{PC}}}_{\mathsf{sim}}$.
The directive transformer is the identity $\ensuremath{{T_{\ensuremath{{\mathcal{D}}}}^{\ensuremath{{\mathsf{pc}}}}}} = \mathit{id}$ for every $\ensuremath{{\mathsf{pc}}}$.
Since $\ensuremath{{\mathcal{PC}}}_{\mathsf{int}}$ is empty, $\ensuremath{{d_{\ensuremath{{\mathsf{pc}}}}}}$ and $\ensuremath{{o_{\ensuremath{{\mathsf{pc}}}}}}$ need not be defined.

\paragraph*{Constant Folding}
Let $\ensuremath{{\mathcal{T}}}$ be the transformation on the code.
The simulation $\sim$ requires equal memory and register contents,
\[
  \ensuremath{{\langle \ensuremath{{\mathit{c}}}, \rho, \mu \rangle}} \sim \ensuremath{{\langle \ensuremath{{\mathit{c}}}', \rho, \mu \rangle}} \quad
  \textit{if and only if} \quad \ensuremath{{\ensuremath{{\mathit{c}}}\mathrel{\ensuremath{{\mathcal{T}}}}\ensuremath{{\mathit{c}}}'}}\,.
\]
Again, the set $\ensuremath{{\mathcal{PC}}}_{\mathsf{int}}$ is empty, so every program point belongs to $\ensuremath{{\mathcal{PC}}}_{\mathsf{sim}}$.
Also the instructions did not change between $\ensuremath{{\mathit{c}}}$ and $\ensuremath{{\mathit{c}}}'$,
so the directive transformer is once more the identity $\ensuremath{{T_{\ensuremath{{\mathcal{D}}}}^{\ensuremath{{\mathsf{pc}}}}}} = \mathit{id}$ for every $\ensuremath{{\mathsf{pc}}}$.
Again, since $\ensuremath{{\mathcal{PC}}}_{\mathsf{int}}$ is empty, $\ensuremath{{d_{\ensuremath{{\mathsf{pc}}}}}}$ and $\ensuremath{{o_{\ensuremath{{\mathsf{pc}}}}}}$ need not be defined.

\subsection{Nontransparent Transformations}\label{appendix:sct:counterexamples}
Most of the transformations we discussed in \cref{sec:passes} are not
SCT transparent.
Since passes that are not CT transparent are immediately not SCT
transparent either, we focus on counterexamples for SCT transparency in
three of the passes we proved CT transparent in \cref{sec:passes}.

\paragraph*{Dead Branch Elimination}

\begin{figure}
  \centering
  \begin{jasmincode}[outerpos=t,outerwidth=35ex]
    \jasminindent{0}\jasminkw{if} (\jasminconstant{false}) \jasminopenbrace{} x = [secret]; \jasminclosebrace{}\\
    \jasminindent{0}\jasminkw{else} \jasminopenbrace{} x = \jasminconstant{42}; \jasminclosebrace{}
  \end{jasmincode}\hspace{2em}\begin{jasmincode}[outerpos=t,outerwidth=12ex]
    \jasminindent{0}\\
    \jasminindent{0}x = \jasminconstant{42};
  \end{jasmincode}
  \caption{Dead Branch Elimination breaks SCT transparency.}
  \label{lst:dead_branch}
\end{figure}

This transformation is not SCT transparent because dead branches may
execute speculatively.
This situation occurs whenever the compiler can prove a branch is dead,
for instance the \(n\)th iteration in a loop of \(n - 1\) iterations.
In \cref{lst:dead_branch}, the conditional branch at Line~1 is always
false, and, therefore, it is functionally correct to remove it.
Unfortunately, the input program exhibits a speculative constant-time
violation, as the branch in Line~1 may be executed under misspeculation.
Inside the branch, using \texttt{secret} as an address leaks its value,
even if only speculatively.
Thus, this pass is not SCT transparent, since its output program no
longer contains the SCT violation.

\paragraph*{Dead Assignment Elimination}

\begin{figure}
  \centering
  \begin{jasmincode}[outerpos=t,outerwidth=18ex]
    \jasminindent{0}x = secret;\\
    \jasminindent{0}\jasminkw{if} (\jasminconstant{false}) \jasminopenbrace{}\\
    \jasminindent{1}[x] = \jasminconstant{0};\\
    \jasminindent{0}\jasminclosebrace{}\\
    \jasminindent{0}x = \jasminconstant{0};
  \end{jasmincode}\hspace{2em}\begin{jasmincode}[outerpos=t,outerwidth=18ex]
    \jasminindent{0}\\
    \jasminindent{0}\jasminkw{if} (\jasminconstant{false}) \jasminopenbrace{}\\
    \jasminindent{1}[x] = \jasminconstant{0};\\
    \jasminindent{0}\jasminclosebrace{}\\
    \jasminindent{0}x = \jasminconstant{0};
  \end{jasmincode}
  \caption{Dead Assignment Elimination breaks SCT transparency.}
  \label{lst:sct:dead-assignment}
\end{figure}

In \cref{lst:sct:dead-assignment}, the transformation infers that the
assignment in Line~1 is unnecessary, because the conditional branch is
never taken.
Under speculative execution, however, the input program on the left
violates SCT by leaking the value of \texttt{secret},
because Line~3 can be executed speculatively.
This violation no longer exists after the transformation, which means
that Dead Assignment Elimination does not reflect \ensuremath{{\phi\text{-SCT}}}{}.

\paragraph*{Unspilling}

\begin{figure}
  \centering
  \begin{jasmincode}[outerpos=t,outerwidth=36ex]
    \jasminindent{0}[sp + \jasminconstant{4}] = secret; \jasmincomment{// spill}\\
    \jasminindent{0}\jasminkw{if} (i < \jasminconstant{5}) \jasminopenbrace{}\\
    \jasminindent{1}x = [p + i];\\
    \jasminindent{0}\jasminclosebrace{}\\
    \jasminindent{0}[x] = \jasminconstant{0};\\
    \jasminindent{0}secret = [sp + \jasminconstant{4}]; \jasmincomment{// unspill}
  \end{jasmincode}\hspace{2em}\begin{jasmincode}[outerpos=t,outerwidth=20ex]
    \jasminindent{0}tmp = secret;\\
    \jasminindent{0}\jasminkw{if} (i < \jasminconstant{5}) \jasminopenbrace{}\\
    \jasminindent{1}x = [p + i];\\
    \jasminindent{0}\jasminclosebrace{}\\
    \jasminindent{0}[x] = \jasminconstant{0};\\
    \jasminindent{0}secret = tmp;
  \end{jasmincode}
  \caption{Unspilling breaks SCT transparency.}
  \label{lst:sct:unspilling}
\end{figure}

Unspilling does not reflect SCT because speculatively executed
out-of-bounds loads may load spilled secrets.
\Cref{lst:sct:unspilling} presents an example of Unspilling.
On the left side, the value of \texttt{secret} is temporarily stored on
the stack at \texttt{sp+4}.
The memory access in Line~3 may speculatively read out-of-bounds and
load the value from \texttt{secret} at \texttt{sp+4}.
In Line~5, the program leaks the loaded value, which is an SCT
violation, and, in Line~6, \texttt{secret} is loaded back to a register.
On the right side of \cref{lst:sct:unspilling} we have the result of
applying the Unspilling transformation to this program.
It introduces a temporary variable \texttt{tmp}, and replaces the load
and store of \texttt{secret} with copies to \texttt{tmp}.
The output program no longer contains the SCT violation above: since
\texttt{secret} is not on the stack, the out-of-bounds access can not
load it.

\end{document}